\documentclass[acmsmall,screen,nonacm]{acmart}

\usepackage[utf8]{inputenc} %
\usepackage[T1]{fontenc}    %
\usepackage{anyfontsize}
\usepackage{hyperref}       %
\usepackage{url}            %
\usepackage{booktabs}       %
\usepackage{amsfonts}       %
\usepackage{nicefrac}       %
\usepackage{microtype}      %
\usepackage{xcolor}         %
\usepackage{amsmath}
\usepackage{mathtools}
\usepackage{amsthm}
\usepackage{stackengine}
\usepackage{lipsum}
\usepackage{dsfont}
\usepackage{graphicx}
\usepackage{tabularx}
\usepackage{float}
\usepackage{mathtools}
\usepackage{mathrsfs}
\usepackage{stmaryrd}
\usepackage{enumitem}
\usepackage{algorithm}
\usepackage{xspace}
\usepackage[capitalize,noabbrev]{cleveref}
\usepackage{multirow}
\usepackage{graphicx}
\usepackage[export]{adjustbox}
\usepackage[framemethod=TikZ]{mdframed}
\usepackage[noend]{algpseudocode}
\usepackage[most]{tcolorbox}
\usepackage{titletoc}
\usepackage{natbib}

\newtheorem{theorem}{Theorem}[section]
\newtheorem{proposition}{Proposition}[section]
\newtheorem{lemma}[theorem]{Lemma}

\newtheorem{definition}[theorem]{Definition}

\newtheorem{corollary}[theorem]{Corollary}

\newtheorem{fact}[theorem]{Fact}

\renewenvironment{quote}
  {\list{}{\rightmargin=0.3cm \leftmargin=0.3cm}%
   \item\relax}
  {\endlist}

\newcommand{\cA}{\mathcal{A}}
\newcommand{\cB}{\mathcal{B}}

\newcommand{\cF}{\mathcal{F}}

\newcommand{\cI}{\mathcal{I}}

\newcommand{\cM}{\mathcal{M}}

\newcommand{\cS}{\mathcal{S}}

\newcommand{\cX}{\mathcal{X}}
\newcommand{\cY}{\mathcal{Y}}
\newcommand{\cZ}{\mathcal{Z}}

\newcommand{\N}{\mathbb{N}}
\newcommand{\Q}{\mathbb{Q}}
\newcommand{\R}{\mathbb{R}}

\newcommand{\Z}{\mathbb{Z}}

\newcommand{\range}[1]{\mathbb{N}_{< {#1}}}
\newcommand{\zeroton}[1]{\mathbb{N}_{\le {#1}}}
\newcommand\clamp[3]{\text{clamp} \bracket{ {#1}, {#2}, {#3} } }

\newcommand{\tvdistance}[2]{\mathds{TV} ( {#1} ,\, {#2} )}
\newcommand{\fdivergence}[2]{\mathds{D}_f \left( {#1} ,\, {#2} \right)}
\newcommand{\norm}[1]{\lVert {#1} \rVert}

\newcommand{\PAREN}[1]{{\left( {#1} \right)}}
\newcommand{\Biggparen}[1]{{\Bigg( {#1} \Bigg)}}
\newcommand{\Bigparen}[1]{{\Big( {#1} \Big)}}
\newcommand{\bigparen}[1]{{\big( {#1} \big)}}
\newcommand{\paren}[1]{{( {#1} )}}
\newcommand{\floor}[1]{{\lfloor {#1} \rfloor}}
\newcommand{\ceil}[1]{{\lceil {#1} \rceil}}

\newcommand{\bracket}[1]{{[ {#1} ]}}

\newcommand{\BRACKET}[1]{{\left[ {#1} \right]}}
\newcommand{\CARD}[1]{\left| {#1} \right|}
\newcommand{\card}[1]{\lvert {#1} \rvert}
\newcommand{\SET}[1]{\left\{ {#1} \right\}}
\newcommand{\set}[1]{\{ {#1} \}}
\newcommand{\eps}{\varepsilon}

\DeclareMathOperator*{\argmin}{arg\,min}

\NewDocumentCommand\p{ m g }{
  \ensuremath{
    \IfNoValueTF{#2}
    {\Pr [ #1 ]}
    {\Pr_{#1}[#2]}
  }
}
\RenewDocumentCommand\P{ m g }{
  \ensuremath{
    \IfNoValueTF{#2}
    {\Pr \left[#1\right]}
    {\Pr_{#1}\left[#2\right]}
  }
}

\DeclareMathOperator*{\Expectation}{\mathbb{E}}
\NewDocumentCommand\E{ m g }{
  \ensuremath{
    \IfNoValueTF{#2}
    {\Expectation \left[#1\right]}
    {\Expectation_{#1}\left[#2\right]}
  }
}
\DeclareMathOperator*{\Variance}{\mathbb{V}\mathrm{ar}}
\NewDocumentCommand\Var{ m g }{
  \ensuremath{
    \IfNoValueTF{#2}
    {\Variance \left[#1\right]}
    {\Variance_{#1}\left[#2\right]}
  }
}

\NewDocumentCommand\GammaDist{ m g }{
  \ensuremath{
    \IfNoValueTF{#2}
    {\operatorname{\mathcal{G}amma}\left( {#1} \right)}
    {\operatorname{\mathcal{G}amma}\left( {#1}, {#2} \right)}
  }
}

\NewDocumentCommand\LapNoise{ m g }{
  \ensuremath{
    \IfNoValueTF{#2}
    {\operatorname{\mathbb{L}ap}\left( {#1} \right)}
    {\operatorname{\mathbb{L}ap}\left( {#1}, {#2} \right)}
  }
}

\NewDocumentCommand\GumbelNoise{ m g }{
  \ensuremath{
    \IfNoValueTF{#2}
    {\operatorname{\mathbb{G}umbel}\left( {#1} \right)}
    {\operatorname{\mathbb{G}umbel}\left( {#1}, {#2} \right)}
  }
}

\newcommand{\DataSet}{\mathcal{X}}

\newcommand{\element}{i}
\newcommand{\user}{u}
\newcommand{\UserSet}{\mathcal{U}}
\newcommand{\UserData}[1]{x^{({#1})}}
\newcommand{\NeighborUserData}[1]{x'^{\,({#1})}}
\newcommand{\SortedUserData}[1]{y^{({#1})}}
\newcommand{\DataDomain}{\mathcal{D}}
\newcommand{\IntSet}[2]{[{#1}\,.\,.\,{#2}]}
\newcommand{\Binomial}[2]{\operatorname{\mathbb{B}in} ( {#1} , {#2} )}
\newcommand{\Bernoulli}[1]{\operatorname{\mathbb{B}ernoulli}\left( {#1} \right)}

\newcommand{\indicator}[1]{\mathds{1}_{\left[#1\right]}}

\newcommand{\SortCircuitCost}[2]{\text{CostSort}\paren{#1,#2}}
\newcommand{\MulCost}[1]{\text{CostMul}\paren{#1}}

\NewDocumentCommand\PurifiedApproxDiscreteLaplaceMechanism{ m m m g }{
  \ensuremath{
    \IfNoValueTF{#4}
    {
        \cM_{
            \operatorname{\mathbb{PA}px\mathbb{DL}ap} \paren{#1, #2, #3}
        }
    }
    {
        \cM_{
            \operatorname{\mathbb{PA}px\mathbb{DL}ap} \paren{#1, #2, #3}
        }
        \left( {#4} \right)
    }
  }
}

\newcommand{\DLapSampler}[2]{\cS_{\operatorname{DLap}} \paren{#1, #2}}

\NewDocumentCommand\algoPurify{ g }{
  \ensuremath{
    \IfNoValueTF{#1}
    {
        \cM_{\operatorname{purify}} 
    }
    {
        \cM_{\operatorname{purify}} \paren{#1}
    }
  }
}

\newcommand{\CoreSupportSet}[1]{\operatorname{C}_{#1}}

\newcommand{\LabelOracle}{\mathcal{O}_{\text{label}}}
\newcommand{\ProbOracle}{\mathcal{O}_{\text{bin}}}
\newcommand{\poly}[1]{\textsc{poly} ( {#1} ) }

\newcommand{\dlcenter}{q_{\text{center}}}

\newcommand{\BinaryExp}[2]{\mathbb{B}inary\left(#1, #2\right)}

\newcommand{\runningtime}[1]{T_{{#1}}}

\newcommand{\ourAlgoName}{\text{Private Tail Padding Histogram}\xspace}

\newcommand{\ourAlgo}{\mathcal{M}_{\textsc{\scriptsize Hist}}}

\NewDocumentCommand\algoBernoulli{g}{
  \ensuremath{
    \IfNoValueTF{#1}
    {{\mathcal{M}_{\mathbb{B}ern} } }
    {{\mathcal{M}_{\mathbb{B}ern} \left( {#1} \right)} }
  }
}

\NewDocumentCommand\algoGeo{g}{
  \ensuremath{
    \IfNoValueTF{#1}
    {{\mathcal{M}_{\mathbb{G}eo} } }
    {{\mathcal{M}_{\mathbb{G}eo} \left( {#1} \right)} }
  }
}

\NewDocumentCommand\algoDLap{g}{
  \ensuremath{
    \IfNoValueTF{#1}
    {{\mathcal{M}_{\mathbb{DL}ap} } }
    {{\mathcal{M}_{\mathbb{DL}ap} \left( {#1} \right)} }
  }
}

\NewDocumentCommand\algoBatchDLap{g}{
  \ensuremath{
    \IfNoValueTF{#1}
    {{\mathcal{M}_{\mathbb{BDL}ap} } }
    {{\mathcal{M}_{\mathbb{BDL}ap} \left( {#1} \right)} }
  }
}

\NewDocumentCommand\algoSBatchDLap{g}{
  \ensuremath{
    \IfNoValueTF{#1}
    {{\mathcal{M}_{\mathbb{S}\text{-}\mathbb{BDL}ap} } }
    {{\mathcal{M}_{\mathbb{S}\text{-}\mathbb{BDL}ap} \left( {#1} \right)} }
  }
}

\NewDocumentCommand\algoSBatchDLapPerturb{g}{
  \ensuremath{
    \IfNoValueTF{#1}
    {{\mathcal{M}_{\mathbb{S}\text{-}\mathbb{BDL}ap}.\textsc{\small Perturb} } }
    {{\mathcal{M}_{\mathbb{S}\text{-}\mathbb{BDL}ap}.\textsc{\small Perturb} \left( {#1} \right)} }
  }
}

\NewDocumentCommand\algoAlias{g}{
  \ensuremath{
    \IfNoValueTF{#1}
    {{\mathcal{M}_{\mathbb{A}lias} } }
    {{\mathcal{M}_{\mathbb{A}lias} \left( {#1} \right)} }
  }
}

\NewDocumentCommand\algoFiniteAlias{g}{
  \ensuremath{
    \IfNoValueTF{#1}
    {{\mathcal{M}_{\mathbb{FA}lias} } }
    {{\mathcal{M}_{\mathbb{FA}lias} \left( {#1} \right)} }
  }
}

\NewDocumentCommand\algoAliasInitialization{g}{
  \ensuremath{
    \IfNoValueTF{#1}
    {{\mathcal{M}_{\mathbb{A}lias}.\textsc{\small Initialization}}}
    {{\mathcal{M}_{\mathbb{A}lias}.\textsc{\small Initialization} \left( {#1} \right)}}
  }
}

\NewDocumentCommand\algoFiniteAliasInitialization{g}{
  \ensuremath{
    \IfNoValueTF{#1}
    {{\mathcal{M}_{\mathbb{FA}lias}.\textsc{\small Initialization}}}
    {{\mathcal{M}_{\mathbb{FA}lias}.\textsc{\small Initialization} \left( {#1} \right)}}
  }
}

\newcommand{\distSupport}{\mu}
\newcommand{\eventSelectIstart}{\mathcal{E}}
\newcommand{\eventSelectJstart}{\mathcal{F}}
\newcommand{\ConditionalSupportRatio}{\kappa}

\newcommand{\sample}{{\textsc{\small Sample}}}

\NewDocumentCommand\algoAliasSample{g}{
  \ensuremath{
    \IfNoValueTF{#1}
    {
        {\mathcal{M}_{\mathbb{A}lias}.\textsc{\small Sample}}
    }
    {
        {\mathcal{M}_{\mathbb{A}lias}.\textsc{\small Sample} \left( {#1} \right)}
    }
  }
}

\DeclareMathOperator\supp{\textsc{supp}}

\NewDocumentCommand\hist{ g }{
  \ensuremath{
    \IfNoValueTF{#1}
    {\vec{h}}
    {\vec{h}\BRACKET{#1}}
  }
}

\NewDocumentCommand\noisyhist{ g }{
  \ensuremath{
    \IfNoValueTF{#1}
    {\Tilde{h}}
    {\Tilde{h}\BRACKET{#1}}
  }
}

\NewDocumentCommand\intermediateHist{ g }{
  \ensuremath{
    \IfNoValueTF{#1}
    {\hat{h}}
    {\hat{h}\BRACKET{#1}}
  }
}

\NewDocumentCommand\PoissonNoise{ m g }{
  \ensuremath{
    \IfNoValueTF{#2}
    {\operatorname{\mathbb{P}oisson}\left( {#1} \right)}
    {\operatorname{\mathbb{P}oisson}\left( {#1}, {#2} \right)}
  }
}

\NewDocumentCommand\DiscreteLapNoise{ g g }{
  \ensuremath{
    \IfNoValueTF{#2}
    {
        \IfNoValueTF{#1}
        {\operatorname{\mathbb{DL}ap}}
        {\operatorname{\mathbb{DL}ap}\left( {#1} \right)}
    }
    {\operatorname{\mathbb{DL}ap}\left( {#1}, {#2} \right)}
  }
}

\NewDocumentCommand\TDiscreteLapNoise{ m g }{
  \ensuremath{
    \IfNoValueTF{#2}
    {\operatorname{\mathbb{TDL}ap}\left( {#1} \right)}
    {\operatorname{\mathbb{TDL}ap}\left( {#1}, {#2} \right)}
  }
}
\NewDocumentCommand\ExpNoise{ m g }{
  \ensuremath{
    \IfNoValueTF{#2}
    {\operatorname{\mathbb{E}xp}\left( {#1} \right)}
    {\operatorname{\mathbb{E}xp}\left( {#1}, {#2} \right)}
  }
}

\NewDocumentCommand\GeometricNoise{ m g }{
  \ensuremath{
    \IfNoValueTF{#2}
    {\operatorname{\mathbb{G}eo}\left( {#1} \right)}
    {\operatorname{\mathbb{G}eo}\left( {#1}, {#2} \right)}
  }
}

\newcommand{\ApproxDiscreteLaplaceNoise}[4]{{\operatorname{\mathbb{A}px\mathbb{DL}ap}}_{#2, #3, #4}\left( {#1} \right)}

\NewDocumentCommand\UniformNoise{ m g }{
  \ensuremath{
    \IfNoValueTF{#2}
    {\operatorname{\mathbb{U}niform}\left( {#1} \right)}
    {\operatorname{\mathbb{U}niform}\left( {#1}, {#2} \right)}
  }
}

\definecolor{LightBlue}{rgb}{0.31,0.52,0.78}

\newif\ifcomment
\commenttrue

\definecolor{DarkGreen}{rgb}{0.1,0.5,0.1}
\ifcomment
\newcommand{\hao}[1]{\textcolor{blue}{[HAO: #1]}}

\else

\makeatletter
\newcommand{\hao}[1]{%
  \@bsphack
  \@esphack
}
\makeatother

\fi

\acmDOI{XXXXXXX.XXXXXXX}

\begin{document}

\title{Optimal Pure Differentially Private Sparse Histograms in Deterministic Linear Time}

\author{Florian Kerschbaum }
\email{florian.kerschbaum@uwaterloo.ca}
\affiliation{%
  \institution{Cheriton School of Computer Science, University of Waterloo}
  \country{Canada}
}

\author{Steven Lee}
\email{hj44lee@uwaterloo.ca}
\affiliation{%
  \institution{Cheriton School of Computer Science, University of Waterloo}
  \country{Canada}}

\author{Hao Wu}
\email{hao.wu1@uwaterloo.ca}
\affiliation{%
  \institution{Cheriton School of Computer Science, University of Waterloo}
  \country{Canada}
}

\begin{abstract}
    We present an algorithm that releases a pure differentially private (under the replacement neighboring relation) sparse histogram for $n$ participants over a domain of size $d \gg n$. 
    Our method achieves the optimal $\ell_\infty$-estimation error and runs in strictly $O(n)$ time in the Word-RAM model, improving upon the previous best deterministic-time bound of $\Tilde{O}(n^2)$ and resolving the open problem of breaking this quadratic barrier~(\citeauthor{BalcerV19}~\citeyear{BalcerV19}).  
    Moreover, the algorithm admits an efficient circuit implementation, enabling the first near-linear communication and computation cost pure DP histogram MPC protocol with optimal $\ell_\infty$-estimation error.  
    Central to our algorithm is a novel \emph{private item blanket} technique with target-length padding, 
    which hides differences in the supports of neighboring histograms while remaining efficiently implementable.
\end{abstract}



\maketitle

\section{Introduction}
\label{sec:introduction}

Differential privacy is a rigorous mathematical framework for protecting individual data in statistical analyses. 
It ensures that algorithms produce similar output distributions on neighboring datasets---those differing in a single individual's data---making it difficult to infer whether any particular individual is included in the input. 
This strong privacy guarantee has made differential privacy widely adopted in both theory and practice.

We focus on one of the most fundamental tasks in differential privacy: histogram publishing.
In this setting, there are $n$ participants, each holding an element drawn from the domain $[d] = \IntSet{1}{d}$. 
The histogram, $\hist \in \IntSet{0}{n}^d$, is a vector where $\hist[i]$ equals the number of times element $i$ appears among the participants. 
The goal is to publish a differentially private version $\noisyhist$ that well approximates $\hist$ while protecting individual participant's data.

One of the earliest and simplest solutions for privatizing a histogram $\hist$ is the Laplace mechanism \citep{DworkMNS06}, which adds independent, continuous Laplace noise to each entry of $\hist$.
This approach achieves various asymptotically optimal error guarantees \citep{HardtT10, BeimelBKN14}.
However, implementing the mechanism presents subtle challenges, as it assumes operations on real numbers and requires operating on all entries of $\hist$. 
The former can introduce potential privacy vulnerabilities, while the latter may incur prohibitive computational costs, which we elaborate on below.

\paragraph{Floating-Point Attack}
Sampling Laplace noise requires representing real numbers on discrete machines, typically via double-precision floating-point arithmetic. 
As early as \citeyear{Mironov12}, \citet{Mironov12} observed that due to finite precision and rounding artifacts, certain floating-point values cannot be generated.
These artifacts can be exploited to distinguish exactly between neighboring inputs based on the outputs of the Laplace mechanism, thereby completely compromising its privacy guarantees.
A fundamental mitigation is to avoid real-number arithmetic altogether by using discrete analogues of Laplace noise---such as the discrete Laplace distribution \citep{GhoshRS09}.

\paragraph{Timing Attack}
However, avoiding floating-point vulnerabilities introduces a new threat: the timing attack.
Since discrete Laplace noise has unbounded support, its sampler can require unbounded memory and has running time only bounded in expectation \citep{Canonne0S20}. 
\citet{JinMRO22} showed a positive correlation between the sampled value and the algorithm’s running time, allowing adversaries to infer the output from timing information---again violating privacy.
This motivates the design of \emph{time-oblivious} algorithms, whose running time leaks no---or only a limited amount of---information about the input.  
For instance, one possible approach is to enforce strict upper bounds on running time (ideally polynomial in the bit length of the input), allowing all executions to be padded to run for the same duration.

\paragraph{Exponential-Size Domains}

Adding noise to all entries of $\hist$ results in a running time proportional to $d$, which can be exponential in the input size. 
For example, consider the task of identifying popular URLs of up to 20 characters in length,\footnote{Valid URL characters include digits ($0$--$9$), letters (A--Z, a--z), and a few special characters ("-", ".", "\_", "\texttildelow").} as discussed in \citet{FPE16}. 
This setting corresponds to a domain size of at least $d \ge 10^{36}$.
To avoid adding noise to all entries, a natural idea is to publish a sparse histogram $\noisyhist$ to approximate the original histogram $\hist$, since the latter is also sparse: it has at most $n$ nonzero entries. 
Some early works along this line, either satisfy only approximate differential privacy \citep{KorolovaKMN09, BunNS19}---in contrast to the pure differential privacy guaranteed by the Laplace mechanism \footnote{The distinction between pure and approximate differential privacy will be formally explained in \cref{sec:problem-description}}---or require sampling from complicated binomial distributions \citep{CormodePST12}, for which it remains open whether one can sample efficiently while avoiding timing attacks \citep{BalcerV19}.

\begin{quote}
   \vspace{1mm}
   \textit{\underline{Research Problem}: Design an efficient algorithm for publishing a pure differentially private histogram that avoids both floating-point and timing attacks.}
\end{quote}

\paragraph{The Quadratic Time Barrier.}
\citet{BalcerV19} initiated a systematic study of differentially private sparse histograms and proposed a sparsified variant of the Laplace mechanism suited for Word-RAM model. 
They designed a deterministic-time sampler for a relaxed discrete Laplace distribution supported on a finite domain, as well as an algorithm that reports only the top-$n$ elements with the highest noisy counts, without explicitly adding noise to all entries in $\hist$.
Their algorithm achieves asymptotically optimal error among sparse histograms. 
However, the noise sampler has a deterministic per-sample cost of $\tilde{O} \PAREN{1 / \eps}$, and identifying the top-$n$ elements requires $\tilde{O}(n^2)$ (deterministic) time. 
The paper concludes with the following open question:

\begin{quote}
    \vspace{1mm}
    \textit{\underline{Research Question}: 
        Can one design deterministic and linear-time pure differentially private histogram algorithms in the Word-RAM model while preserving privacy and accuracy guarantees?
    }
\end{quote}

\subsection{Our Contributions}
Our contributions are twofold. 
First, we provide a positive answer to the open question posed by \citet{BalcerV19}. 
Second, as a byproduct, we develop a general framework of deterministic-time approximate sampler for arbitrary discrete distributions, and instantiate it specifically for generating discrete Laplace noise.
We discuss both results separately.

\subsubsection*{\bf Histogram Construction.}
Consistent with \citep{BalcerV19}, our construction builds on Word-RAM model, and the privacy guarantee is respect to \emph{replacement neighboring relation}, both formally defined in \cref{sec:problem-description}.

\begin{table}[t!]
    \centering
    \begin{tabular}{lcccc}
    \specialrule{0.9pt}{0pt}{0pt} %
    Algorithm & $\ell_\infty$-Error & Model & Type & Running Time \\
    \specialrule{0.8pt}{0pt}{0pt} %
    Laplace Mechanism~\citep{Dwork06} & $\Theta \paren{1/\eps \cdot \ln d}$ & Real-RAM & Deterministic & $O(d)$ \\
    Geometric Mechanism~\citep{GhoshRS09} & $\Theta \paren{1/\eps \cdot \ln d}$ & Word-RAM & Randomized & Expected $O(d)$ \\
    Filter Mechanism~\citep{CormodePST12, DBLP:journals/corr/QiuY25}  & $\Theta \paren{1/\eps \cdot \ln d}$ & Real-RAM & Randomized & Expected $O(n)$ \\
    PureSparseHistogram~\citep{BalcerV19} & $\Theta \paren{1/\eps \cdot \ln d}$ & Word-RAM & Deterministic & $\tilde{O}(n^2)$ \\
    {\color{LightBlue!60!blue} Theorem~\ref{thm:private-sparse-histogram}} & {\color{LightBlue!60!blue} $\Theta \paren{1/\eps \cdot \ln d}$} & {\color{LightBlue!60!blue} Word-RAM} & {\color{LightBlue!60!blue} Deterministic} & {\color{LightBlue!60!blue} $O(n)$} \\
    \specialrule{0.9pt}{0pt}{0pt} %
    \end{tabular}
    \vspace{1mm}
    \caption{Comparison of sparse histogram algorithms: error, model, type, and running time. 
    The Real-RAM model assumes unit-cost arithmetic on infinite-precision floating-point numbers, while the Word-RAM model assumes unit-cost arithmetic on standard word-sized values. 
    }
    \vspace{-6mm}
    \label{tab:sparse-histograms}
\end{table}

\begin{theorem}[Private Sparse Histogram, Informal version of \cref{thm:private-sparse-histogram-formal}]
    \label{thm:private-sparse-histogram}
    Let $n, d \in \N_+$ and $\eps \in \Q_+$, with representations that fit in a constant number of machine words.
    Given a dataset $\DataSet$ of $n$ user-contributed elements from domain $[d]$, 
    there exists an $\eps$-differentially private algorithm that runs deterministically in $O(n)$ time
    and outputs a histogram $\noisyhist \in \zeroton{n}^d$ approximating the original histogram $\hist$ of $\DataSet$, 
    such that $\norm{\noisyhist}_0 \in O(n)$ and 
    $\E{\norm{\hist - \noisyhist}_\infty} \in O\!\left(\frac{1}{\eps} \cdot \ln d\right)$.
\end{theorem}

The expected error bound above is asymptotically optimal, matching the lower bounds established in \citep{HardtT10, BeimelBKN14, BalcerV19}.
Table~\ref{tab:sparse-histograms} compares our algorithm wiht the previous ones.
Beyond its applications to private data mining and machine learning, our algorithm also helps \emph{close a longstanding utility gap between the central and distributed differential privacy models for histograms.}

\subsubsection*{\bf \small From Central to Distributed DP: Matching Utility.}
Our algorithm operates in the \emph{central} DP model, where a trusted server collects participants' data to release a private histogram.
A commonly studied alternative is the \emph{distributed} DP (DDP) model, where the server is untrusted.
The simplest DDP setting is the \emph{local} model, where each participant perturbs their own data before reporting it to the server.
Unfortunately, this added privacy comes at a steep cost: local protocols incur substantial noise, and the best known frequency-estimation protocol achieves an expected $\ell_\infty$ error of $\Theta(\tfrac{1}{\eps} \cdot \sqrt{n \ln d})$~\citep{BS15, BNS19, DBLP:conf/aistats/WuW22}, creating a tight $\Omega(\sqrt{n})$ utility gap compared to the central model.

A significant body of work has attempted to bridge this gap by enabling participants (and, in some settings, multiple servers) to jointly and securely compute a function such that no party learns anything beyond their own input and the final DP output.
Initial efforts focused on the \emph{shuffle model}~\citep{BittauEMMRLRKTS17,BalleBGN19,CheuSUZZ19,GhaziG0PV21}, where the joint computation is restricted to randomly permuting locally perturbed messages.
Despite these advances, a gap remains: even the best near–linear-time shuffle protocols~\citep{GhaziG0PV21} still incur a polynomial (in $\log d$) degradation in utility compared to the optimal bounds achieved in the central model.

Recent works have further bridged this gap by exploiting the full potential of the DDP model, 
supporting richer classes of functionalities beyond shuffling via secure multiparty computation (MPC), 
as demonstrated by systems from Mozilla (Prio)~\cite{Corrigan-GibbsB17}, Google~\cite{GG0M0S22}, and Microsoft~\cite{AndersonCDLW24}. 
This line of work culminated in a recent $(\epsilon, \delta)$-DP sparse histogram algorithm~\citep{GG0M0S22} that achieves optimal $\ell_\infty$ error, on par with the best central-model guarantees.
\emph{However, a Pure $\eps$-DP counterpart with optimal error remains elusive.}
The primary obstacle has been the absence of efficient, time-oblivious Pure DP algorithms, even within the central model. 
Such algorithms are necessary in the MPC setting because data-dependent execution time or branching patterns can introduce side channels that compromise the privacy of the secret participant inputs.
Our algorithm overcomes this barrier: it admits an efficient circuit implementation, enabling secure execution in the MPC setting using standard techniques~\citep{BGW88,GMW87}.

\begin{theorem}[Private Sparse Histogram Circuit, Informal version of \cref{thm:purified-approximate-discrete-laplace-sampler-formal}]
    Let $n, d \in \N_+$ and $\eps \in \Q_+$, with representations that fit in a constant number of machine words.
    There exists a circuit with $\tilde{O}(n \omega)$ gates (where $\omega$ is the machine word size) implementing an $\eps$-DP mechanism that takes 
    a dataset $\DataSet$ of $n$ user-contributed elements from domain $[d]$ as input, and outputs a histogram 
    $\noisyhist \in \zeroton{n}^d$ approximating the original histogram $\hist$ of $\DataSet$, 
    such that $\norm{\noisyhist}_0 \in O(n)$ and 
    $\E{\norm{\hist - \noisyhist}_\infty} \in O \paren{ \frac{1}{\eps} \cdot \ln d }$.
\end{theorem}

The number of gates scales with $n \omega$ rather than $n$ alone because, in the circuit model, merely reading $n$ input elements from $[d]$ (with each element described using $\log d \in O(\omega)$ bits) requires $O(n \omega)$ gates; hence, $n \omega$ is considered linear in the input size. 
Additionally, simulating Random Access Memory (RAM) operations of our algorithm in the circuit model introduces a logarithmic overhead in the gate complexity.

\subsubsection*{\bf Noisy Generation.}
We develop a general framework for constructing deterministic-time approximate samplers for arbitrary discrete distributions. 
This framework is of independent interest and applies broadly to differentially private noise-generation tasks.
In addition, it strengthens the result of \citet*{DovDNT23} by achieving comparable running time and accuracy guarantees while using asymptotically less space.

\begin{table}[t!]
    \centering
    \begin{tabular}{lccc}
    \specialrule{0.9pt}{0pt}{0pt}
    Sampler & Space Usage (words) & Running Time / Sample & TV Distance \\
    \specialrule{0.8pt}{0pt}{0pt}
    \addlinespace[0.5mm] %
    
    \citet{DovDNT23} &
    $O\!\left( 
        \frac{1}{\omega} \cdot
        \card{\CoreSupportSet{\delta/2}} \cdot \frac{1}{\delta} \cdot
        \ln \card{\CoreSupportSet{\delta/2}} 
    \right)$ &
    $\frac{1}{\omega} \cdot \paren{ \log \card{\CoreSupportSet{\delta/2}} 
       + \log \frac{1}{\delta}} + O(1) $ &
    $\le \delta$ \\

    \addlinespace[1mm] %
    
    {\color{LightBlue!60!blue} Theorem~\ref{thm:time-oblivious-distribution-sampler-informal}} &
    \color{LightBlue!60!blue} $O\!\left( 
        \frac{1}{\omega} \cdot
        \card{\CoreSupportSet{\delta/2}} \cdot
        \bigl( \ln \frac{1}{\delta} 
               + \ln \card{\CoreSupportSet{\delta/2}} \bigr)
    \right)$ &
    \color{LightBlue!60!blue} $\frac{1}{\omega} \cdot \paren{ \log \card{\CoreSupportSet{\delta/2}} 
       + \log \frac{1}{\delta}} + O(1) $ &
    \color{LightBlue!60!blue} $\le \delta$ \\

    \addlinespace[0.5mm] %
    
    \specialrule{0.9pt}{0pt}{0pt}
    \end{tabular}
    \caption{
        Comparison of deterministic approximate samplers for a discrete distribution $\mu$.
        Here, $\CoreSupportSet{\delta/2}$ denotes a smallest subset satisfying 
        $\mu(\CoreSupportSet{\delta/2}) \ge 1 - \delta/2$, and $\omega$ is the word size.
        Our sampler improves space complexity by reducing its dependence on $\delta$
        from $1/\delta$ to $\ln(1/\delta)$.
    }    
    \vspace{-6mm}
    \label{tab:noise-sampler-comparison}
\end{table}

\begin{theorem}[Deterministic Approximate Sampler, Informal~\cref{thm:time-oblivious-distribution-sampler}]
\label{thm:time-oblivious-distribution-sampler-informal}
Let $\mu$ be a distribution over a discrete domain, let $\delta \in (0,1)$, and let 
$\CoreSupportSet{\delta/2}$ denote a smallest subset satisfying 
$\mu(\CoreSupportSet{\delta/2}) \ge 1 - \delta/2$.
Then there exists a deterministic-time sampler for $\mu$ with the following guarantees:
\begin{itemize}[label=$\triangleright$, leftmargin=0.7cm]
    \item \textbf{Memory:} 
    $O\!\left( 
        \frac{1}{\omega} \,
        \card{\CoreSupportSet{\delta/2}} \,
        \bigl( \ln \tfrac{1}{\delta} + \ln \card{\CoreSupportSet{\delta/2}} \bigr)
    \right)$,
    with $\omega$ the Word-RAM word size.
    \item \textbf{Sampling time per draw:} 
    $\frac{1}{\omega} \, \bigparen{ \log \card{\CoreSupportSet{\delta/2}} 
    + \log \tfrac{1}{\delta} } + O(1) $.
    \item \textbf{Accuracy:} The resulting distribution has total variation distance at most $\delta$ from $\mu$.
\end{itemize}
\end{theorem}

\noindent
\textit{Comparison to Prior Work.} 
For reference, \citet*{DovDNT23} (Claim~2.22) achieve the same guarantees for sampling time and accuracy (Items~2 and~3 above).  
However, their sampler requires space 
\(
O\!\left( 
    \omega^{-1} \card{\CoreSupportSet{\delta/2}} \,\delta^{-1} 
    \ln \card{\CoreSupportSet{\delta/2}} 
\right),
\)
which is exponentially worse in its dependence on $\delta$ compared to our construction.  
Table~\ref{tab:noise-sampler-comparison} summarizes the comparison.

We now instantiate our general framework to generate a variant of discrete Laplace noise tailored for perturbing elements in the bounded range $\IntSet{0}{n}$. 
This instantiation incorporates two key optimizations: 
(1) we adopt a technique from \citet{BalcerV19} (see \cref{sec:preliminaries}) to mix the sampled noise with a uniform distribution over $\IntSet{0}{n}$, which ensures that the truncated output preserves pure differential privacy; 
(2) we exploit the memoryless property of the discrete Laplace distribution to further reduce the space required for sampling.  
The formal guarantees of this sampler are presented below and serve as a crucial building block for our linear-time private histogram protocol.

\begin{theorem}[Relaxed Discrete Laplace, Informal version of \cref{thm:purified-approximate-discrete-laplace-sampler}]
\label{thm:purified-approximate-discrete-laplace-sampler-informal}
    Let $n \in \N_+$, $\eps \in \Q_+$ and $\gamma \in \Q_+ \cap (0, 1)$, with representations that fit in a constant number of machine words.
    There exists a randomized algorithm $\PurifiedApproxDiscreteLaplaceMechanism{n}{\eps}{\gamma} : \zeroton{n} \rightarrow \zeroton{n}$ with initialization time $\tilde{O}(\frac{1}{\eps})$, memory usage 
    $O\bigl( \frac{1}{\eps} + \ln \frac{1}{\gamma} + \ln n \bigr)$,
    deterministic per-query running time $O(1)$ after initialization.
    It satisfies the following privacy and utility guarantees:
    \begin{itemize}[label=$\triangleright$, leftmargin=0.7cm]
        \item For each $t \in [n]$, 
        $\P{ \PurifiedApproxDiscreteLaplaceMechanism{n}{\eps}{\gamma}{t - 1} = i} \big/ \P{ \PurifiedApproxDiscreteLaplaceMechanism{n}{\eps}{\gamma}{t} = i} \in [e^{-\eps}, e^{\eps}]$.
        \item For each $t \in \zeroton{n}$, $\beta > 2 \cdot \gamma$, 
        $\P{ \card{\PurifiedApproxDiscreteLaplaceMechanism{n}{\eps}{\gamma}{t} - t} \ge \alpha} \le \beta$
        for
        $\alpha \in O \paren{ \frac{1}{\eps} \cdot \ln \frac{1}{\beta} }$.
    \end{itemize}
\end{theorem}

\subsubsection*{\bf Technical Overview.}  
We provide a high-level overview of our private sparse histogram algorithm and noise sampler. 
More detailed overviews appear in the respective sections where each algorithm is formally presented and analyzed.

\subsubsection*{Private Histogram.}
The algorithm has a remarkably simple structure with non-trivial analysis.  
It first adds a relaxed variant of discrete Laplace noise to the non-zero elements of $\hist$, and selects those whose noisy counts exceed a specified threshold.  
For a neighboring histogram $\hist'$ with different support, this step might lead to a different set of selected elements.

To hide this difference, the algorithm pads the selected set with uniformly random elements (without replacement) until its size reaches $n$.  
These random elements serve as a \emph{privacy blanket}, and are drawn from all elements not selected in the first step---including both zero-count and possibly some non-zero elements.   
This sidesteps costly procedures such as the binomial-sampling step of \citet{CormodePST12}, and the identification of top-$n$ noisy counts among zero elements in $\hist$, as used by \citet{BalcerV19}.
Finally, for all selected elements, fresh noises are regenerated before releasing the output.

As we detail in \cref{sec:dp-sparse-histogram}, the utility guarantee follows directly from the properties of the relaxed discrete Laplace noise, and the running time is inherited from our sampler's efficiency (Theorem~\ref{thm:purified-approximate-discrete-laplace-sampler-informal}).  
The privacy analysis, however, is the most delicate part.  
It requires proving that the distribution over the set of selected elements remains stable between neighboring histograms, which involves a careful case-by-case analysis.

\subsubsection*{Noise Sampler.}
Our general framework for approximately sampling from arbitrary discrete distributions builds directly on the classic Alias method~\citep{Walker77}, applied to carefully chosen truncated supports and truncated probability representations to achieve a prescribed total variation distance. 
Despite its simplicity, this yields a deterministic-time sampler and improves upon the recent time-oblivious sampling scheme of~\citet*{DovDNT23}, making it worthwhile to document in this context.

Next, we instantiate the framework for the discrete Laplace distribution. 
Specifically, we reduce sampling from the discrete Laplace to sampling geometric noise, exploiting the fact that the discrete Laplace is almost a two-sided geometric distribution. 
We further leverage the property that geometric noise has identical conditional distributions over intervals of equal length, allowing decomposition into two smaller geometric components. 
Finally, following a technique by \citet{BalcerV19} (see \cref{sec:preliminaries}), we mix the sampled noise with a uniform distribution over $\IntSet{0}{n}$ to ensure that the truncated output preserves pure differential privacy.

\subsubsection*{\bf Organization.}
The remainder of the paper is organized as follows.
\Cref{sec:problem-description} formally introduces the problem.
\Cref{sec:preliminaries} reviews the relevant probability distributions and a key building block used in our algorithms.
Due to space constraints, we present only our private sparse histogram algorithm in \cref{sec:dp-sparse-histogram} and defer the description of our noise-sampling procedures to \cref{sec:private-noise-samplers}.
This reordering does not affect the exposition, as \cref{sec:dp-sparse-histogram} is self-contained and can be read independently of the sampler implementation details.
The circuit realization of our sparse histogram algorithm appears in \cref{subsec:mpc-circuit-sparse-hist}.
Finally, \cref{sec:related work} provides a comprehensive overview of related work.

\section{Problem Description}
\label{sec:problem-description}

\subsubsection*{Notation.}
Let $\N$ denote the set of natural numbers, $\Z$ the set of integers, $\Q$ the set of rational numbers, and $\R$ the set of real numbers.  
We use $\N_+$, $\Q_+$, and $\R_+$ to denote the sets of positive integers, positive rationals, and positive reals, respectively.
For any $n \in \N$, we write $\range{n}$ for the set $\set{0, 1, \dots, n - 1 }$, $\zeroton{n}$ for the set $\set{0, 1, \dots, n}$, and $[n]$ for the set $\set{1, \ldots, n}$.

\subsubsection*{Computational Model.}
We design our algorithms in the standard Word-RAM model, where each machine word consists of $\omega$ bits. 
In this model, basic operations on $\omega$-bit words—such as arithmetic, comparisons, bitwise manipulations, and memory access (read/write)—are assumed to take constant time.
We also assume that sampling a uniformly random integer from $\range{2^\omega}$ takes constant time.
Rational inputs are represented as pairs of integers.
We measure running time by counting the number of word-level operations, so the possible values of running time lie in $\N$.

\subsection{Differentially Private Histogram}

Let $\DataDomain \doteq \set{1, \ldots, d}$ be a set of $d$ elements and $\UserSet \doteq \set{1, \ldots, n}$ be a set of $n$ participants.  
Each participant $\user \in \UserSet$ holds an element $\UserData{\user} \in \DataDomain$.  
The dataset $\DataSet \doteq \set{\UserData{1}, \ldots, \UserData{n}}$ consists of all participant elements.  
The \emph{frequency} of an element $\element \in \DataDomain$, denoted by  
$
    \hist{\element} \doteq \set{
        \user \in \UserSet : \UserData{\user} = \element
    },
$
is defined as the number of participants holding the element $\element$.  
The \emph{histogram}, denoted by $\hist \doteq \paren{\hist{1}, \ldots, \hist{d}} \in \zeroton{n}^d$, is the frequency vector.  

The objective is to release a histogram $\noisyhist \in \N^d$ that accurately approximates $\hist$, while preserving the privacy of each individual participant's data, as formally defined below.

\subsubsection*{\bf Privacy Guarantee}
We require that the output $\noisyhist$ reveal only a limited amount of information about the input.
The privacy of the output is formally captured by the notion of differential privacy.

\subsubsection*{Neighboring Datasets.}
Consistent with \citet{BalcerV19}, we adopt the \emph{replacement neighboring} model%
\footnote{
    Another commonly used notion is the \emph{add/remove neighboring} model.
    \emph{However, under this model, no optimal-error private frequency histogram algorithm can run in deterministic time.}
    A formal definition of the model and further discussion appear in \cref{sec:add-remove-model}.
    \emph{This limitation is precisely why \citet{BalcerV19} restrict their study to the replacement model, where the question of breaking the quadratic running-time barrier was originally posed.}
}, 
where two datasets $\DataSet$ and $\DataSet'$ of equal size are \emph{neighboring}, denoted $\DataSet \sim \DataSet'$, if they differ in exactly one participant's data; that is, there exists a unique $\user \in [n]$ such that $\UserData{\user} \neq \NeighborUserData{\user}$ and $\UserData{i} = \NeighborUserData{i}$ for all $i \neq \user$.
Accordingly, two histograms $\hist, \hist' \in \N^d$ are called \emph{neighboring} if they are induced by neighboring datasets under one of the above definitions.

\subsubsection*{Private Algorithm.}
A histogram releasing algorithm is said to be \emph{differentially private} if its output distributions are similar on all pairs of neighboring inputs.
Formally:

\begin{definition}[$\paren{\eps, \delta}$-Indistinguishability]
    \label{def:indistinguishability}
    Let $\eps \in \R_+$, $\delta \in [0, 1]$, and let $(\cZ, \cF)$ be a measurable space. 
    Two probability measures $\mu$ and $\nu$ on this space are $\paren{\eps, \delta}$-indistinguishable if
    \begin{equation}
        e^{-\eps} \cdot \paren{ \mu(E) - \delta } 
        \le \nu(E) 
        \le e^{\eps} \cdot \mu(E) + \delta,
        \quad
        \forall E \in \cF.
    \end{equation}
    Similarly, two random variables are $\paren{\eps, \delta}$-indistinguishable if their corresponding distributions are.
\end{definition}

\begin{definition}[$\paren{\eps, \delta}$-Private Algorithm~\citep{DR14}] 
    \label{def: Differential Privacy}
    Given $\eps \in \R_+$, $\delta \in [0, 1]$, a randomized algorithm 
    $\cM: \DataDomain^n \rightarrow \N^d$
    is called~$\paren{\eps, \delta}$-differentially private (DP),
    if for every~$\DataSet, \DataSet' \in \cY$ such that~$\DataSet \sim \DataSet'$, 
    $\cM (\DataSet)$ and $\cM (\DataSet')$ are $\paren{\eps, \delta}$-indistinguishable.
\end{definition}

An algorithm~$\cM$ is said to be \emph{pure differentially private}, or simply \emph{$\eps$-DP}, if it satisfies $(\eps, 0)$-differential privacy. 
If instead $\cM$ satisfies $(\eps, \delta)$-differential privacy for some $\delta > 0$, it is referred to as \emph{approximate differentially private}.

Although we present differential privacy in the context of histogram release, the definition applies more generally to any randomized algorithm $\cM: \cY \rightarrow \cZ$, where $\cY$ is an arbitrary input domain endowed with a symmetric neighboring relation~$\sim$.

\subsubsection*{\bf Utility Guarantee.}
We evaluate utility using the $\ell_\infty$ error \citep{BalcerV19}.

\begin{definition}[$(\alpha, \beta)$-Simultaneous Accurate Estimator]
\label{def:simultaneous-accuracy}
    Let $\alpha \in \R_+$ and $\beta \in [0, 1]$.  
    A random vector $\vec{X} \doteq \paren{ X_1, \ldots, X_m }$ is an $(\alpha, \beta)$-simultaneously accurate estimator of a vector $\vec{t} \doteq \paren{ t_1, \ldots, t_m } \in \R^m$ if \,$\P{ \norm{ \vec{X} - \vec{t} }_\infty \ge \alpha } \le \beta$.
\end{definition}

In addition to $(\alpha,\beta)$-simultaneous guarantees for $\noisyhist$, we also provide bounds on the expected error $\mathbb{E} \big[\norm{\noisyhist - \hist}_\infty \big]$; the latter follows from the former using standard integration techniques (see \cref{sec:dp-sparse-histogram} for details).

\section{Preliminaries}
\label{sec:preliminaries}

In this section, we define several probability distributions used throughout the paper, and review a framework of converting approximate DP into pure DP algorithms. 

\subsection{Probability Distributions}

\begin{definition}[Geometric Distribution]
\label{def:geometric-distribution}
    Given $p \in [0, 1)$ and $u \in \N \cup \set{ \infty }$,  
    let $\GeometricNoise{p}{\zeroton{u}}$ be a random variable supported on $\zeroton{u}$ with probability mass function  
    \begin{equation*}
        \begin{array}{c}
            \P{ \GeometricNoise{p}{\zeroton{u}} = t } = \frac{1 - p}{1 - p^{1 + u}} \cdot p^t, \quad \forall t \in \zeroton{u}.
        \end{array}
    \end{equation*}
    When $u = \infty$, we write $\GeometricNoise{p}$ as shorthand for $\GeometricNoise{p}{\zeroton{\infty}}$.
\end{definition}

\begin{definition}[Discrete Laplace Distribution]
    \label{def:discrete-laplace}
    Given $p \in [0, 1)$, let $\DiscreteLapNoise{p}$ denote a random variable following the \emph{discrete Laplace distribution} with probability mass function:
    \begin{equation*}
        \P{ \DiscreteLapNoise{p} = t } = \frac{1 - p}{1 + p} \cdot p^{\card{t}}, \quad \forall t \in \Z.
    \end{equation*}
\end{definition}

\subsection{Purification}

\citet{BalcerV19} showed that if an algorithm $\cM' : \cY \to \cZ$ with finite output domain is close in total variation distance to some pure DP algorithm $\cM : \cY \to \cZ$, then one can convert $\cM'$ into a pure DP algorithm by mixing its output with the uniform distribution over $\cZ$.
The purification procedure is given in \cref{alg:purification}: with probability $1 - \gamma$, the algorithm outputs $\cM'(y)$, and otherwise outputs a uniformly random $Z \in \cZ$.
The formal guarantee is stated in \cref{proposition:purification}.

\begin{algorithm}[H]
    \caption{Purification $\algoPurify$}
    \label{alg:purification}
    \begin{algorithmic}[1]
        \Statex \hspace{-4.6mm} {\bf Input:} 
            An algorithm $\cM': \cY \rightarrow \cZ$, where $\card{\cZ}$ is finite; 
            input $y \in \cY$; 
            and $\gamma \in \Q_+ \cap (0, 1)$
        \Statex \hspace{-4.6mm} {\bf Output:} 
            $Z \in \cZ$
        \State $B \gets \Bernoulli{\gamma}$ \Comment{$\P{B = 1} = \gamma$}
        \State Output $\paren{1 - B} \cdot \cM'(y) + B \cdot Z$, where $Z \sim \UniformNoise{\cZ}$.
    \end{algorithmic}
\end{algorithm}

\begin{proposition}[\cite{BalcerV19}]
    \label{proposition:purification}
    Given an algorithm $\cM': \cY \rightarrow \cZ$, assume there exists an $\eps$-DP algorithm $\cM: \cY \rightarrow \cZ$ such that $\tvdistance{ \cM(y) }{ \cM'(y) } \le \delta$ for all $y \in \cY$, with parameter $\delta \in [0, 1)$.
    Then Algorithm~\ref{alg:purification} satisfies:
    \begin{itemize}[label=$\triangleright$, leftmargin=0.4cm]
        \item $\eps$-DP, provided that
        $\delta \le \frac{e^\eps - 1}{e^\eps + 1} \cdot \frac{\gamma}{1 - \gamma} \cdot \frac{1}{\card{\cZ}}$.
        
        \item Running time $O \PAREN{ \frac{1}{\omega} \cdot \log \frac{1}{\gamma} + \frac{1}{\omega} \cdot \log \card{\cZ}} + \runningtime{ \cM' }$, where $\runningtime{ \cM' }$ is the running time of $\cM'$.
    \end{itemize}
\end{proposition}

\section{Differentially Private Sparse Histograms}
\label{sec:dp-sparse-histogram}

In this section, we present our algorithm for releasing sparse histograms under pure differential privacy.  
The following theorem summarizes our main result.

\begin{theorem}[\ourAlgoName]
    \label{thm:private-sparse-histogram-formal}
    Given integers 
    $n, d, k, a_\eps, b_\eps, a_\gamma, b_\gamma \in \N_+$ 
    that fit in a constant number of machine words, 
    where $\eps \doteq \frac{a_\eps}{b_\eps} \in \Q_+$ and $\gamma \doteq a_\gamma / b_\gamma \in \paren{ 1 / n^{O(1)},  1 }$,
    and a dataset $\DataSet \doteq \set{\UserData{1}, \ldots, \UserData{n}} \in [d]^n$,
    there exists an algorithm, denoted by $\ourAlgo$, that outputs a histogram $\noisyhist \in \zeroton{n}^d$ approximating the original histogram $\hist$ of $\DataSet$, satisfying $\norm{\noisyhist}_0 \in O(n)$ and 
    \begin{itemize}[leftmargin=0.4cm, label=$\triangleright$]
        \item \textbf{Privacy Guarantee:} $\ourAlgo$ satisfies 
        $\max \Bigparen{
            \frac{3\eps}{2} + \frac{
                1
            }{
                k + 1
            },
            2\eps   
        }$-differential privacy.
        \item \textbf{Utility Guarantee:} There exists a universal constant $c_\alpha \in \R_+$, such that for each $\beta \ge 2 \eps \gamma$, $\alpha \doteq c_\alpha / \eps \cdot \ln \paren{d / \beta}$,
        the output $\noisyhist$ is an $(\alpha, \beta)$-simultaneous accurate estimator of $\hist$.

        \item \textbf{Running Time:} $\ourAlgo$ runs deterministically in $O(n + k)$ time.
    \end{itemize}
\end{theorem}

\subsubsection*{Remark.}
The utility guarantee of \cref{thm:private-sparse-histogram-formal} is an $\ell_\infty$ bound and is asymptotically optimal, matching the lower bound $\Omega \paren{ \min \set{ \tfrac{1}{\eps}\ln d , n } }$ of \citet{BalcerV19}. 
The bound provided here is a high-probability one, in contrast to the expected one stated in \cref{thm:private-sparse-histogram}.  
The latter can be recovered from the former by the standard identity $\mathbb{E}[X] = \int_0^\infty \Pr[X \ge t] \, dt$ for a nonnegative random variable $X$.

Setting $k = 2/\eps$ yields a privacy guarantee of $2\eps$-DP and running time $O(n + 1/\eps)$, which simplifies to $O(n)$ in the meaningful regime $n \ge 1/\eps$.  
When $n < 1/\eps$, simply returning the all-zero frequency vector incurs error at most $n < 1/\eps$, matching the lower bound.

Also, in what follows, we assume $d \ge 10n$ to focus on the sparse regime where our algorithm is most relevant.  
Otherwise, when $d < 10n$, one can simply release the entire histogram $\hist$ after perturbing each count using the algorithm in \cref{thm:purified-approximate-discrete-laplace-sampler-informal}, which yields an algorithm with $O(n)$ running time.

\subsubsection*{Roadmap.} 
The remainder of this section is organized as follows.  
We begin by discussing \cref{thm:private-sparse-histogram-formal} and its connection to the informal version stated in the introduction (\cref{thm:private-sparse-histogram}).  
\Cref{subsec:overview-of-private-histogram} presents the key ideas behind our algorithm, along with its pseudocode. 
\Cref{subsec:proof-private-sparse-histogram} provides the proof of \cref{thm:private-sparse-histogram-formal}. 

We present an MPC-friendly prototype circuit implementation of our algorithm in \cref{subsec:mpc-circuit-sparse-hist}.
A detailed comparison with prior work is deferred to \cref{sec:related work}.

\subsection{Overview of \texorpdfstring{$\ourAlgo$}{Private Histogram Algorithm}}
\label{subsec:overview-of-private-histogram}

To motivate our algorithm, we begin by reviewing the \textsc{Filter} algorithm for releasing pure-DP sparse histograms \citep{CormodePST12}, which lacks a known efficient time-oblivious implementation.
Given a histogram $\hist \in \zeroton{n}^d$ with $\norm{\hist}_1 = n$, the \textsc{Filter} algorithm produces an output $\noisyhist$ distributed identically to the output of the following procedure:
(1) Add independent $\DiscreteLapNoise{e^{-\eps}}$ noise to each entry of $\hist$.
(2) Retain only the entries whose noisy values exceed a threshold $\tau \in \Theta\left( \frac{1}{\eps} \cdot \ln d \right)$, chosen so that the expected number of outputs is $O(n)$.

Implementing the above procedure directly takes $O(d)$ time.
To avoid this,
the \textsc{Filter} algorithm treats non-zero and zero entries separately.
Non-zero entries are processed directly as described, taking $O(n)$ time.
For zero entries, \citet{CormodePST12} observed that after noise addition, each entry exceeds the threshold $\tau$ independently with the same probability, denoted by $q_\tau \in O(n / d)$.
Thus, the algorithm proceeds as follows:
\begin{enumerate}[leftmargin=0.6cm]
    \item[(1)] Sample the number $X$ of zero entries exceeding $\tau$ from a binomial distribution 
    $\operatorname{\mathbb{B}in} (d - \card{\supp(\hist)}, q_\tau)$
    where $\supp(\hist)$ denotes the support of $\hist$.
    \item[(2)] Sample a random subset of size $X$ from $[d] \setminus \supp(\hist)$.
    \item[(3)] 
    Assign i.i.d.\ noise to the $X$ entries from distribution $\DiscreteLapNoise{e^{-\eps}}$, conditioned on exceeding $\tau$.
\end{enumerate}

\subsubsection*{\bf Obstacles to a Time-Oblivious Implementation.}
Though \textsc{Filter} has a simple description, implementing it properly is not straightforward.
The first issue is that sampling from $\DiscreteLapNoise{e^{-\eps}}$ exactly is vulnerable to timing attacks.
To address this, one can replace $\DiscreteLapNoise{e^{-\eps}}$ with a noise distribution with similar properties that has deterministic sampling time, such as the one presented in \cref{thm:purified-approximate-discrete-laplace-sampler-informal}. 
In addition, one can avoid sampling from the conditional distribution in the last step by instead using fresh noises, at the slight cost of a factor-$2$ degradation in utility.

The real challenge lies in sampling from the distribution $\Binomial{d - \card{\supp(\hist)}}{q_\tau}$.
Even writing down the exact probabilities of this binomial distribution requires exponential space.
A natural idea is to approximately sample: generate a random variable $X$ such that $\tvdistance{X}{\Binomial{d - \card{\supp(\hist)}}{q_\tau}}$ $\le \delta$ for some small $\delta$.
This leads to an algorithm that satisfies approximate DP.
To recover pure DP, one could apply the purification technique (\cref{proposition:purification}), which mixes the final output distribution of \textsc{Filter} with the uniform distribution over the output space to smooth out the $\delta$ gap.
However, \cref{proposition:purification} requires $\delta \in O(1 / 2^d)$, since the output of \textsc{Filter} may be any subset of $[d]$.
For such a small $\delta$, we do not know how to efficiently sample $X$ satisfying $\tvdistance{X}{\Binomial{d - \card{\supp(\hist)}}{q_\tau}} \le \delta$.

\subsubsection*{\bf Rethinking Zero Entry Sampling for Privacy.}  
The challenge above motivates a modification to how \textsc{Filter} handles zero entries in $\hist$, in order to avoid sampling from a complex distribution.
We revisit the role that the selected zero entries play in ensuring privacy.
Consider two neighboring histograms $\hist$ and $\hist'$.
If they have the same support, then the selected zero entries have no impact on privacy.
However, if their supports differ (in one or two entries under the replacement neighboring model), the sampled zero entries help obscure this difference.
Informally, they serve as a \emph{privacy blanket} over the support set.
This protection can be decomposed into two parts:
\begin{enumerate}[leftmargin=0.6cm]
    \item[(1)] After adding the privacy blanket, the total number of selected entries (including both zero and nonzero entries) should have similar distributions for neighboring histograms.
    \item[(2)] Conditioned on the total number being the same, the distribution over the selected sets themselves should also be similar across neighbors.
\end{enumerate}

\begin{algorithm}[!t]
    \caption{\ourAlgoName $\ourAlgo$}
    \label{algo:tail-item-padding}
    \begin{algorithmic}[1]
        \Statex \hspace{-4.8mm} {\bf Input:} 
            parameters $d, n, k, a_\eps, b_\eps, a_\gamma, b_\gamma \in \N_+$, s.t., $\eps = a_\eps / b_\eps$ and $\gamma = a_\gamma / b_\gamma$
        \Statex \hspace{-4.8mm} \hspace*
        {\widthof{\textbf{Input:}}} 
        dataset $\DataSet \doteq \set{\UserData{1}, \ldots, \UserData{n}} \in [d]^n$ 
        \Statex \hspace{-4.8mm} {\bf Output:} 
            histogram $\noisyhist \in \zeroton{n}^d$
        \State construct histogram $\hist \in \zeroton{n}^d$ based on $\DataSet$
        \label{line:histogram-construction}
        \For{each $i \in \supp \paren{\hist}$}
            \vspace{0.5mm}
            \State $\intermediateHist{i} \gets \PurifiedApproxDiscreteLaplaceMechanism{n}{\eps}{\eps \gamma / d} \paren{\hist{i}}$
            \Comment{$\intermediateHist{i} \approx \clamp{\hist{i} + \DiscreteLapNoise{e^{-\eps}}}{0}{n}$}
        \EndFor
        \State $I_1 \gets \{ i \in \supp \paren{\hist} : \intermediateHist{i} \ge \tau \}$
        \Comment{$\tau \in \Theta \paren{\tfrac{\ln d}{\eps}}$ as defined in \cref{eq:def-tau}}
        \State $I_2 \gets$ a uniform random subset in $\binom{[d] \setminus I_1}{n + k - \card{I_1}}$
        \label{line:blanket-sampling}
        \State $I \gets I_1 \cup I_2$
        \label{line:item-padding}
        \For{each $i \in I$}
            \State $\noisyhist{i} \gets \PurifiedApproxDiscreteLaplaceMechanism{n}{\eps}{\eps \gamma / d} \paren{\hist{i}}$
            \label{line:re-genrating-noises}
            \Comment{$\noisyhist{i} \approx \clamp{\hist{i} + \DiscreteLapNoise{e^{-\eps}}}{0}{n}$}
        \EndFor
        \State \Return $\noisyhist$ \Comment{where $\noisyhist[i] = 0$ for all $i \notin I$}
    \end{algorithmic}
\end{algorithm}

\subsubsection*{\bf Simplifying Privacy Blanket.}
Our goal is to design a simpler privacy blanket that satisfies both requirements.
The pseudocode for our algorithm is provided in \cref{algo:tail-item-padding}.
The noise sampler $\PurifiedApproxDiscreteLaplaceMechanism{n}{\eps}{\eps \gamma / d}$ used in \cref{algo:tail-item-padding} is  described in \cref{thm:purified-approximate-discrete-laplace-sampler-informal}.
This mechanism produces noise with the same privacy guarantees and nearly identical tail bounds as discrete Laplace noise, but clamped to the range $[0, n]$.
For ease of exposition, \emph{it is helpful to conceptually treat these samples as clamped discrete Laplace noise.}

For the first requirement, we leverage the fact that, under the replacement neighboring model, the total count $n$ can be made public.
Denote by $I_1$ the set of selected entries from $\supp(\hist)$, obtained by adding noise to each entry and retaining those whose noisy counts exceed the threshold $\tau$.
We then pad $I_1$ to a fixed target size $n + k$, for some chosen $k \in O(n)$.
This completely avoids sampling the blanket size from the distribution $\Binomial{d - \card{\supp(\hist)}}{q_\tau}$.

For the second requirement, we sample a subset of $n + k - \card{I_1}$ random entries from the set $[d] \setminus I_1$ (without replacement).
This includes not only zero counts, but may also include non-zero ones; i.e., entries in $\supp(\hist)$ can have a second chance of being selected.
To see how this works, consider one case where $\supp(\hist) \cup \{ i^* \} = \supp(\hist')$ for some $i^* \in [d]$.
In this case, the count $\hist[i^*] = 0$ and $\hist'[i^*] = 1$.
Further, define $I_1'$ for $\hist'$ analogously to $I_1$ for $\hist$.
Conditioned on $i^* \notin I_1'$ (which happens in most cases), $I_1$ and $I_1'$ have identical distributions.
Padding them with sampled entries from the same distribution preserves the distributional identity of the final selected sets.

We only need to handle the case where $i^* \in I_1'$.
In this case, $I_1$ receives $n + k - \card{I_1}$ padding entries, whereas $I_1'$ receives $n + k - \card{I_1'} = n + k - \card{I_1} - 1$ padding entries.
The privacy blanket for $I_1$ is thus larger by one entry than that for $I_1'$. 
Conditioned on this additional entry being $i^*$, the remaining parts of the two privacy blankets are identically distributed.
Hence, it suffices to show that the probability of $i^*$ appearing in the privacy blanket is comparable to its probability of appearing in $I_1'$, which we will formally establish in \cref{subsec:proof-private-sparse-histogram}.

\subsubsection*{\bf Further Discussion and Open Questions on Padding Strategies.}
Up to this point, our discussion naturally raises several questions regarding the proposed padding strategy.
For instance, one might wonder what happens if we sample the $n + k - \card{I_1}$ padding entries \emph{with replacement} from the set $[d] \setminus I_1$, thereby allowing the final total number of selected entries to be potentially smaller than $n + k$.
Another natural question is whether the padding entries could be sampled (with or without replacement) from the larger set $[d] \setminus \supp(\hist)$ instead of $[d] \setminus I_1$.
Do these variations still guarantee pure differential privacy?  
Do they help avoid re-generating fresh noise for all selected elements in Line~\ref{line:re-genrating-noises} of \cref{algo:tail-item-padding}?

While these modifications might also be rigorously analyzed by suitably adapting the proof techniques presented here, our current analysis does not directly cover them.
Our primary goal in this work is to initiate the study of, and provide one workable padding strategy, rather than to offer an exhaustive treatment.
We therefore leave these questions open as interesting directions for future research.

\subsection{Proof of Theorem~\ref{thm:private-sparse-histogram-formal}}
\label{subsec:proof-private-sparse-histogram}

In this subsection, we begin the formal proof of \cref{thm:private-sparse-histogram-formal}. 
We show that \cref{algo:tail-item-padding} satisfies the theorem, and analyze its privacy, utility guarantees, and running time separately.
All missing proofs of the lemmas in this section can be found in \cref{sec:proofs-for-dp-sparse-histogram}.

Throughout the proof, we set 
\begin{equation}
    \label{eq:def-tau}
    \tau \doteq \argmin_{t \in \N} \SET{ 
        \P{ 1 + \PurifiedApproxDiscreteLaplaceMechanism{n}{\eps}{\eps\gamma / d}{1} \ge t } \le 1 / d 
    },
\end{equation}
Based on Theorem~\ref{thm:purified-approximate-discrete-laplace-sampler-informal}, it holds that $\tau \in \Theta \paren{1 / \eps \cdot \ln d}$.
Further, denote 
\begin{equation}
    p_\tau \doteq \P{ 1 + \PurifiedApproxDiscreteLaplaceMechanism{n}{\eps}{\eps\gamma / d}{1} \ge \tau }.    
\end{equation}
The definition of $\tau$ implies that $p_\tau \le 1 / d$.

\subsubsection{\bf Privacy Guarantee}

    Let $\hist' \sim \hist$ be a neighboring histogram.
    To distinguish the notations, denote $I_1', I_2'$ and $I'$ be the items sampled by Algorithm~\ref{algo:tail-item-padding} when the input is $\hist'$.
    We want to show that, 
    \begin{equation}
        \label{ineq:private-output-support}
        e^{- \max \paren{ \frac{\eps}{2} + 2 \cdot p_\tau, \eps} } \P{I = S}
        \le \P{I' = S} 
        \le e^{
            \max \PAREN{
                \frac{\eps}{2} + p_\tau \cdot \frac{
                    d - k 
                }{
                    k + 1
                },
                \eps
            }
        } \cdot \P{I = S}, 
        \, 
        \forall \, S \subseteq [d].
    \end{equation}

    Conditioned on $I = I'$, by the privacy guarantee of 
    $\PurifiedApproxDiscreteLaplaceMechanism{n}{\eps}{\eps \gamma / d}$ (Theorem~\ref{thm:purified-approximate-discrete-laplace-sampler-informal}), 
    the values $\set{ \noisyhist{i} : i \in I }$ and $\set{ \noisyhist'[i] : i \in I' }$ are $(\eps, 0)$-indistinguishable.
    By basic composition theorem of DP (\cref{fact:basic-composition}) and that $p_\tau \le 1 / d$, the \cref{algo:tail-item-padding} is 
    \begin{equation*}
        \max \Bigparen{
            \frac{3\eps}{2} + p_\tau \cdot \frac{
                d - k 
            }{
                k + 1
            },
            2\eps   
        }
        \le 
        \max \Bigparen{
            \frac{3\eps}{2} + \frac{
                1
            }{
                k + 1
            },
            2\eps   
        }
        \text{-DP}
    \end{equation*}
    It remains to prove \cref{ineq:private-output-support}.
    
    \subsubsection*{\bf Proving \cref{ineq:private-output-support}.}
    Without lose of generality, we assume that $\hist$ and $\hist'$ differ in two entries indexed by $i^*, j^* \in [d]$.
    Further, we assume that 
    \begin{align*}
        \begin{array}{ccc}
            \hist{i^*} + 1 
                = \hist'[i^*], 
            &
            \hist{j^*} 
                = \hist'[j^*] + 1, 
            &
            \hist{i}
                = \hist'[i],  
                \,
                \forall i \in [d] \setminus \set{ i^*, j^* }.  
        \end{array}
    \end{align*}

    \noindent
    We consider four cases.

    \subsubsection*{\bf Case 1. $\supp \paren{\hist} = \supp \paren{\hist'}$}
    Then clearly the distributions of $I_1$ and $I_1'$ are $(\eps, 0)$-indistinguishable, due to the privacy guarantee of $\PurifiedApproxDiscreteLaplaceMechanism{n}{\eps}{\eps \gamma / d}$.
    Hence so are the ones of $I$ and $I'$ and \cref{ineq:private-output-support} holds.

    \subsubsection*{\bf Case 2. $\supp \paren{\hist} \cup \set{ i^* } = \supp \paren{\hist'}$.}
    In this case, we have 
    $
            \hist'[i^*] = 1, \, \hist{i^*} = 0, \,
            \hist{j^*} > \hist'[j^*]  > 0.
    $
    For convenience, denote $\distSupport$ as the distribution of $I$, i.e.,
    $\distSupport(S) = \P{I = S}$.
    Similarly, let $\distSupport'$ be the distribution of $I'$, with
    $\distSupport'(S) = \P{I' = S}$.
    Let $\eventSelectIstart$ be the event that $i^* \in I_1'$, and let
    $\distSupport'_{\eventSelectIstart}$ and
    $\distSupport'_{\bar{\eventSelectIstart}}$ denote the distributions of $I'$
    conditioned on $\eventSelectIstart$ and $\bar{\eventSelectIstart}$,
    respectively.
    By the definition of $p_\tau$,
    \begin{equation*}
        \distSupport' 
        = p_\tau \cdot \distSupport'_{\eventSelectIstart}
        + \paren{1 - p_\tau} \cdot \distSupport'_{\bar{\eventSelectIstart}}.
        \footnote{
            Given a $(\cZ, \cF)$ measurable space and two probability measures $\mu$ and $\nu$ on it, we write $\mu = \nu$ if $\mu(F) = \nu(F)$ for each $F \in \cF$. 
            We similarly write $\mu \ge \nu$ when $\mu(F) \ge \nu(F)$ for all $F \in \cF$.
        }
    \end{equation*}

    First, to prove the LHS of \cref{ineq:private-output-support}, we invoke Lemma~\ref{lem:dist-support-case-2a}.      
    
    \begin{lemma}\label{lem:dist-support-case-2a}
        $\distSupport$ and $\distSupport'_{\bar{\eventSelectIstart}}$ are
        $\paren{\eps/2, 0}$-indistinguishable.
    \end{lemma}
    
    Intuitively, the lemma holds because, conditioned on the event $\bar{\eventSelectIstart}$ (i.e., $i^* \notin I_1'$), 
    the distributions of $I_1'$ and $I_1$ become $\paren{\eps/2,0}$-indistinguishable: the only remaining difference between $\hist$ and $\hist'$ is the single-coordinate change 
    $\hist[j^*] = \hist'[j^*] + 1$, and the privacy guarantee of 
    $\PurifiedApproxDiscreteLaplaceMechanism{n}{\eps}{\eps\gamma/d}$ ensures this.  
    Moreover, conditioned on $I_1' = I_1$, the distributions of $I_2'$ and $I_2$ coincide.  
    Consequently, $\distSupport'_{\bar{\eventSelectIstart}}$ and $\distSupport$ are $\paren{\eps/2, 0}$-indistinguishable.
  
    By \cref{lem:dist-support-case-2a},
    \begin{align*}
        \distSupport' 
            &\ge 
            \paren{1 - p_\tau} \cdot \distSupport'_{\bar{\eventSelectIstart}}
            \ge 
            \paren{1 - p_\tau} \cdot e^{-\eps/2} \cdot \distSupport
            \ge 
            e^{- 2 p_\tau - \eps/2} \cdot \distSupport,
    \end{align*}
    where we use that $1 - p_\tau \ge e^{-2 p_\tau}$ for $p_\tau \in (0, 0.5)$.

    Second, to prove the RHS of \cref{ineq:private-output-support}, we use Lemma~\ref{lem:dist-support-case-2b}.

    \begin{lemma}
        \label{lem:dist-support-case-2b}
        $\distSupport \ge \frac{k}{d} \cdot e^{-\eps / 2} \cdot \distSupport'_{\eventSelectIstart}$.
    \end{lemma}

    Informally, the lemma holds because, the element $i^*$ can be added to $I$ only during the padding phase
    (Algorithm~\ref{algo:tail-item-padding}, Lines~\ref{line:blanket-sampling}--\ref{line:item-padding}),
    which occurs with probability at least $k/d$.
    Conditioned on this event and on $\eventSelectIstart$, 
    the distributions of $I_1 \cup \set{ i^* }$ and $I_1'$ are $\paren{\eps/2, 0}$-indistinguishable:
    the only remaining difference between $\hist$ and $\hist'$ is the single-coordinate change 
    $\hist[j^*] = \hist'[j^*] + 1$, and the privacy guarantee of 
    $\PurifiedApproxDiscreteLaplaceMechanism{n}{\eps}{\eps\gamma/d}$ ensures this.  
    Moreover, when $I_1 \cup \set{ i^* } = I_1'$ and $i^* \in I_2$, the remaining padding elements in $I_2 \setminus \set{ i^* }$ and $I_2'$ have identical distributions.
    
    Based on \cref{lem:dist-support-case-2a,lem:dist-support-case-2b},
    \begin{align*}
        \distSupport'
        &= p_\tau \cdot \distSupport'_{\eventSelectIstart}  + \paren{1 - p_\tau} \cdot \distSupport'_{\bar{\eventSelectIstart}}  
        \le p_\tau \cdot e^{\eps / 2} \cdot d / k \cdot \distSupport  + \paren{1 - p_\tau} \cdot e^{\eps / 2} \cdot \distSupport \\
        &= e^{\eps / 2} \cdot \PAREN{
            \paren{1 - p_\tau}  + 
            p_\tau \cdot d / k
        } \cdot \distSupport 
        =  e^{\eps / 2} \cdot \PAREN{
            1 + p_\tau \cdot \paren{ d / k - 1}
        } \cdot \distSupport 
        \le e^{
            \eps / 2 + p_\tau \cdot \paren{ d / k - 1}
        } \cdot \distSupport,
    \end{align*}
    where we use $1 + p_\tau \cdot \paren{ d / k - 1} \le \exp \paren{p_\tau \cdot \paren{ d / k - 1}}$.

    \subsubsection*{\bf Case 3. $\supp \paren{\hist} = \supp \paren{\hist'} \cup \set{ j^* }$.}
    The discussion follows from the symmetry of Case 2.

    \subsubsection*{\bf Case 4. $\supp \paren{\hist} \cup \set{i^*} = \supp \paren{\hist'} \cup \set{ j^* }$.}
    Let $\distSupport$, $\distSupport'$, $\distSupport'_{\eventSelectIstart}$ and
    $\distSupport'_{\bar{\eventSelectIstart}}$ be defined as before. 
    Further, let $\eventSelectJstart$ be the event that $j^* \in I_1$, and let
    $\distSupport_{\eventSelectJstart}$ and
    $\distSupport_{\bar{\eventSelectJstart}}$ denote the distributions of $I$
    conditioned on $\eventSelectJstart$ and $\bar{\eventSelectJstart}$,
    respectively.
    By the definition of $p_\tau$,
    \begin{equation*}
        \distSupport
        = p_\tau \cdot \distSupport_{\eventSelectJstart}
        + \paren{1 - p_\tau} \cdot \distSupport_{\bar{\eventSelectJstart}}, 
        \quad 
        \distSupport' 
        = p_\tau \cdot \distSupport'_{\eventSelectIstart}
        + \paren{1 - p_\tau} \cdot \distSupport'_{\bar{\eventSelectIstart}}.
    \end{equation*}
    To prove \cref{ineq:private-output-support}, we need Lemma~\ref{lem:dist-support-case-4}.
    \begin{lemma}
        \label{lem:dist-support-case-4}
        $\distSupport_{\bar{\eventSelectJstart}} = \distSupport'_{\bar{\eventSelectIstart}}$,
        $\distSupport_{\bar{\eventSelectJstart}} \ge \frac{k}{d} \cdot \distSupport'_{\eventSelectIstart}$, 
        $\distSupport'_{\bar{\eventSelectIstart}} \ge \frac{k}{d} \cdot \distSupport_{\eventSelectJstart}
        $.
    \end{lemma}

    We present an intuitive explanation of the lemma.  
    First, conditioned on $\bar{\eventSelectJstart}$ and $\bar{\eventSelectIstart}$, the pair $I_1$ and $I_1'$, as well as the pair $I_2$ and $I_2'$, have identical distributions.
    Therefore, $\distSupport_{\bar{\eventSelectJstart}} = \distSupport'_{\bar{\eventSelectIstart}}$ holds.  
    Second, the element $i^*$ can be added to $I$ only during the padding phase
    (Algorithm~\ref{algo:tail-item-padding}, Lines~\ref{line:blanket-sampling}--\ref{line:item-padding}), which occurs with probability at least $k/d$.  
    Conditioned on this event and on $\bar{\eventSelectJstart}$ and $\eventSelectIstart$, the pair $I_1$ and $I_1'$, as well as the pair $I_2$ and $I_2'$, have identical distributions.
    Therefore, $\distSupport_{\bar{\eventSelectJstart}} \ge \frac{k}{d} \cdot \distSupport'_{\eventSelectIstart}$ holds.  
    The last inequality follows from symmetry of the second one.
    
    Based on \cref{lem:dist-support-case-4}
    \begin{align*}
        \distSupport
        &= p_\tau \cdot \distSupport_{\eventSelectJstart}
        + \paren{1 - p_\tau} \cdot \distSupport_{\bar{\eventSelectJstart}} 
        \le p_\tau \cdot \tfrac{d}{k} \cdot \distSupport'_{\bar{\eventSelectIstart}} 
        + \paren{1 - p_\tau} \cdot \distSupport'_{\bar{\eventSelectIstart}} 
        \le \PAREN{
            1 + p_\tau \cdot \Bigparen{ \frac{d}{k} - 1}
        } \cdot \distSupport'
        \le e^{
            \, p_\tau \cdot \frac{d - k}{k}
        } \distSupport'.
    \end{align*}

\subsubsection{\bf Utility Guarantee}
Consider any $\beta > 2 \eps \gamma$.
Set $\gamma' \doteq \eps \gamma / d$ and $\beta' \doteq \beta / d$, so that $\beta' > 2 \gamma'$.
Recall from \cref{thm:purified-approximate-discrete-laplace-sampler-informal} that for each $t \in \zeroton{n}$,
$\PurifiedApproxDiscreteLaplaceMechanism{n}{\eps}{\gamma'}{t}$ is an $(\alpha,\beta')$-accurate estimator of $t$, where
\[
    \alpha \in O \PAREN{
      \frac{1}{\eps} \ln \frac{1}{\beta'}   
    } 
    = O \PAREN{
        \frac{1}{\eps} \ln \frac{d}{\beta}
    }.
\]
We consider the accuracy of each coordinate $\noisyhist[i]$ across three cases.

\subsubsection*{Case 1 ($i \in I$)}: In this case, $\noisyhist[i]$ is directly constructed from the output of $\PurifiedApproxDiscreteLaplaceMechanism{n}{\eps}{\gamma'} \paren{\hist[i]}$, and is thus an $(\alpha, \beta')$-accurate estimator of $\hist[i]$.

\subsubsection*{Case 2 ($i \in [d] \setminus \paren{ \supp\paren{\hist} \cup I }$)}: Here, $i$ is neither in the support of $\hist$ nor in $I$, so $\noisyhist[i] = 0 = \hist[i]$, and the error is exactly zero.

\subsubsection*{Case 3 ($i \in \supp\paren{\hist} \setminus I$)}: In this case, $\intermediateHist[i]$ is an $(\alpha, \beta')$-accurate estimator of $\hist[i]$.
Moreover, the fact that $i \notin I$ implies that $\intermediateHist[i] < \tau$.
Hence,
\begin{align*}
        \card{\hist{i} - \noisyhist{i}} 
            &= \hist{i} 
            \le \card{\hist{i} - \intermediateHist{i}} + \intermediateHist{i} 
            = \alpha + \tau
            \in O \PAREN{ \frac{1}{\eps} \cdot \ln \frac{d}{\beta} }.
\end{align*}
By a union bound, we conclude that $\noisyhist$ is an $(\alpha, \beta)$-simultaneous accurate estimator of $\hist$.

\subsubsection{\bf Running Time} 
We need to analyze the time for the noise generation, the time for sampling the privacy blanket, and the time of the dictionary operations, separately.

\subsubsection*{Noise Generation.}

Under the assumption in \cref{thm:private-sparse-histogram-formal} that the descriptions of $n$, $\eps$, and $\gamma$ fit in a constant number of machine words (i.e., $\ln n, \ln \tfrac{1}{\eps}, \ln \tfrac{1}{\gamma} \in O(\omega)$, where $\omega$ is the word size), 
Theorem~\ref{thm:purified-approximate-discrete-laplace-sampler-informal} implies that the noise sampler 
$\PurifiedApproxDiscreteLaplaceMechanism{n}{\eps}{\eps\gamma / d}$ 
has initialization time $\tilde{O}(1)$ and uses 
$O\!\left( \tfrac{1}{\eps} + \ln \tfrac{d}{\eps\gamma} + \ln n \right)$ 
words of memory.

After initialization, each sample from the noise distribution can be generated in $O(1)$ time.  
Thus, producing $n+k$ samples takes total time $O(n+k)$.

\subsubsection*{Privacy Blanket Sampling.}
To sample a uniform random subset from $\binom{[d] \setminus I_1}{n + k - \card{I_1}}$ 
(the set of all subsets of size $n + k - \card{I_1}$ from $[d] \setminus I_1$), 
we first sample a uniform random subset of size $n + k$ from $[d]$, 
and then randomly select $n + k - \card{I_1}$ elements from this subset that do not appear in $I_1$.

To generate a uniform random subset of size $n + k$ from $[d]$, we employ the following efficient method:  
sample $m = 4(n + k) \ll d$ elements independently and uniformly from $[d]$, and remove duplicates.  
If the resulting set contains at least $n + k$ distinct elements, we return the first $n + k$ of them.  
Otherwise, there are two options:
\begin{itemize}[leftmargin=0.4cm, label=$\circ$]
    \item \textit{Abort and output a fixed sparse histogram.} 
    Terminate \cref{algo:tail-item-padding} early and return a fixed sparse histogram, 
    e.g., one where elements $1$ through $n$ each have count $1$.
    Although this may slightly degrade utility, it does not compromise privacy, since sampling $m$ independent elements from $[d]$ is independent of the input dataset.  
    We show that this increases the failure probability of the utility guarantee by at most an additive $\sqrt{n + k} \cdot \exp(- (n + k) / 16)$.
    \item \textit{Repeat the sampling procedure until success.}  
    This approach maintains both correctness and privacy, as the repetition process reveals nothing about the underlying histogram $\hist$, nor about the size or contents of the privacy blanket.  
    However, it deviates from our goal of a deterministic linear-time algorithm, yielding instead only expected linear-time performance.
\end{itemize}

\noindent
To bound the failure probability of this sampling procedure, we apply the standard
balls-and-bins model~\citep{MU05}: we throw $m$ balls into $d$ bins and analyze
the number of non-empty bins.

Let $X_i$ denote the number of balls in bin $i$, for each $i \in [d]$, so that
$\sum_{i \in [d]} X_i = m$.  
Let $S_X = \sum_{i \in [d]} \indicator{X_i > 0}$ be the number of non-empty bins.
We aim to upper bound the probability $\P{S_X \le n + k}$.

To analyze this, we apply the Poisson approximation method~\citep{MU05}.  
Let $Y_1, \ldots, Y_d$ be independent Poisson random variables (\cref{def:poisson})
with mean $m/d$, i.e., 
$
    \P{Y_i = t} = e^{-m/d} \cdot \frac{(m/d)^t}{t!}, \quad \forall t \in \N.
$  
Define $S_Y = \sum_{i \in [d]} \indicator{Y_i > 0}$ analogously.  
The Poisson approximation theorem (\cref{thm:poisson-approximation}) implies that
\[
    \P{ S_X \le n + k } \le e \sqrt{m} \cdot \P{ S_Y \le n + k } .
\]
Thus, it suffices to bound $\P{S_Y \le n + k}$.
Since $\P{Y_i > 0} = 1 - e^{-m/d}$, we have
$
    \E{S_Y} = d \cdot \paren{1 - e^{-m/d}} \ge d \cdot \frac{m}{2d} = m/2,
$ 
where the inequality uses $1 - e^{-x} \ge x/2$ for $0 \le x \le 1$, applicable whenever
$m \le d$.

Via the Poisson Chernoff bound (\cref{thm:poisson-tails}),
\begin{equation*}
    \P{ S_Y \le n + k }
        \le \P{ S_Y \le \frac{1}{2} \E{S_Y} }
        \le e^{ - \frac{1}{8} \cdot \E{S_Y} }
        \le e^{ - \frac{m}{16} } .
\end{equation*}

\noindent
\subsubsection*{Dictionary Operations.}
Implementing \cref{algo:tail-item-padding} requires standard dictionary operations—
inserting a key–value pair, deleting a pair, updating the value associated with a key, and
checking whether a key is present.
These operations arise throughout the algorithm:
during the histogram construction step (Line~\ref{line:histogram-construction}),
during the blanket sampling and padding steps
(Lines~\ref{line:blanket-sampling} and~\ref{line:item-padding}),
and when reading or updating values stored in
$\hist$, $\intermediateHist$, and $\noisyhist$.

To support these operations, we use the hashing-based dictionary (Backyard Cuckoo Hashing) of \citet{arbitman2010backyard}.
Deleting a pair, updating a value, and checking whether a key is present all run in deterministic constant time per operation.
Insertions require more care, as they may not always complete in constant time.

The key observation is that in Algorithm~\ref{algo:tail-item-padding}, there are $O(n + k)$ intersections in total.
By \citet{arbitman2010backyard}, for any constant $c$, the dictionary performs all these intersections in total time $O(n + k)$, with failure probability at most $(n + k)^{-c}$.

A natural approach is to enforce a deterministic time bound $T \in O(n + k)$, so that the combined dictionary operations exceed this bound with probability at most $(n + k)^{-c}$.
However, this creates a subtle issue for differential privacy: two neighboring datasets $\DataSet$ and $\DataSet'$ may have different failure probabilities, potentially revealing information about their difference.  
A careful strategy is therefore required to control the impact of these failures on the overall privacy guarantee.

To handle potential failures, we adopt the following padding procedure, ensuring that neighboring datasets have similar failure probabilities with the hash dictionary:
\begin{itemize}[label=$\circ$, leftmargin=6mm]
    \item Flip a coin with probability $\frac{1}{e^\eps - 1} \cdot (n + k)^{-c}$. If heads, output a uniform random histogram.
    \item Otherwise, run \cref{algo:tail-item-padding}. If all dictionary operations combined exceed time $T$, halt \cref{algo:tail-item-padding} and output a uniform random histogram $\noisyhist$.
\end{itemize}

It remains to analyze its impact on the privacy guarantee.
The output distribution of Algorithm~\ref{algo:tail-item-padding} ($\ourAlgo$) on input dataset $\DataSet$, denoted $\bar{\mu}$, can be decomposed into two components: the distribution conditioned on the dictionary operations completing within the enforced time bound, denoted $\mu$, and a uniform distribution over all possible outputs, denoted $\nu$.  
There exist weights $w_1, w_2 \in [0,1]$ such that  
\[
    (1)\;\bar{\mu} = w_1 \mu + w_2 \nu, \; 
    (2)\;w_1 + w_2 = 1, \text{ and } \; 
    (3)\;
    \frac{1}{e^\eps - 1} \cdot (n + k)^{-c} \le w_2 \le \frac{e^\eps}{e^\eps - 1} \cdot (n + k)^{-c}.
\]

Similarly, for a neighboring dataset $\DataSet'$, let $\mu'$ be the output distribution of $\ourAlgo$ when the dictionary operations do not exceed the time bound.  
There exist weights $w_1', w_2' \in [0,1]$ such that  
\[
    (1) \bar{\mu}' = w_1' \mu' + w_2' \nu,\;
    (2) w_1' + w_2' = 1, \text{ and }\; 
    (3)\,
    \frac{1}{e^\eps - 1} \cdot (n + k)^{-c} \le w_2' \le \frac{e^\eps}{e^\eps - 1} \cdot (n + k)^{-c}.
\]
Therefore, 
$
    \frac{w_2}{w_2'} \le e^\eps, \, 
    \text{ and  } \, 
    \frac{w_1}{w_1'} \le \frac{1}{
        {1 - \frac{e^\eps}{e^\eps - 1} \cdot (n + k)^{-c}}
    }.
$
Further, based on the privacy analysis, $\mu$ and $\mu'$ are $\eps$-indistinguishable. 
For each measurable set $E$ in the range of $\ourAlgo$, it holds that
\begin{align*}
    \frac{\bar{\mu}(E)}{\bar{\mu}'(E)}
    &=
    \frac{w_1 \, \mu(E) + w_2 \, \nu(E)}{w_1' \, \mu'(E) + w_2' \, \nu(E)} 
    \le 
    \max \SET{
        \frac{w_1 \, \mu(E)}{w_1' \, \mu'(E)},
        \frac{w_2 \, \nu(E)}{w_2' \, \nu(E)}
    } 
    \le
    \frac{e^\eps}{1 - \frac{e^\eps}{e^\eps - 1} \cdot (n + k)^{-c}}
    \le e^{\eps + \frac{2}{n^2}}.
\end{align*}
where the last inequality holds since, for a sufficiently large constant $c$, 
the term 
$\frac{e^\eps}{e^\eps - 1} \cdot (n + k)^{-c} \le \frac{1}{n^2}$, 
so that 
$1 - \frac{e^\eps}{e^\eps - 1} \cdot (n + k)^{-c} \ge e^{- 2 / n^2}$.
Finally, under the assumption that $1/\eps \le n$, we have $e^{\eps + \frac{2}{n^2}} \le e^{\eps + \frac{2 \eps}{n}}$.

\clearpage
\bibliographystyle{IEEEtranSN}
\bibliography{reference}

\begin{thebibliography}{52}
\providecommand{\natexlab}[1]{#1}
\providecommand{\url}[1]{#1}
\csname url@samestyle\endcsname
\providecommand{\newblock}{\relax}
\providecommand{\bibinfo}[2]{#2}
\providecommand{\BIBentrySTDinterwordspacing}{\spaceskip=0pt\relax}
\providecommand{\BIBentryALTinterwordstretchfactor}{4}
\providecommand{\BIBentryALTinterwordspacing}{\spaceskip=\fontdimen2\font plus
\BIBentryALTinterwordstretchfactor\fontdimen3\font minus \fontdimen4\font\relax}
\providecommand{\BIBforeignlanguage}[2]{{%
\expandafter\ifx\csname l@#1\endcsname\relax
\typeout{** WARNING: IEEEtranSN.bst: No hyphenation pattern has been}%
\typeout{** loaded for the language `#1'. Using the pattern for}%
\typeout{** the default language instead.}%
\else
\language=\csname l@#1\endcsname
\fi
#2}}
\providecommand{\BIBdecl}{\relax}
\BIBdecl

\bibitem[Ajtai et~al.(1983)Ajtai, Koml{\'{o}}s, and Szemer{\'{e}}di]{AjtaiKS83}
\BIBentryALTinterwordspacing
M.~Ajtai, J.~Koml{\'{o}}s, and E.~Szemer{\'{e}}di, ``An o(n log n) sorting network,'' in \emph{Proceedings of the 15th Annual {ACM} Symposium on Theory of Computing, 25-27 April, 1983, Boston, Massachusetts, {USA}}, D.~S. Johnson, R.~Fagin, M.~L. Fredman, D.~Harel, R.~M. Karp, N.~A. Lynch, C.~H. Papadimitriou, R.~L. Rivest, W.~L. Ruzzo, and J.~I. Seiferas, Eds.\hskip 1em plus 0.5em minus 0.4em\relax {ACM}, 1983, pp. 1--9. [Online]. Available: \url{https://doi.org/10.1145/800061.808726}
\BIBentrySTDinterwordspacing

\bibitem[Anderson et~al.(2024)Anderson, Chase, Durak, Laine, and Weng]{AndersonCDLW24}
E.~Anderson, M.~Chase, F.~B. Durak, K.~Laine, and C.~Weng, ``Precio: Private aggregate measurement via oblivious shuffling,'' in \emph{Proceedings of the 2024 on {ACM} {SIGSAC} Conference on Computer and Communications Security, {CCS} 2024, Salt Lake City, UT, USA, October 14-18, 2024}.\hskip 1em plus 0.5em minus 0.4em\relax {ACM}, 2024, pp. 1819--1833.

\bibitem[Arbitman et~al.(2010)Arbitman, Naor, and Segev]{arbitman2010backyard}
Y.~Arbitman, M.~Naor, and G.~Segev, ``Backyard cuckoo hashing: Constant worst-case operations with a succinct representation,'' in \emph{2010 IEEE 51st Annual Symposium on Foundations of Computer Science}, 2010, pp. 787--796.

\bibitem[Balcer and Vadhan(2019)]{BalcerV19}
\BIBentryALTinterwordspacing
V.~Balcer and S.~P. Vadhan, ``Differential privacy on finite computers,'' \emph{J. Priv. Confidentiality}, vol.~9, no.~2, 2019. [Online]. Available: \url{https://doi.org/10.29012/jpc.679}
\BIBentrySTDinterwordspacing

\bibitem[Balle et~al.(2019)Balle, Bell, Gasc{\'{o}}n, and Nissim]{BalleBGN19}
B.~Balle, J.~Bell, A.~Gasc{\'{o}}n, and K.~Nissim, ``The privacy blanket of the shuffle model,'' in \emph{Advances in Cryptology - {CRYPTO} 2019 - 39th Annual International Cryptology Conference, Santa Barbara, CA, USA, August 18-22, 2019, Proceedings, Part {II}}, ser. Lecture Notes in Computer Science, A.~Boldyreva and D.~Micciancio, Eds., vol. 11693.\hskip 1em plus 0.5em minus 0.4em\relax Springer, 2019, pp. 638--667.

\bibitem[Bassily and Smith(2015)]{BS15}
R.~Bassily and A.~D. Smith, ``Local, private, efficient protocols for succinct histograms,'' in \emph{Proceedings of the Forty-Seventh Annual {ACM} on Symposium on Theory of Computing, {STOC} 2015, Portland, OR, USA, June 14-17, 2015}, R.~A. Servedio and R.~Rubinfeld, Eds.\hskip 1em plus 0.5em minus 0.4em\relax {ACM}, 2015, pp. 127--135.

\bibitem[Beimel et~al.(2014)Beimel, Brenner, Kasiviswanathan, and Nissim]{BeimelBKN14}
\BIBentryALTinterwordspacing
A.~Beimel, H.~Brenner, S.~P. Kasiviswanathan, and K.~Nissim, ``Bounds on the sample complexity for private learning and private data release,'' \emph{Mach. Learn.}, vol.~94, no.~3, pp. 401--437, 2014. [Online]. Available: \url{https://doi.org/10.1007/s10994-013-5404-1}
\BIBentrySTDinterwordspacing

\bibitem[Bell et~al.(2022)Bell, Gasc{\'{o}}n, Ghazi, Kumar, Manurangsi, Raykova, and Schoppmann]{GG0M0S22}
\BIBentryALTinterwordspacing
J.~Bell, A.~Gasc{\'{o}}n, B.~Ghazi, R.~Kumar, P.~Manurangsi, M.~Raykova, and P.~Schoppmann, ``Distributed, private, sparse histograms in the two-server model,'' in \emph{Proceedings of the 2022 {ACM} {SIGSAC} Conference on Computer and Communications Security, {CCS} 2022, Los Angeles, CA, USA, November 7-11, 2022}, H.~Yin, A.~Stavrou, C.~Cremers, and E.~Shi, Eds.\hskip 1em plus 0.5em minus 0.4em\relax {ACM}, 2022, pp. 307--321. [Online]. Available: \url{https://doi.org/10.1145/3548606.3559383}
\BIBentrySTDinterwordspacing

\bibitem[Ben{-}Or et~al.(1988)Ben{-}Or, Goldwasser, and Wigderson]{BGW88}
\BIBentryALTinterwordspacing
M.~Ben{-}Or, S.~Goldwasser, and A.~Wigderson, ``Completeness theorems for non-cryptographic fault-tolerant distributed computation (extended abstract),'' in \emph{Proceedings of the 20th Annual {ACM} Symposium on Theory of Computing, May 2-4, 1988, Chicago, Illinois, {USA}}, J.~Simon, Ed.\hskip 1em plus 0.5em minus 0.4em\relax {ACM}, 1988, pp. 1--10. [Online]. Available: \url{https://doi.org/10.1145/62212.62213}
\BIBentrySTDinterwordspacing

\bibitem[Bittau et~al.(2017)Bittau, Erlingsson, Maniatis, Mironov, Raghunathan, Lie, Rudominer, Kode, Tinn{\'{e}}s, and Seefeld]{BittauEMMRLRKTS17}
\BIBentryALTinterwordspacing
A.~Bittau, {\'{U}}.~Erlingsson, P.~Maniatis, I.~Mironov, A.~Raghunathan, D.~Lie, M.~Rudominer, U.~Kode, J.~Tinn{\'{e}}s, and B.~Seefeld, ``Prochlo: Strong privacy for analytics in the crowd,'' in \emph{Proceedings of the 26th Symposium on Operating Systems Principles, Shanghai, China, October 28-31, 2017}.\hskip 1em plus 0.5em minus 0.4em\relax {ACM}, 2017, pp. 441--459. [Online]. Available: \url{https://doi.org/10.1145/3132747.3132769}
\BIBentrySTDinterwordspacing

\bibitem[Brent and Zimmermann(2010)]{Brent_Zimmermann_2010}
\BIBentryALTinterwordspacing
R.~P. Brent and P.~Zimmermann, \emph{Modern Computer Arithmetic}, ser. Cambridge Monographs on Applied and Computational Mathematics.\hskip 1em plus 0.5em minus 0.4em\relax Cambridge University Press, 2010. [Online]. Available: \url{https://doi.org/10.1017/CBO9780511921698}
\BIBentrySTDinterwordspacing

\bibitem[Bun et~al.(2019{\natexlab{b}})Bun, Nelson, and Stemmer]{BNS19}
M.~Bun, J.~Nelson, and U.~Stemmer, ``Heavy hitters and the structure of local privacy,'' \emph{{ACM} Trans. Algorithms}, vol.~15, no.~4, pp. 51:1--51:40, 2019.

\bibitem[Bun et~al.(2019{\natexlab{a}})Bun, Nissim, and Stemmer]{BunNS19}
\BIBentryALTinterwordspacing
M.~Bun, K.~Nissim, and U.~Stemmer, ``Simultaneous private learning of multiple concepts,'' \emph{J. Mach. Learn. Res.}, vol.~20, pp. 94:1--94:34, 2019. [Online]. Available: \url{https://jmlr.org/papers/v20/18-549.html}
\BIBentrySTDinterwordspacing

\bibitem[Canonne et~al.(2020)Canonne, Kamath, and Steinke]{Canonne0S20}
\BIBentryALTinterwordspacing
C.~L. Canonne, G.~Kamath, and T.~Steinke, ``The discrete gaussian for differential privacy,'' in \emph{Advances in Neural Information Processing Systems 33: Annual Conference on Neural Information Processing Systems 2020, NeurIPS 2020, December 6-12, 2020, virtual}, H.~Larochelle, M.~Ranzato, R.~Hadsell, M.~Balcan, and H.~Lin, Eds., 2020. [Online]. Available: \url{https://proceedings.neurips.cc/paper/2020/hash/b53b3a3d6ab90ce0268229151c9bde11-Abstract.html}
\BIBentrySTDinterwordspacing

\bibitem[Champion et~al.(2019)Champion, Shelat, and Ullman]{ChampionSU19}
\BIBentryALTinterwordspacing
J.~Champion, A.~Shelat, and J.~R. Ullman, ``Securely sampling biased coins with applications to differential privacy,'' in \emph{Proceedings of the 2019 {ACM} {SIGSAC} Conference on Computer and Communications Security, {CCS} 2019, London, UK, November 11-15, 2019}, L.~Cavallaro, J.~Kinder, X.~Wang, and J.~Katz, Eds.\hskip 1em plus 0.5em minus 0.4em\relax {ACM}, 2019, pp. 603--614. [Online]. Available: \url{https://doi.org/10.1145/3319535.3354256}
\BIBentrySTDinterwordspacing

\bibitem[Cheu et~al.(2019)Cheu, Smith, Ullman, Zeber, and Zhilyaev]{CheuSUZZ19}
A.~Cheu, A.~D. Smith, J.~R. Ullman, D.~Zeber, and M.~Zhilyaev, ``Distributed differential privacy via shuffling,'' in \emph{Advances in Cryptology - {EUROCRYPT} 2019 - 38th Annual International Conference on the Theory and Applications of Cryptographic Techniques, Darmstadt, Germany, May 19-23, 2019, Proceedings, Part {I}}, ser. Lecture Notes in Computer Science, Y.~Ishai and V.~Rijmen, Eds., vol. 11476.\hskip 1em plus 0.5em minus 0.4em\relax Springer, 2019, pp. 375--403.

\bibitem[Cormode et~al.(2012)Cormode, Procopiuc, Srivastava, and Tran]{CormodePST12}
\BIBentryALTinterwordspacing
G.~Cormode, C.~M. Procopiuc, D.~Srivastava, and T.~T.~L. Tran, ``Differentially private summaries for sparse data,'' in \emph{15th International Conference on Database Theory, {ICDT} '12, Berlin, Germany, March 26-29, 2012}, A.~Deutsch, Ed.\hskip 1em plus 0.5em minus 0.4em\relax {ACM}, 2012, pp. 299--311. [Online]. Available: \url{https://doi.org/10.1145/2274576.2274608}
\BIBentrySTDinterwordspacing

\bibitem[Corrigan{-}Gibbs and Boneh(2017)]{Corrigan-GibbsB17}
H.~Corrigan{-}Gibbs and D.~Boneh, ``Prio: Private, robust, and scalable computation of aggregate statistics,'' in \emph{14th {USENIX} Symposium on Networked Systems Design and Implementation, {NSDI} 2017, Boston, MA, USA, March 27-29, 2017}, A.~Akella and J.~Howell, Eds.\hskip 1em plus 0.5em minus 0.4em\relax {USENIX} Association, 2017, pp. 259--282.

\bibitem[Dov et~al.(2023)Dov, David, Naor, and Tzalik]{DovDNT23}
\BIBentryALTinterwordspacing
Y.~B. Dov, L.~David, M.~Naor, and E.~Tzalik, ``Resistance to timing attacks for sampling and privacy preserving schemes,'' in \emph{4th Symposium on Foundations of Responsible Computing, {FORC} 2023, June 7-9, 2023, Stanford University, California, {USA}}, ser. LIPIcs, K.~Talwar, Ed., vol. 256.\hskip 1em plus 0.5em minus 0.4em\relax Schloss Dagstuhl - Leibniz-Zentrum f{\"{u}}r Informatik, 2023, pp. 11:1--11:23. [Online]. Available: \url{https://doi.org/10.4230/LIPIcs.FORC.2023.11}
\BIBentrySTDinterwordspacing

\bibitem[Dwork(2006)]{Dwork06}
C.~Dwork, ``Differential privacy,'' in \emph{Automata, Languages and Programming, 33rd International Colloquium, {ICALP} 2006, Venice, Italy, July 10-14, 2006, Proceedings, Part {II}}, ser. Lecture Notes in Computer Science, M.~Bugliesi, B.~Preneel, V.~Sassone, and I.~Wegener, Eds., vol. 4052.\hskip 1em plus 0.5em minus 0.4em\relax Springer, 2006, pp. 1--12.

\bibitem[Dwork and Roth(2014)]{DR14}
C.~Dwork and A.~Roth, ``The algorithmic foundations of differential privacy,'' \emph{Found. Trends Theor. Comput. Sci.}, vol.~9, no. 3-4, pp. 211--407, 2014.

\bibitem[Dwork et~al.(2006{\natexlab{b}})Dwork, Kenthapadi, McSherry, Mironov, and Naor]{DworkKMMN06}
\BIBentryALTinterwordspacing
C.~Dwork, K.~Kenthapadi, F.~McSherry, I.~Mironov, and M.~Naor, ``Our data, ourselves: Privacy via distributed noise generation,'' in \emph{Advances in Cryptology - {EUROCRYPT} 2006, 25th Annual International Conference on the Theory and Applications of Cryptographic Techniques, St. Petersburg, Russia, May 28 - June 1, 2006, Proceedings}, ser. Lecture Notes in Computer Science, S.~Vaudenay, Ed., vol. 4004.\hskip 1em plus 0.5em minus 0.4em\relax Springer, 2006, pp. 486--503. [Online]. Available: \url{https://doi.org/10.1007/11761679\_29}
\BIBentrySTDinterwordspacing

\bibitem[Dwork et~al.(2006{\natexlab{a}})Dwork, McSherry, Nissim, and Smith]{DworkMNS06}
C.~Dwork, F.~McSherry, K.~Nissim, and A.~D. Smith, ``Calibrating noise to sensitivity in private data analysis,'' in \emph{Theory of Cryptography, Third Theory of Cryptography Conference, {TCC} 2006, New York, NY, USA, March 4-7, 2006, Proceedings}, ser. Lecture Notes in Computer Science, S.~Halevi and T.~Rabin, Eds., vol. 3876.\hskip 1em plus 0.5em minus 0.4em\relax Springer, 2006, pp. 265--284.

\bibitem[Fanti et~al.(2016)Fanti, Pihur, and Erlingsson]{FPE16}
G.~C. Fanti, V.~Pihur, and {\'{U}}.~Erlingsson, ``Building a {RAPPOR} with the unknown: Privacy-preserving learning of associations and data dictionaries,'' \emph{Proc. Priv. Enhancing Technol.}, vol. 2016, no.~3, pp. 41--61, 2016.

\bibitem[Franzese et~al.(2025)Franzese, Fang, Garg, Jha, Papernot, Wang, and Dziedzic]{franzese2025secure}
\BIBentryALTinterwordspacing
O.~Franzese, C.~Fang, R.~Garg, S.~Jha, N.~Papernot, X.~Wang, and A.~Dziedzic, ``Secure noise sampling for differentially private collaborative learning,'' Cryptology {ePrint} Archive, Paper 2025/1025, 2025. [Online]. Available: \url{https://eprint.iacr.org/2025/1025}
\BIBentrySTDinterwordspacing

\bibitem[Ghazi et~al.(2021)Ghazi, Golowich, Kumar, Pagh, and Velingker]{GhaziG0PV21}
\BIBentryALTinterwordspacing
B.~Ghazi, N.~Golowich, R.~Kumar, R.~Pagh, and A.~Velingker, ``On the power of multiple anonymous messages: Frequency estimation and selection in the shuffle model of differential privacy,'' in \emph{Advances in Cryptology - {EUROCRYPT} 2021 - 40th Annual International Conference on the Theory and Applications of Cryptographic Techniques, Zagreb, Croatia, October 17-21, 2021, Proceedings, Part {III}}, ser. Lecture Notes in Computer Science, A.~Canteaut and F.~Standaert, Eds., vol. 12698.\hskip 1em plus 0.5em minus 0.4em\relax Springer, 2021, pp. 463--488. [Online]. Available: \url{https://doi.org/10.1007/978-3-030-77883-5\_16}
\BIBentrySTDinterwordspacing

\bibitem[Ghosh et~al.(2009)Ghosh, Roughgarden, and Sundararajan]{GhoshRS09}
A.~Ghosh, T.~Roughgarden, and M.~Sundararajan, ``Universally utility-maximizing privacy mechanisms,'' in \emph{Proceedings of the 41st Annual {ACM} Symposium on Theory of Computing, {STOC} 2009, Bethesda, MD, USA, May 31 - June 2, 2009}, M.~Mitzenmacher, Ed.\hskip 1em plus 0.5em minus 0.4em\relax {ACM}, 2009, pp. 351--360.

\bibitem[Goldreich et~al.(1987)Goldreich, Micali, and Wigderson]{GMW87}
\BIBentryALTinterwordspacing
O.~Goldreich, S.~Micali, and A.~Wigderson, ``How to play any mental game or {A} completeness theorem for protocols with honest majority,'' in \emph{Proceedings of the 19th Annual {ACM} Symposium on Theory of Computing, 1987, New York, New York, {USA}}, A.~V. Aho, Ed.\hskip 1em plus 0.5em minus 0.4em\relax {ACM}, 1987, pp. 218--229. [Online]. Available: \url{https://doi.org/10.1145/28395.28420}
\BIBentrySTDinterwordspacing

\bibitem[Google(2020)]{google_noise}
\BIBentryALTinterwordspacing
Google, ``Secure noise generation,'' 2020. [Online]. Available: \url{https://github.com/google/differential-privacy/blob/master/common_docs/Secure_Noise_Generation.pdf}
\BIBentrySTDinterwordspacing

\bibitem[Hardt and Talwar(2010)]{HardtT10}
\BIBentryALTinterwordspacing
M.~Hardt and K.~Talwar, ``On the geometry of differential privacy,'' in \emph{Proceedings of the 42nd {ACM} Symposium on Theory of Computing, {STOC} 2010, Cambridge, Massachusetts, USA, 5-8 June 2010}, L.~J. Schulman, Ed.\hskip 1em plus 0.5em minus 0.4em\relax {ACM}, 2010, pp. 705--714. [Online]. Available: \url{https://doi.org/10.1145/1806689.1806786}
\BIBentrySTDinterwordspacing

\bibitem[Harris et~al.(2020)Harris, Millman, van~der Walt, Gommers, Virtanen, Cournapeau, Wieser, Taylor, Berg, Smith, Kern, Picus, Hoyer, van Kerkwijk, Brett, Haldane, del R{\'{i}}o, Wiebe, Peterson, G{\'{e}}rard-Marchant, Sheppard, Reddy, Weckesser, Abbasi, Gohlke, and Oliphant]{harris2020array}
\BIBentryALTinterwordspacing
C.~R. Harris, K.~J. Millman, S.~J. van~der Walt, R.~Gommers, P.~Virtanen, D.~Cournapeau, E.~Wieser, J.~Taylor, S.~Berg, N.~J. Smith, R.~Kern, M.~Picus, S.~Hoyer, M.~H. van Kerkwijk, M.~Brett, A.~Haldane, J.~F. del R{\'{i}}o, M.~Wiebe, P.~Peterson, P.~G{\'{e}}rard-Marchant, K.~Sheppard, T.~Reddy, W.~Weckesser, H.~Abbasi, C.~Gohlke, and T.~E. Oliphant, ``Array programming with {NumPy},'' \emph{Nature}, vol. 585, no. 7825, pp. 357--362, Sep. 2020. [Online]. Available: \url{https://doi.org/10.1038/s41586-020-2649-2}
\BIBentrySTDinterwordspacing

\bibitem[Jin et~al.(2022)Jin, McMurtry, Rubinstein, and Ohrimenko]{JinMRO22}
\BIBentryALTinterwordspacing
J.~Jin, E.~McMurtry, B.~I.~P. Rubinstein, and O.~Ohrimenko, ``Are we there yet? timing and floating-point attacks on differential privacy systems,'' in \emph{43rd {IEEE} Symposium on Security and Privacy, {SP} 2022, San Francisco, CA, USA, May 22-26, 2022}.\hskip 1em plus 0.5em minus 0.4em\relax {IEEE}, 2022, pp. 473--488. [Online]. Available: \url{https://doi.org/10.1109/SP46214.2022.9833672}
\BIBentrySTDinterwordspacing

\bibitem[Keller et~al.(2024)Keller, M{\"{o}}llering, Schneider, Tkachenko, and Zhao]{KellerM0TZ24}
\BIBentryALTinterwordspacing
H.~Keller, H.~M{\"{o}}llering, T.~Schneider, O.~Tkachenko, and L.~Zhao, ``Secure noise sampling for {DP} in {MPC} with finite precision,'' in \emph{Proceedings of the 19th International Conference on Availability, Reliability and Security, {ARES} 2024, Vienna, Austria, 30 July 2024 - 2 August 2024}.\hskip 1em plus 0.5em minus 0.4em\relax {ACM}, 2024, pp. 25:1--25:12. [Online]. Available: \url{https://doi.org/10.1145/3664476.3664490}
\BIBentrySTDinterwordspacing

\bibitem[Keller(2020)]{mp-spdz}
\BIBentryALTinterwordspacing
M.~Keller, ``{MP-SPDZ}: A versatile framework for multi-party computation,'' in \emph{Proceedings of the 2020 ACM SIGSAC Conference on Computer and Communications Security}, 2020. [Online]. Available: \url{https://doi.org/10.1145/3372297.3417872}
\BIBentrySTDinterwordspacing

\bibitem[Knott et~al.(2021)Knott, Venkataraman, Hannun, Sengupta, Ibrahim, and van~der Maaten]{KnottVHSIM21}
\BIBentryALTinterwordspacing
B.~Knott, S.~Venkataraman, A.~Y. Hannun, S.~Sengupta, M.~Ibrahim, and L.~van~der Maaten, ``Crypten: Secure multi-party computation meets machine learning,'' in \emph{Advances in Neural Information Processing Systems 34: Annual Conference on Neural Information Processing Systems 2021, NeurIPS 2021, December 6-14, 2021, virtual}, M.~Ranzato, A.~Beygelzimer, Y.~N. Dauphin, P.~Liang, and J.~W. Vaughan, Eds., 2021, pp. 4961--4973. [Online]. Available: \url{https://proceedings.neurips.cc/paper/2021/hash/2754518221cfbc8d25c13a06a4cb8421-Abstract.html}
\BIBentrySTDinterwordspacing

\bibitem[Knuth and Yao(1976)]{Knuth1976TheCO}
\BIBentryALTinterwordspacing
D.~E. Knuth and A.~C.-C. Yao, ``The complexity of nonuniform random number generation,'' in \emph{Algorithm and Complexity, New Directions and Results}, 1976. [Online]. Available: \url{https://api.semanticscholar.org/CorpusID:115400979}
\BIBentrySTDinterwordspacing

\bibitem[Korolova et~al.(2009)Korolova, Kenthapadi, Mishra, and Ntoulas]{KorolovaKMN09}
\BIBentryALTinterwordspacing
A.~Korolova, K.~Kenthapadi, N.~Mishra, and A.~Ntoulas, ``Releasing search queries and clicks privately,'' in \emph{Proceedings of the 18th International Conference on World Wide Web, {WWW} 2009, Madrid, Spain, April 20-24, 2009}, J.~Quemada, G.~Le{\'{o}}n, Y.~S. Maarek, and W.~Nejdl, Eds.\hskip 1em plus 0.5em minus 0.4em\relax {ACM}, 2009, pp. 171--180. [Online]. Available: \url{https://doi.org/10.1145/1526709.1526733}
\BIBentrySTDinterwordspacing

\bibitem[Lebeda and Tetek(2024)]{LebedaT24}
\BIBentryALTinterwordspacing
C.~J. Lebeda and J.~Tetek, ``Better differentially private approximate histograms and heavy hitters using the misra-gries sketch,'' \emph{{SIGMOD} Rec.}, vol.~53, no.~1, pp. 7--14, 2024. [Online]. Available: \url{https://doi.org/10.1145/3665252.3665255}
\BIBentrySTDinterwordspacing

\bibitem[Lebeda et~al.(2022)Lebeda, Aum{\"{u}}ller, and Pagh]{LebedaAP22}
\BIBentryALTinterwordspacing
C.~J. Lebeda, M.~Aum{\"{u}}ller, and R.~Pagh, ``Representing sparse vectors with differential privacy, low error, optimal space, and fast access,'' \emph{J. Priv. Confidentiality}, vol.~12, no.~2, 2022. [Online]. Available: \url{https://doi.org/10.29012/jpc.809}
\BIBentrySTDinterwordspacing

\bibitem[Liese and Vajda(2006)]{LieseV06}
\BIBentryALTinterwordspacing
F.~Liese and I.~Vajda, ``On divergences and informations in statistics and information theory,'' \emph{{IEEE} Trans. Inf. Theory}, vol.~52, no.~10, pp. 4394--4412, 2006. [Online]. Available: \url{https://doi.org/10.1109/TIT.2006.881731}
\BIBentrySTDinterwordspacing

\bibitem[Lolck and Pagh(2024)]{LolckP24}
\BIBentryALTinterwordspacing
D.~R. Lolck and R.~Pagh, ``Shannon meets gray: Noise-robust, low-sensitivity codes with applications in differential privacy,'' in \emph{Proceedings of the 2024 {ACM-SIAM} Symposium on Discrete Algorithms, {SODA} 2024, Alexandria, VA, USA, January 7-10, 2024}, D.~P. Woodruff, Ed.\hskip 1em plus 0.5em minus 0.4em\relax {SIAM}, 2024, pp. 1050--1066. [Online]. Available: \url{https://doi.org/10.1137/1.9781611977912.40}
\BIBentrySTDinterwordspacing

\bibitem[Mironov(2012)]{Mironov12}
\BIBentryALTinterwordspacing
I.~Mironov, ``On significance of the least significant bits for differential privacy,'' in \emph{the {ACM} Conference on Computer and Communications Security, CCS'12, Raleigh, NC, USA, October 16-18, 2012}, T.~Yu, G.~Danezis, and V.~D. Gligor, Eds.\hskip 1em plus 0.5em minus 0.4em\relax {ACM}, 2012, pp. 650--661. [Online]. Available: \url{https://doi.org/10.1145/2382196.2382264}
\BIBentrySTDinterwordspacing

\bibitem[Mitzenmacher and Upfal(2005)]{MU05}
M.~Mitzenmacher and E.~Upfal, \emph{Probability and Computing: Randomized Algorithms and Probabilistic Analysis}.\hskip 1em plus 0.5em minus 0.4em\relax Cambridge University Press, 2005.

\bibitem[Qiu and Yi(2025)]{DBLP:journals/corr/QiuY25}
\BIBentryALTinterwordspacing
Y.~Qiu and K.~Yi, ``Approximate {DBSCAN} under differential privacy,'' \emph{CoRR}, vol. abs/2508.08749, 2025. [Online]. Available: \url{https://doi.org/10.48550/arXiv.2508.08749}
\BIBentrySTDinterwordspacing

\bibitem[Ratliff and Vadhan(2024)]{RatliffV24}
\BIBentryALTinterwordspacing
Z.~Ratliff and S.~P. Vadhan, ``A framework for differential privacy against timing attacks,'' in \emph{Proceedings of the 2024 on {ACM} {SIGSAC} Conference on Computer and Communications Security, {CCS} 2024, Salt Lake City, UT, USA, October 14-18, 2024}, B.~Luo, X.~Liao, J.~Xu, E.~Kirda, and D.~Lie, Eds.\hskip 1em plus 0.5em minus 0.4em\relax {ACM}, 2024, pp. 3615--3629. [Online]. Available: \url{https://doi.org/10.1145/3658644.3690206}
\BIBentrySTDinterwordspacing

\bibitem[Ratliff and Vadhan(2025)]{ratliff2025securing}
------, ``Securing unbounded differential privacy against timing attacks,'' \emph{arXiv preprint \href{https://www.arxiv.org/abs/2506.07868}{arXiv:2506.07868}}, 2025.

\bibitem[Sendler(1975)]{sendler1975note}
\BIBentryALTinterwordspacing
W.~Sendler, ``{A Note on the Proof of the Zero-One Law of Blum and Pathak},'' \emph{The Annals of Probability}, vol.~3, no.~6, pp. 1055 -- 1058, 1975. [Online]. Available: \url{https://doi.org/10.1214/aop/1176996234}
\BIBentrySTDinterwordspacing

\bibitem[Steerneman(1983)]{steerneman1983total}
\BIBentryALTinterwordspacing
T.~Steerneman, ``On the total variation and hellinger distance between signed measures; an application to product measures,'' \emph{Proceedings of the American Mathematical Society}, vol.~88, no.~4, pp. 684--688, 1983. [Online]. Available: \url{http://www.jstor.org/stable/2045462}
\BIBentrySTDinterwordspacing

\bibitem[Walker(1977)]{Walker77}
\BIBentryALTinterwordspacing
A.~J. Walker, ``An efficient method for generating discrete random variables with general distributions,'' \emph{{ACM} Trans. Math. Softw.}, vol.~3, no.~3, pp. 253--256, 1977. [Online]. Available: \url{https://doi.org/10.1145/355744.355749}
\BIBentrySTDinterwordspacing

\bibitem[Wei et~al.(2023)Wei, Yu, Fan, Chen, and Wang]{WeiYFCW23}
\BIBentryALTinterwordspacing
C.~Wei, R.~Yu, Y.~Fan, W.~Chen, and T.~Wang, ``Securely sampling discrete gaussian noise for multi-party differential privacy,'' in \emph{Proceedings of the 2023 {ACM} {SIGSAC} Conference on Computer and Communications Security, {CCS} 2023, Copenhagen, Denmark, November 26-30, 2023}, W.~Meng, C.~D. Jensen, C.~Cremers, and E.~Kirda, Eds.\hskip 1em plus 0.5em minus 0.4em\relax {ACM}, 2023, pp. 2262--2276. [Online]. Available: \url{https://doi.org/10.1145/3576915.3616641}
\BIBentrySTDinterwordspacing

\bibitem[Wu and Wirth(2022)]{DBLP:conf/aistats/WuW22}
H.~Wu and A.~Wirth, ``Asymptotically optimal locally private heavy hitters via parameterized sketches,'' in \emph{International Conference on Artificial Intelligence and Statistics, {AISTATS} 2022, 28-30 March 2022, Virtual Event}, ser. Proceedings of Machine Learning Research, G.~Camps{-}Valls, F.~J.~R. Ruiz, and I.~Valera, Eds., vol. 151.\hskip 1em plus 0.5em minus 0.4em\relax {PMLR}, 2022, pp. 7766--7798.

\bibitem[Zahur and Evans(2013)]{ZahurE13}
\BIBentryALTinterwordspacing
S.~Zahur and D.~Evans, ``Circuit structures for improving efficiency of security and privacy tools,'' in \emph{2013 {IEEE} Symposium on Security and Privacy, {SP} 2013, Berkeley, CA, USA, May 19-22, 2013}.\hskip 1em plus 0.5em minus 0.4em\relax {IEEE} Computer Society, 2013, pp. 493--507. [Online]. Available: \url{https://doi.org/10.1109/SP.2013.40}
\BIBentrySTDinterwordspacing

\end{thebibliography}

\newpage

\appendix
\addcontentsline{toc}{section}{Appendix}
\setcounter{tocdepth}{2} %
\startcontents[appendix]
\printcontents[appendix]{l}{1}{\centering \textbf{\MakeUppercase{Contents of Appendix}}}

\newpage
\section{Additional Preliminaries}

This section collects additional definitions and technical tools used throughout the paper, including probabilistic inequalities (\cref{sec:Probabilities Inequalities}), the Poisson approximation technique (\cref{sec:The Poisson Approximation}), and background on the Alias method (\cref{sec:Alias-Method}).

\subsection{Probabilities Inequalities}
\label{sec:Probabilities Inequalities}

\begin{fact}[Basic Composition \citep{DR14}]
    \label{fact:basic-composition}
    Assume that $\cM_1: \cY_1 \rightarrow \cZ_1$ and $\cM_2: \cY_2 \rightarrow \cZ_2$ satisfy $\eps_1$-DP and $\eps_2$-DP, respectively,
    then $\cM=(\cM_1, \cM_2): \cY_1 \times \cY_2 \rightarrow \cZ_1 \times \cZ_2$ satisfies $(\eps_1 + \eps_2)$-DP.
\end{fact}

\begin{fact}[Tail Bound of Discrete Laplace] 
    \label{fact:discrete-laplace-tail}
    Let $X \sim \DiscreteLapNoise{ e^{ - \eps } }$. 
    Then for each $t \in \N_+$, 
    \begin{equation} \label{eq:discrete-laplace-tail}
        \P{ \card{ X } \ge t } = \frac{ 2 \cdot e^{ - \eps \cdot t } }{ 1 + e^{ - \eps } }.
    \end{equation}
\end{fact}

\begin{proof}[Proof of \cref{fact:discrete-laplace-tail}]
    For each $t \ge 1$,
    \begin{align}
        \P{\card{X} \ge t}
            &= 2 \cdot \frac{1 - e^{-\eps}}{1 + e^{-\eps}} \cdot \sum_{i = t}^\infty e^{-\eps \cdot i}
            = 2 \cdot \frac{1 - e^{-\eps}}{1 + e^{-\eps}} \cdot \frac{e^{-\eps \cdot t}}{1 - e^{-\eps}}
            = 2 \cdot \frac{e^{-\eps \cdot t}}{1 + e^{-\eps}}.
    \end{align}
\end{proof}

\begin{definition}[$f$-divergence \citep{LieseV06}]
    Let $\mu, \nu$ be probability measures on a measurable space $\PAREN{\cX, \cA}$, and let $\cF$ be the class of convex functions $f : \PAREN{0, \infty} \to \R$.  
    Suppose that $\mu$ and $\nu$ are dominated by a $\sigma$-finite measure $\lambda$ (e.g., $\lambda = \mu + \nu$), with Radon–Nikodym derivatives $p = \frac{d\mu}{d\lambda}$ and $q = \frac{d\nu}{d\lambda}$.
    For every $f \in \cF$, define the $f$-divergence between $\mu$ and $\nu$ as:
    \begin{equation}
        \fdivergence{\mu}{\nu} = \int_{pq > 0} f \PAREN{ \frac{p}{q} } \, d \nu + f(0) \cdot \nu(p = 0) + f^*(0) \cdot \mu(q = 0),
    \end{equation}
    where 
    \begin{align}
        f(0) \doteq \lim_{t \downarrow 0} f(t) \in (0, \infty], \quad 
        f^*(0) \doteq \lim_{t \rightarrow \infty} \frac{f(t)}{t}, \quad
        f(0) \doteq 0 \doteq 0, \quad 
        f^*(0) \doteq 0 \doteq 0. 
    \end{align}
\end{definition}

\noindent
It can be verified that \emph{the total variation distance} between $\mu$ and $\nu$ corresponds to the $f$-divergence with $f(t) = \card{t - 1}$ for all $t \in \PAREN{0, \infty}$.

\begin{definition}[Stochastic Kernel \citep{LieseV06}]
    Let $\PAREN{\cX, \cA}$ and $\PAREN{\cY, \cB}$ be measurable spaces. 
    A \emph{stochastic kernel} $K : \cX \times \cB \rightarrow \bracket{0, 1}$ is a family of probability measures $\set{ K ( \cdot \mid x ) : x \in \cX }$ on $\PAREN{\cY, \cB}$ such that for every $B \in \cB$, the function $x \mapsto K(B \mid x)$ is $\cA$-measurable.
    Given a probability measure $\mu$ on $\PAREN{\cX, \cA}$, the \emph{measure} $K \mu$ on $\PAREN{\cY, \cB}$ is defined by
    \begin{equation}
        (K \mu)(B) \doteq \int_{\cX} K(B \mid x) \, \mu(dx), \quad \forall B \in \cB.
    \end{equation}
\end{definition}

\begin{proposition}[Data Processing Inequality (see, e.g., Theorem 14 of \cite{LieseV06})]
    \label{prop:data-processing-inequality}
    Let $\PAREN{\cX, \cA}$ and $\PAREN{\cY, \cB}$ be measurable spaces, and let $\mu, \nu$ be probability measures on $\PAREN{\cX, \cA}$. 
    For every convex function $f \in \cF$ and every stochastic kernel $K : \cX \times \cB \rightarrow \bracket{0, 1}$,
    \begin{equation}
        \fdivergence{K \mu}{K \nu} \le \fdivergence{\mu}{\nu}.
    \end{equation}
\end{proposition}

\begin{proposition}[\cite{sendler1975note}, Lemma 2.1; \cite{steerneman1983total}, Theorem 3.1]
    \label{prop:triangle-inequality-for-tv-of-product-measures}
    Let $\PAREN{\Omega, \cF, \mu_k}$ and $\PAREN{\Omega, \cF, \nu_k}$ for $k \in [n]$ be two sequences of probability spaces. Then the total variation distance between the corresponding product measures satisfies
    \begin{equation}
        \tvdistance{\times_{k \in [n]} \mu_k}{\times_{k \in [n]} \nu_k}
        \le \sum_{k \in [n]} \tvdistance{\mu_k}{\nu_k},
    \end{equation}
    where $\times_{k \in [n]} \mu_k$ and $\times_{k \in [n]} \nu_k$ are the product measures on $\PAREN{\Omega^n, \cF^n}$ constructed from $\mu_k$ and $\nu_k$, respectively.
\end{proposition}

Since the joint distribution of independent random variables equals the product of their marginal distributions, we have the following corollary.

\begin{proposition}[Total Variation Subadditivity for Independent Variables]
    \label{prop:tvd-subadditivity}
    Let $X = (X_1, \ldots, X_n)$ and $Y = (Y_1, \ldots, Y_n)$ be tuples of independent random variables. 
    Then
    \begin{equation}
        \tvdistance{X}{Y}
        \le \sum_{k \in [n]} \tvdistance{X_k}{Y_k}.
    \end{equation}
\end{proposition}

\subsection{The Poisson Approximation}
\label{sec:The Poisson Approximation}

When throwing $m$ balls into $d$ bins independently and uniformly at random, the load of any fixed
bin follows the binomial distribution $\Binomial{m}{1/d}$. 
In the regime where $m/d$ is fixed and $m \to \infty$, this binomial distribution converges to a 
Poisson distribution.

\begin{definition}[Poisson Distribution]\label{def:poisson}
A discrete Poisson random variable $X$ with parameter $\lambda > 0$ has probability mass function
\[
  \P{X = j} = e^{-\lambda} \frac{\lambda^{j}}{j!}, 
  \qquad j \in \N.
\]
\end{definition}

More generally, in many settings it is fruitful to approximate the \emph{joint distribution} of the
loads of all bins by assuming that each bin receives an independent Poisson number of balls with
mean $m/d$.

\begin{theorem}[Poisson Approximation Theorem~\citep{MU05}]
    \label{thm:poisson-approximation}
    Suppose $m$ balls are thrown into $d$ bins independently and uniformly at random, and let 
    $X_i$ denote the number of balls in bin~$i$, for $1 \le i \le d$.  
    Let $Y_1,\dots,Y_d$ be independent Poisson random variables, each with mean $m/d$.
    For any nonnegative function $f(x_1,\ldots,x_d)$,
    \[
      \E{ f(X_1,\ldots,X_d) }
      \;\le\;
      e \sqrt{m} \,\E{ f(Y_1,\ldots,Y_d) }.
    \]
    In particular, for any event $\mathcal{E}$, taking $f$ to be its indicator function yields
    \[
      \P{\mathcal{E}\text{ under the exact model}}
      \;\le\;
      e \sqrt{m}\,
      \P{\mathcal{E}\text{ under the Poisson model}}.
    \]
\end{theorem}

We also use the following standard tail bounds for Poisson random variables.

\begin{theorem}[Poisson Tail Bounds~\citep{MU05}]\label{thm:poisson-tails}
    Let $Y \sim \PoissonNoise{\mu}$. For any $0 \le \alpha \le 1$,
    \[
      \P{Y \le \mu (1-\alpha)}
      \;\le\;
      \exp\!\left(-\frac{\alpha^{2}\mu}{2}\right),
      \qquad
      \P{Y \ge \mu (1+\alpha)}
      \;\le\;
      \exp\!\left(-\frac{\alpha^{2}\mu}{3}\right).
    \]
\end{theorem}

\subsection{Alias Method}
\label{sec:Alias-Method}

The Alias method~\citep{Walker77}, denoted by $\algoAlias$, is an efficient array-based sampling algorithm that, given a distribution $\mu = \set{ p_0, \ldots, p_{m - 1} }$ over $\range{m}$, has $O(m)$ initialization time and memory usage, and then draw samples from $\mu$ in $O(1)$ time.  
The pseudocode is given in \cref{algo:Alias-method}.
During the initialization phase, the algorithm constructs two arrays $\vec{a}$ and $\vec{b}$ of length $\bar{m} \ge m$.
The array $\vec{a} \in \range{m}^{\bar{m}}$ is called the \emph{alias array}, and $\vec{b} \in [0, 1]^{\bar{m}}$ stores Bernoulli probabilities. 

\begin{algorithm}[!ht]
    \caption{Alias Method $\algoAlias$}
    \label{algo:Alias-method}
    \begin{algorithmic}[1]
        \Statex \hspace{-6.8mm} {\bf Procedure:}        \textsc{Initialization}
        \Statex \hspace{-4.6mm} {\bf Input:} 
        $m \in \N^+$, $p_0, \ldots, p_{m - 1} \in \paren{0, 1}$
        \State  Construct arrays $\vec{a}$ and $\vec{b}$ as in \cref{prop:alias-method-exact}~\citep{Walker77};

        \vspace{2mm}
        \Statex \hspace{-6.8mm} {\bf Procedure:}        \textsc{Sample}
        \Statex \hspace{-4.8mm} {\bf Output}: Random variable $Z \in \range{m}$
        \State $I \gets \UniformNoise{\range{\bar{m}}}$, $B \gets \Bernoulli{\vec{b}[I]}$ 
        \If{$B = 1$}
            {\bf return} $Z \gets I$
        \EndIf
        \If{$B = 0$}
            {\bf return} $Z \gets \vec{a}[I]$
        \EndIf
    \end{algorithmic}
\end{algorithm}

\subsubsection*{Sampling Procedure ($\algoAliasSample$):}  
To sample from $\mu$, the algorithm proceeds as follows: 
(1) sample a uniform random index $I \in \range{\bar{m}}$,
(2) sample a Bernoulli random variable $B$ that equals $1$ probability $\vec{b}[I]$,
(3) if $B = 1$, return $I$; otherwise, return $\vec{a}[I]$.
That is why $\vec{a}$ is called the alias array: it specifies the alternative index to return when the Bernoulli test fails.

To facilitate uniform sampling over $\range{\bar{m}}$, one typically sets $\bar{m} = 2^{ \ceil{\log_2 m} }$, the smallest power of $2$ not less than $m$.
For indices $i > m$, the Bernoulli probabilities satisfy $\vec{b}[i] = 0$, to ensure the outputs lie in $\range{m}$.

Given this sampling procedure, a natural question is whether there exists a construction of $\vec{a}$ and $\vec{b}$ such that the output distribution of $\algoAliasSample$ exactly matches $\mu$.

\begin{proposition}[\citep{Walker77}]
\label{prop:alias-method-exact}
    Assume that $\vec{b}$ can store real numbers and that we can sample exactly from the Bernoulli distribution $\Bernoulli{q}$ for each $q \in [0, 1]$. Then there exists a construction of $\vec{a}$ and $\vec{b}$ such that $\algoAliasSample$ generates samples exactly according to $\mu$.
\end{proposition}

\noindent
For completeness, we also include the pseudocode as well as the proof of the construction of $\vec{a}$ and $\vec{b}$. 
It is worth noting that \citet{Walker77} assumes real arithmetic, which is infeasible in the Word-RAM model used in this paper. 
In \cref{sec:approx-discrete-laplace}, we relax this assumption and apply the Alias method to approximately sample from a discrete distribution.

\subsubsection*{\bf Alias Method Initialization}
\label{sec:alias-initialization}

We now describe how to initialize the two arrays $\vec{a}$ and $\vec{b}$ used in the Alias method.

\begin{algorithm}[!ht]
    \caption{Alias Method Initialization $\algoAlias$}
    \label{algo:Alias-Method-Initialization}
    \begin{algorithmic}[1]
        \Statex \hspace{-6.8mm} {\bf Procedure:}        \textsc{Initialization}
        \Statex \hspace{-4.6mm} {\bf Input:} 
        $m \in \N^+$, $p_0, \ldots, p_{m - 1} \in \paren{0, 1}$
        \State $\bar{m} \gets 2^{ \ceil{\log_2 m} }$
        \Comment{rounding $m$ up to the nearest power of $2$}
        \State initialize an empty (Alias) array $\Vec{a}$ of size $\bar{m}$
        \State initialize an array $\vec{b}$, s.t., 
        $\vec{b}[i] \gets \begin{cases}
                \bar{m} \cdot p_i,  & i \in \zeroton{m}, \\
                0,                  & m < i < \bar{m}.
        \end{cases}
        $
        \State $S_{> 1} \gets \set{ i \in \zeroton{\bar{m}} : \vec{b}[i] > 1}$, and $S_{\le 1} \gets \zeroton{\bar{m}} \setminus S_{ > 1}$
        \For{$k \in [\bar{m}]$}
            \State Select and remove an arbitrary element $i \in S_{\le 1}$
            \If{$\vec{b}[i] < 1$}
                \State select an arbitrary element $j \in S_{> 1}$
                \State $\Vec{a}[i] \gets j$, and $\vec{b}[j] \gets \vec{b}[j] - \paren{1 - \vec{b}[i]}$
                \If{$m_j \le 1$} 
                    move $j$ from $S_{> 1}$ to $S_{\le 1}$
                \EndIf
            \EndIf
        \EndFor
    \end{algorithmic}
\end{algorithm}

For illustration purposes, assume we have an auxiliary array $\vec{c}$ initialized as $\vec{c}[i] = p_i \cdot \bar{m}$ for all $i \in \range{m}$, and $\vec{c}[i] = 0$ for all $i > m$.

In the ideal case where all entries of $\vec{c}$ are integers, we can allocate and set $\vec{c}[i]$ buckets of the alias array $\vec{a}$ to $i$ for each $i \in \range{m}$.
Then sampling a random entry from $\vec{a}$ follows distribution exactly $\mu$.

To handle the case where $\vec{c}$ has fractional values, we allow an element $i \in \range{\bar{m}}$ to occupy only a fraction of a bucket in $\vec{a}$.
We temporarily allow the array $\vec{a}$ to store two elements in each bucket $\vec{a}[i]$, denoted by $\vec{a}[i].\textsc{first}$ and $\vec{a}[i].\textsc{second}$, so that $\vec{a}[i].\textsc{first}$ occupies $\vec{b}[i] \in [0, 1]$ fraction of $\vec{a}[i]$, and $\vec{a}[i].\textsc{second}$ occupies the remaining $1 - \vec{b}[i]$ fraction.
We later describe how to eliminate the need for storing $\vec{a}[i].\textsc{first}$.  

Assume all buckets of $\vec{a}$ and $\vec{b}$ are initially empty.  
The allocation proceeds as follows.  
We initialize two sets:  
$S_{> 1} = \set{ i \in \zeroton{\bar{m}} : \vec{c}[i] > 1 }$ and  
$S_{\le 1} = \zeroton{\bar{m}} \setminus S_{> 1}$.  
We then iteratively update $\vec{a}$, $\vec{b}$, and $\vec{c}$ while maintaining the following invariants:
\begin{itemize}[leftmargin=0.5cm, label=$\triangleright$]
    \item For all $k \in \range{m}$,
    \begin{equation}
        \vec{c}[k] + \sum_{i \in \bar{m}} \PAREN{
            \vec{b}[i] \cdot \indicator{ \vec{a}[i].\textsc{first} = k } 
            + (1 - \vec{b}[i]) \cdot \indicator{ \vec{a}[i].\textsc{second} = k }
        } = p_k \cdot \bar{m}.
    \end{equation}
    We interpret $\vec{c}[k]$ as the number of fractional buckets still to be assigned to element $k$.  
    The second term in the LHS counts how many have already been assigned to $k$.  
    The invariant holds at initialization, since $\vec{c}[k] = p_k \cdot \bar{m}$ and both $\vec{a}$ and $\vec{b}$ are empty.

    \item If $\vec{c}[k] > 0$, then $\vec{a}[k]$ has not yet been allocated.
\end{itemize}

At each iteration, we remove an arbitrary $i \in S_{\le 1}$.
If $\vec{c}[i] = 1$, we fully assign the bucket $\vec{a}[i]$ to $i$ by setting $\vec{a}[i].\textsc{first} = \vec{a}[i].\textsc{second} = i$ and $\vec{b}[i] = \vec{c}[i] = 1$.
If $\vec{c}[i] < 1$, then there exists some $j \in S_{> 1}$.
We assign a $\vec{c}[i]$ fraction of the bucket $\vec{a}[i]$ to $i$, and the remaining $1 - \vec{c}[i]$ fraction to $j$, by setting both $\vec{a}[i].\textsc{first} = i$, $\vec{a}[i].\textsc{second} = j$, and $\vec{b}[i] = \vec{c}[i]$.
We then decrement $\vec{c}[j]$ by $1 - \vec{b}[i]$; if $\vec{c}[j] \le 1$, we move $j$ from $S_{> 1}$ to $S_{\le 1}$.
It is straightforward to verify that both invariants remain true after each iteration.  

Additionally, we observe that either the size of $S_{\le 1}$ or $S_{> 1}$ decreases by $1$ in each iteration.  
Since their total size is $\bar{m}$, the allocation process terminates in $\bar{m}$ iterations.  

\noindent
\textit{Optimization.}
Since the construction always sets $\vec{a}[i].\textsc{first} = i$, we don't need to store it explicitly.
Instead, we store only $\vec{a}[i].\textsc{second}$ in $\vec{a}[i]$, and use $\vec{b}[i]$ to decide whether to return $i$ or $\vec{a}[i]$ during sampling.  
Furthermore, the arrays $\vec{b}$ and $\vec{c}$ can be merged into one without affecting correctness. 
Together, these lead to a compact and efficient construction algorithm.

\newpage
\section{Efficient Private Noise Samplers}
\label{sec:private-noise-samplers}

In this section, we present, in \cref{sec:Approximate Time-Oblivious Sampler Framework}, a general framework for deterministic-time approximate samplers for arbitrary discrete distributions, and, in \cref{subsec:purified-approximate-dl-sampler}, the discrete Laplace sampler used in our private sparse histogram algorithm from \cref{sec:dp-sparse-histogram}.

This deterministic-time sampler belongs to a broader class known as \emph{time-oblivious samplers}, introduced by \citet*{DovDNT23} and formally defined below. 
These samplers guarantee that the runtime of generating each sample reveals little about the sampled value. 
While our sparse histogram algorithm strictly only requires that the total runtime of sampling all noise values reveals a controlled amount of information about the input histogram, per-sample time-obliviousness is desirable for scenarios vulnerable to timing side channels, such as secure multiparty computation.

The framework for designing time-oblivious samplers for arbitrary discrete distributions (\cref{thm:time-oblivious-distribution-sampler})—based on the Alias method (see \cref{algo:Alias-method} in \cref{sec:Alias-Method})—is the core of this section. It strengthens the results of \citet{DovDNT23}, supports a wide range of differentially private noise distributions, and is of independent interest.

Our discrete Laplace sampler (\cref{thm:purified-approximate-discrete-laplace-sampler}) is constructed in two steps: first, we design a time-oblivious sampler (\cref{thm:approximate-discrete-laplace-sampler-formal}) that closely approximates the discrete Laplace distribution; then, we apply a purification technique (\cref{proposition:purification}) to upgrade its approximate DP guarantee to a pure one.

\subsubsection*{\bf Time Oblivious Sampler}
Given a probability space $\paren{\cZ, \cF, \mu}$, a sampler $\cA_\mu$ for $\mu$ is an algorithm that takes as input a countably infinite sequence of unbiased random bits and outputs a sample in $\cZ$ that (approximately) follows the distribution $\mu$.
Importantly, $\cA_\mu$ need not consume the entire input sequence before halting.
Mathematically, $\cA_\mu$ can be viewed as a measurable function from $\set{0,1}^\N$—the space of infinite binary sequences—to $\cZ$.
Informally, a sampler is said to be (approximate) \emph{time-oblivious} if observing its running time does not significantly improve one’s ability to infer the output, compared to guessing according to the target distribution $\mu$.

\begin{definition}[$(\eps, \delta)$-Approximate Time-Oblivious Sampler {\citep{DovDNT23}}]
\label{def:approx-time-oblivious-sampler}
    Let $\paren{\cZ, \cF, \mu}$ be a probability space.  
    An algorithm $\cA_\mu$ is an $(\eps, \delta)$-approximate time-oblivious sampler for $\mu$ if, 
    for every $T \subseteq \N$ such that $\P{\vec{s} \sim \UniformNoise{ \set{0, 1}^\N } }{\runningtime{\cA_\mu \paren{\vec{s}}} \in T} > 0$, 
    it holds that
    \begin{equation}
        e^{-\eps} \paren{ \mu \paren{ E } - \delta }
        \le 
        \P{\vec{s} \sim \UniformNoise{ \set{0, 1}^\N } }{ \cA_\mu \paren{\vec{s}} \in E \mid \runningtime{\cA_\mu \paren{\vec{s}}} \in T }
        \le 
        e^\eps \cdot \mu \paren{ E } + \delta,
        \quad
        \forall E \in \cF,
    \end{equation}
    where $\runningtime{\cA_\mu \paren{\vec{s}}}$ denotes the running time of $\cA_\mu$ on input $\vec{s}$.
\end{definition}

\noindent
\textit{Remark.} 
The original definition in \citet{DovDNT23} adopts a slightly different computational model, where the running time $\runningtime{\cA_\mu(\vec{s})}$ is defined as the number of random bits read by $\cA_\mu$ before halting.
In contrast, throughout this paper we measure running time in terms of the number of word-level operations, as specified by the word-RAM model.

When $\eps = \delta = 0$, the algorithm $\cA_\mu$ is referred to as a (fully) \emph{time-oblivious generating algorithm} for $\mu$~\citep{DovDNT23}.
However, the class of distributions that admit such samplers is quite limited: $\mu$ must have finite support, and all of its probabilities must be rational~\citep{DovDNT23}.
This motivates the need for approximate relaxations in \cref{def:approx-time-oblivious-sampler}.

\subsection{Approximate Time-Oblivious Sampler Framework}
\label{sec:Approximate Time-Oblivious Sampler Framework}

First, we develop a general framework for constructing approximate time-oblivious samplers
for arbitrary discrete distributions, not just the discrete Laplace.
The pseudocode for the framework is presented in \cref{algo:finite-Alias-method}.  
Our design is a direct instantiation of the Alias method (\cref{algo:Alias-method}), implemented truncated support and fixed point representation of probabilities.

\begin{algorithm}[!h]
    \caption{Finite Alias Method $\algoFiniteAlias$}
    \label{algo:finite-Alias-method}
    \begin{algorithmic}[1]
        \Statex \hspace{-6.8mm} {\bf Procedure:}        \textsc{Initialization}
        \Statex \hspace{-4.6mm} {\bf Input:} 
        distribution $\mu$ with discrete support, $\delta \in \paren{0, 1}$
        \State $m \gets \card{\CoreSupportSet{\delta / 2}}$, $\ell \gets \ceil{ \log_2 \paren{ 2 / \delta }  +  \log_2 m }$
        \State $\forall \, i \in \range{m} \, $: $x_i \gets \LabelOracle \paren{i}$
        \State $\forall \, i \in \range{m} \, $: 
        $p_i \gets \ProbOracle \paren{ \mu \paren{ x_i } , \ell }$ 
        \State $\algoAliasInitialization{m, p_0, \ldots, p_{m - 1}}$

        \vspace{2mm}
        \Statex \hspace{-6.8mm} {\bf Procedure:}        \textsc{Sample}
        \Statex \hspace{-4.8mm} {\bf Output}: Random variable $Z \in \CoreSupportSet{\delta / 2}$
        \State $Z \gets \algoAliasSample{}$
        \State {\bf return} $\LabelOracle \paren{ Z }$
    \end{algorithmic}
\end{algorithm}

\begin{theorem}[$(0, \delta)$-Approximate Time-Oblivious Sampler]
\label{thm:time-oblivious-distribution-sampler}
    Let $\mu$ be a distribution with discrete support, let $\delta \in (0, 1)$, and let $\CoreSupportSet{\delta / 2}$ denote a subset of minimum size such that $\mu \PAREN{ \CoreSupportSet{\delta / 2} } \ge 1 - \delta / 2$.
    Assume access to the following two oracles:
    \begin{itemize}[label=$\triangleright$, leftmargin=0.7cm]
        \item The \emph{labeling oracle} $\LabelOracle$, which assigns labels to the elements in $\CoreSupportSet{\delta / 2}$ from $0$ to $m - 1$. For each $i \in \range{m}$, $\LabelOracle(i)$ returns the corresponding element.
        
        \item The \emph{binary probability oracle} $\ProbOracle$, which, given $x \in \CoreSupportSet{\delta / 2}$ and $\ell \in \N_+$, returns the $\ell$-bit binary expansion of $\mu(x)$ after the fractional point, denoted by $\ProbOracle( \mu(x), \ell)$.
    \end{itemize}
    Then there exists a $(0, \delta)$-approximate time-oblivious sampler for $\mu$ with the following properties
    \begin{itemize}[label=$\triangleright$, leftmargin=0.7cm]
        \item It uses 
        $2^{\ceil{ \log_2 \card{ \CoreSupportSet{\delta / 2} } }} \cdot \big( \ceil{ \log_2 \frac{2}{\delta} } + \ceil{ \log_2 \card{ \CoreSupportSet{\delta / 2} } } \big)$ bits of memory, 
        which is in $O \paren{ \card{ \CoreSupportSet{\delta / 2} } \cdot \paren{ \ln \frac{1}{\delta} + \ln \card{ \CoreSupportSet{\delta / 2} } } }$.
        \item Each sample consumes at most $\log_2 \card{ \CoreSupportSet{\delta / 2} } + \log_2 \frac{1}{\delta} + O(1)$ unbiased random bits, which runs in $O \paren{ \frac{1}{\omega} \cdot \paren{ \log_2 \card{ \CoreSupportSet{\delta / 2} } + \log_2 \frac{1}{\delta} } }$ time in word-RAM model.
    \end{itemize}
\end{theorem}

\subsubsection*{Remark.}
Our $(0, \delta)$-approximate time-oblivious sampler achieves the same random bit complexity as the sampler of \citet*{DovDNT23} (Claim~2.22), but reduces the space usage from 
$O \paren{ \card{ \CoreSupportSet{ \delta / 2 } } / \delta \cdot \ln \card{ \CoreSupportSet{ \delta / 2 } } }$ 
to 
$O \paren{ \card{ \CoreSupportSet{ \delta / 2 } } \cdot \paren{ \ln \frac{1}{\delta} + \ln \card{ \CoreSupportSet{ \delta / 2 } } } }$ 
bits, representing an exponential improvement in the dependence on~$\delta$.

At a high level, their construction proceeds by encoding the probabilities of elements in $\CoreSupportSet{ \delta / 2 }$ to $\log_2 (\card{ \CoreSupportSet{ \delta / 2 } } / \delta)$ bits of binary precision.  
They then construct an array of size $\card{ \CoreSupportSet{ \delta / 2 } } / \delta$, and assign to each element in $\CoreSupportSet{ \delta / 2 }$ a number of buckets proportional to its truncated probability mass.  
Sampling amounts to drawing a uniform random bucket and returning the corresponding element.

The running time and memory overhead of $\LabelOracle$ and $\ProbOracle$ depend on the input distribution $\mu$, which we will further discuss when applying this framework to specific distributions.
Here, we assume the existence of $\LabelOracle$ and $\ProbOracle$, as is also implicitly assumed by \citet{DovDNT23}, and we focus on analyzing the running time and memory usage of the general framework itself.

\begin{proof}[Proof of \cref{thm:time-oblivious-distribution-sampler}]

    Denote $m \doteq \card{\CoreSupportSet{\delta / 2}}$ and $\ell \doteq \ceil{ \log_2 \frac{2}{\delta} + \log_2 m }$. 
    For each $i \in \range{m}$, let $x_i \doteq \LabelOracle \paren{i} \in \CoreSupportSet{\delta / 2}$.
    Since $\mu$ has discrete support, we can list the remaining elements as $x_{m}, x_{m + 1}, \ldots$.
    In particular, if $\mu$ has finite support, we set $x_i \doteq \varnothing$ for all $i \ge \card{ \supp \paren{ \mu } }$.
    
    For each $i \in \range{m}$, let $p_i \doteq \ProbOracle \paren{ \mu \paren{ x_i }, \ell }$ denote the $\ell$-bit binary approximation of $\mu \paren{ x_i }$ after the fractional point, so that $0 \le \mu \paren{ x_i } - p_i \le 2^{ - \ell }$. 
    Finally, adjust $p_0$ to ensure $p_0 \doteq 1 - \sum_{i = 1}^{m - 1} p_i$, so that the $p_i$'s sum to $1$.
    Hence 
    \begin{align*}
        \card{p_0 - \mu(x_0)}
            &= \CARD{1 - \sum_{i = 1}^{m - 1} p_i - \Biggparen{1 - \sum_{i \in \N_+} \mu(x_i)}} \\
            &= \CARD{ \sum_{i \ge m} \mu(x_i) + \sum_{i = 1}^{m - 1} \PAREN{ \mu(x_i) - p_i } } \\
            &= \sum_{i \ge m} \mu(x_i) + \sum_{i = 1}^{m - 1} \PAREN{ \mu(x_i) - p_i } \\
            &\le \mu \paren{ \overline{\CoreSupportSet{\delta / 2}} }  + \paren{m - 1} \cdot 2^{-\ell}.
    \end{align*}
    
    It follows that 
    \begin{align*}
        \tvdistance{ \set{ p_i } }{\mu} 
            &\le \frac{1}{2} \cdot \Biggparen{ 
                \mu \PAREN{ \overline{\CoreSupportSet{\delta / 2}} } + \card{p_0 - \mu(x_0)} + \sum_{i \in [m - 1]} \card{p_i - \mu(x_i)}
            } \\
            &\le \frac{1}{2} \cdot \Biggparen{ 
                \mu \PAREN{ \overline{\CoreSupportSet{\delta / 2}} } + \card{p_0 - \mu(x_0)} + \paren{m - 1} \cdot 2^{-\ell}
            } \\
            &\le \mu \PAREN{ \overline{\CoreSupportSet{\delta / 2}} } + \paren{m - 1} \cdot 2^{-\ell} \\
            &\le \delta / 2 + \delta / 2 \\
            &\le \delta.
    \end{align*}

    Next, we describe how to implement an instance of Alias method (Algorithm~\ref{algo:Alias-method}) with input $\paren{ m, \set{ p_i } }$ to sample exactly from the distribution $\set{ p_i }$. 
    Since Alias method has a worst-case running time bound, this directly yields a $(0, \delta)$-approximate time-oblivious sampler for $\mu$. 
    We focus particularly on the space usage and the number of unbiased random bits required to generate a sample.
    
    \subsubsection*{Space Usage.} 
    Clearly, the array $\vec{a}$ requires $\bar{m} \cdot \log_2 m$ bits of space, where $\bar{m} = 2^{ \ceil{ \log_2 m } }$.
    It remains to analyze the space usage of the array $\vec{b}$.
    
    First, since $\bar{m}$ is a power of two, the quantity $\bar{m} \cdot p_i$ corresponds to a left shift of $p_i$ by $\log_2 \bar{m}$ bits. 
    Hence, during the initialization stage, we can represent each entry of $\vec{b}$ using a fixed-length binary representation with exactly $\ell$ bits—comprising $\log_2 \bar{m}$ bits before and $\ell - \log_2 \bar{m}$ bits after the fractional point—without introducing rounding or truncation errors.
    
    Moreover, an intermediate operation $\vec{b}[j] \gets \vec{b}[j] - \paren{ 1 - \vec{b}[i] }$ is performed by Algorithm~\ref{algo:Alias-method} only when $\vec{b}[j] > 1$, after which $\vec{b}[j]$ decreases but remains nonnegative. 
    Therefore, this operation also does not introduce rounding or truncation errors.
    
    Finally, when the algorithm terminates, each entry of $\vec{b}$ lies in the interval $(0, 1]$. 
    At this point, we can discard the $\log_2 \bar{m}$ bits before the fractional point and retain only the $\ell - \log_2 \bar{m}$ fractional bits. 
    There is one caveat: if $\vec{b}[i] = 1$ for some $i$, this truncation would result in $\vec{b}[i] = 0$. 
    In this case, we can simply set $\vec{a}[i] = i$ to preserve the correct sampling probability for sampling item $i$.
    
    \subsubsection*{Sampling Bit.}
    Sampling an $I$ uniformly at random from $\bar{m}$ requires exactly $\log_2 \bar{m}$ bits. 
    Further, since $\vec{b}[I]$ has a binary representation of $\ell - \log_2 \bar{m}$ bit after the fractional point, sampling $\Bernoulli{\vec{b}[I]}$ can be achieved using less $2$ unbiased random bit in expectation and $\ell - \log_2 \bar{m}$ in the worst case.

\end{proof}

\subsection{Discrete Laplace Sampler}
\label{subsec:purified-approximate-dl-sampler}

The properties of the discrete Laplace sampler used in our private sparse histogram algorithm (\cref{algo:tail-item-padding}) are stated below.

\begin{theorem}[Purified Approximate Discrete Laplace Sampler]
\label{thm:purified-approximate-discrete-laplace-sampler}
    Let $n \in \N^+$, $\eps \in \Q_+$, and $\gamma \in \Q_+ \cap (0, 1)$.
    There exists a randomized algorithm $\PurifiedApproxDiscreteLaplaceMechanism{n}{\eps}{\gamma} : \zeroton{n} \rightarrow \zeroton{n}$ with the following properties:
    \begin{itemize}[label=$\triangleright$, leftmargin=0.7cm]
        \item For each $t \in [n]$, 
        $\PurifiedApproxDiscreteLaplaceMechanism{n}{\eps}{\gamma}{t - 1}$ and $\PurifiedApproxDiscreteLaplaceMechanism{n}{\eps}{\gamma}{t}$ are $(\eps, 0)$-indistinguishable.
        
        \item For each $t \in \zeroton{n}$, $\beta > 2 \cdot \gamma$, 
        $\PurifiedApproxDiscreteLaplaceMechanism{n}{\eps}{\gamma}{t}$ is $(\alpha, \beta)$-accurate estimator of $t$,
        for
        $\alpha \doteq \ceil{ \frac{1}{\eps} \cdot \ln \frac{2}{\beta - \frac{2 + n}{1 + n} \cdot \gamma} } \in O \paren{ \frac{1}{\eps} \cdot \ln \frac{1}{\beta} }$

        \item It has initialization time 
        $O\bigl( \frac{1}{\eps} \cdot \poly{ \ln \frac{1}{\eps}, \ln n, \ln \frac{1}{\gamma} } \bigr)$ 
        and memory usage 
        $O\bigl( \bigl( \frac{1}{\eps} + \ln \frac{1}{\gamma} + \ln n \bigr) 
                \cdot \bigl( \ln \frac{1}{\eps} + \ln \frac{1}{\gamma} + \ln n \bigr) \bigr)$ 
        bits.

        \item After initialization, for each $t \in \zeroton{n}$,  
        $\PurifiedApproxDiscreteLaplaceMechanism{n}{\eps}{\gamma}{t}$ has worst-case running time 
        \\
        $O \bigparen{ \frac{1}{\omega} \cdot \bigparen{ \ln \frac{1}{\eps} + \ln \frac{1}{\gamma} + \ln n }}$, 
        where $\omega$ is the machine word size. 
        
    \end{itemize}
\end{theorem}

\subsubsection*{Remark}
The first two properties in \cref{thm:purified-approximate-discrete-laplace-sampler} state that $\PurifiedApproxDiscreteLaplaceMechanism{n}{\eps}{\gamma}$ matches the privacy guarantee and closely approximates the tail bound of the standard sampler $\cM_{\DiscreteLapNoise{e^{-\eps}}}$, which creates shifted discrete Laplace noise: $\cM_{\DiscreteLapNoise{e^{-\eps}}}(t) = t + \DiscreteLapNoise{e^{-\eps}}$ for any $t \in \zeroton{n}$.

Recall that \cref{thm:purified-approximate-discrete-laplace-sampler-informal} is a simplified version of \cref{thm:purified-approximate-discrete-laplace-sampler}.
When $\eps \doteq \frac{a_\eps}{b_\eps}$ and $\delta \doteq \frac{a_\delta}{b_\delta}$ for $a_\eps, b_\eps, a_\delta, b_\delta \in \N_+$ that fit in a constant number of machine words, it holds that $\ln \frac{1}{\delta}, \ln \frac{1}{\eps} \in O(\omega)$.  
Hence, the memory usage becomes $O \left( \frac{1}{\eps} + \ln \frac{1}{\delta} \right)$ words, and the sampling time becomes $O(1)$.

To prove the theorem, it suffices—by the purification technique (\cref{proposition:purification}), which converts approximate DP to pure DP—to construct an efficient time-oblivious sampler whose output distribution is close to the discrete Laplace distribution.

\begin{theorem}[Approximate Discrete Laplace Sampler]
    \label{thm:approximate-discrete-laplace-sampler-formal}
    Let $\eps \in \Q_+$ and $\delta \in \Q_+ \cap \paren{0, 1}$.
    There exists a $(0, \delta)$-approximate time-oblivious sampler, denoted by $\DLapSampler{\varepsilon}{\delta}$, for the discrete Laplace distribution $\DiscreteLapNoise{e^{-\varepsilon}}$, with the following properties:
    \begin{itemize}[label=$\triangleright$, leftmargin=0.7cm]
        \item The sampler has initialization time 
        $O \bigparen{ \frac{1}{\eps} \cdot \poly{ \ln \frac{1}{\eps}, \ln \frac{1}{\delta} } }$,
        and memory usage 
        $O \big( \bigparen{ \frac{1}{\eps} + \ln \frac{1}{\delta} } \cdot \big( \ln \frac{1}{\delta} + \ln \frac{1}{\eps} \big) \big)$ bits, 
        \item Each sample consumes at most
        $\log_2 \bigparen{ \frac{1}{\eps} \cdot \ln \frac{1}{\delta} } + 2 \cdot \log_2 \frac{1}{\delta} + O(1)$ 
        unbiased random bits, which runs in 
        $O \bigparen{ \frac{1}{\omega} \cdot \bigparen{ \ln \frac{1}{\eps} + \ln \frac{1}{\delta} } }$ 
        time in the word-RAM model, where $\omega$ is the machine word size.
    \end{itemize}
\end{theorem}

Compared to prior deterministic-time samplers with comparable accuracy, \cref{thm:approximate-discrete-laplace-sampler-formal} achieves the best-known per-sample running time, and matches the batch running time of the previous state-of-the-art, as summarized in \cref{tab:my-table}.

\begin{table}[H]
\renewcommand{\arraystretch}{1.5} %
\centering
\resizebox{\columnwidth}{!}{%
\begin{tabular}{|c|c|c|c|c|}
\hline
\textbf{Method} & \textbf{Preprocessing Time}                   & \textbf{Memory}                                                                                                                 & \textbf{Time Per Sample}                          & \textbf{Time for $m$ Samples}                      \\ \hline
\citet{DworkKMMN06}        & $\tilde{O} \bigparen{ 1 }$          & $O \bigparen{ \ln \frac{1}{\eps} + \ln \ln \frac{1}{\delta} }$                                  & $O \bigparen{ \ln \frac{1}{\eps} + \ln \ln \frac{1}{\delta} }$  & $O \bigparen{ m \cdot ( \ln \frac{1}{\eps} + \ln \ln \frac{1}{\delta} ) }$ \\ \hline
\citet{BalcerV19} & $\tilde{O} \bigparen{ \frac{1}{\eps} }$ & $O \big( \paren{ \frac{1}{\eps} \cdot \ln \frac{1}{\delta} }^2 \big)$ & $O( \frac{1}{\eps} \cdot \ln \frac{1}{\delta} \cdot \ln n )$ & $O(\frac{m}{\eps} \cdot \ln \frac{1}{\delta} \cdot \ln n)$ \\ \hline
\citet{Canonne0S20}        & $O(1)$                                        & $O(1)$                                                                                                                           & $O \bigparen{ \ln \frac{1}{\delta} }$                     & $O \bigparen{ m + \ln \frac{1}{\delta} }$                      \\ \hline
\cref{thm:approximate-discrete-laplace-sampler-formal}        & $\tilde{O} \bigparen{ \frac{1}{\eps} }$          & $O \bigparen{ \frac{ 1 }{\eps} + \ln \frac{1}{\delta} } $                                       & {$O(1)$}                                  & {$O(m)$}                                        \\ \hline
\end{tabular}%
}
\caption{
    Comparison of approximate discrete Laplace samplers within statistical distance $\delta$ of $\DiscreteLapNoise{e^{-\eps}}$, where $\eps$ and $\delta$ are rational numbers whose reduced fractional representations fit within a constant number of machine words.  
    The bounds for \citet{Canonne0S20} are obtained by imposing a time limit on their sampler.  
    Per-sample running time may be of independent interest, especially when adapting samplers to scenarios where individual sampling time could leak information, such as in distributed settings.
}
\label{tab:my-table}
\vspace{-5mm}
\end{table}

The proof of \cref{thm:approximate-discrete-laplace-sampler-formal} is provided in \cref{sec:approx-discrete-laplace}. 
Using this result, we complete the proof of \cref{thm:purified-approximate-discrete-laplace-sampler}.

\begin{proof}[Proof of \cref{thm:purified-approximate-discrete-laplace-sampler}]
    The $\PurifiedApproxDiscreteLaplaceMechanism{n}{\eps}{\gamma}$ is constructed as follows: 
    \begin{enumerate}[leftmargin=0.7cm]
        \item Initialize an instance of $\DLapSampler{\eps}{\delta}$ stated in \cref{thm:approximate-discrete-laplace-sampler-formal}, with $\delta = \frac{e^\eps - 1}{e^\eps + 1} \cdot \frac{\gamma}{1 - \gamma} \cdot \frac{1}{1 + n}$.
        
        \item $\cM' (t): \zeroton{n} \rightarrow \zeroton{n}$ be defined as $\cM(t) \doteq \clamp{t + X}{0}{n}$, where $X$ is sampled by $\DLapSampler{\eps}{\delta}$.
        It holds that $\tvdistance{X}{\DiscreteLapNoise{e^{-\eps}}} \le \delta$.

        \item Let $\PurifiedApproxDiscreteLaplaceMechanism{n}{\eps}{\gamma}{t} = \algoPurify{\cM', t, \gamma}$, for each $t \in \zeroton{n}$, where $\algoPurify$ is given in \cref{alg:purification}.
    \end{enumerate}

    \noindent
    Since $O \paren{ \ln \frac{1}{\delta} } = O \paren{ \ln \frac{1 + e^{-\eps}}{1 - e^{-\eps}} + \ln \frac{1}{\gamma} + \ln n}$, by \cref{thm:approximate-discrete-laplace-sampler-formal}, initializing $\DLapSampler{\eps}{\delta}$ takes time
    $
        O \bigparen{ \frac{1}{\eps} \cdot \poly{ \ln \frac{1}{\eps}, \ln n, \ln \frac{1}{\gamma} } }
    $ 
    and memory usage
    $O \bigparen{ \bigparen{ \frac{1}{\eps} + \ln \frac{1}{\gamma} + \ln n } \cdot \bigparen{ \ln \frac{1}{\gamma} + \ln n + \ln \frac{1}{\eps} } }$.
    After initialization, by \cref{proposition:purification,thm:approximate-discrete-laplace-sampler-formal}, 
    $\PurifiedApproxDiscreteLaplaceMechanism{n}{\eps}{\gamma}{t}$ 
    can be eventuated in time $O \bigparen{ \frac{1}{\omega} \cdot \bigparen{ \ln \frac{1}{\eps} + \ln \frac{1}{\gamma} + \ln n }}$.

    \subsubsection*{Privacy Guarantee.}   
    Let $\cM (t): \zeroton{n} \rightarrow \zeroton{n}$ be defined as 
    $\cM(t) \doteq \clamp{t + \DiscreteLapNoise{e^{\eps}}}{0}{n}$, which adds discrete Laplace noise to $t$ then clamp it to the range of $\zeroton{n}$.
    It satisfies $\eps$-DP: for each $t \in [n]$, $\cM \paren{t - 1}$ and $\cM \paren{t}$ are $(\eps, 0)$-indistinguishable.
    By data processing inequality (\cref{prop:data-processing-inequality}), we have
    $$\tvdistance{\cM(t)}{\cM'(t)} \le \tvdistance{X}{\DiscreteLapNoise{e^{-\eps}}} \le \delta.$$

    \subsubsection*{Utility Guarantee.}
    Since $\tvdistance{X}{\DiscreteLapNoise{e^{-\eps}}} \le \delta$,
    for each $r \in \N^+$, based on the tail bound of discrete Laplace distribution (\cref{fact:discrete-laplace-tail}), 
    \begin{align*}
        \P{ X \ge r }
            &\le \P{
                \card{ \DiscreteLapNoise{e^{-\eps}} } \ge r
            } + \delta 
            = \frac{ 2 \cdot e^{ - \eps \cdot r } }{ 1 + e^{ - \eps } }
            + \delta
            = \frac{ 2 \cdot e^{ - \eps \cdot r } }{ 1 + e^{ - \eps } } + \frac{e^\eps - 1}{e^\eps + 1} \cdot \frac{\gamma}{1 - \gamma} \cdot \frac{1}{1 + n}.
    \end{align*}
    Therefore, 
    \begin{align*}
        \P{
            \card{
                \cM''\paren{t} - t
            } \ge r
        } 
        &\le \gamma + \PAREN{1 - \gamma} \cdot \P{ \card{ \cM'\paren{t} - t} \ge r }   \\
        &\le \gamma + \PAREN{1 - \gamma} \cdot \PAREN{ \P{ \card{ X } > r }  + \delta } 
        \le  \gamma +  \frac{e^\eps - 1}{e^\eps + 1} \cdot \gamma \cdot \frac{1}{1 + n} + \frac{ 2 \cdot e^{ - \eps \cdot r } }{ 1 + e^{ - \eps } }.
    \end{align*}
    In order to bound the last term with $\beta$, we need 
    $
        \bigparen{1 + \frac{1 - e^{ - \eps }}{1 + n}} \cdot \gamma  + 2 \cdot e^{ - \eps \cdot r }
            \le \beta \cdot \paren{ 1 + e^{ - \eps } }.
    $
    It suffices to take 
    $r \ge \frac{1}{\eps} \cdot \ln \frac{2}{ \beta - \PAREN{\frac{2 + n}{1 + n}} \cdot \gamma}.$

\end{proof}

\subsection{Approximate Discrete Laplace Sampler}
\label{sec:approx-discrete-laplace}

In this subsection, we prove \cref{thm:approximate-discrete-laplace-sampler-formal}.  
The proof proceeds in two steps.

First, we instantiate the general sampling framework (\cref{thm:time-oblivious-distribution-sampler}) for the discrete Laplace distribution, obtaining a sampler that satisfies all the guarantees in \cref{thm:approximate-discrete-laplace-sampler-formal} except for the memory usage.

Second, to reduce space complexity, we reduce discrete Laplace sampling to geometric sampling. 
We leverage the fact that the geometric distribution has identical conditional distributions over intervals of equal length, which allows us to decompose a geometric sample into the sum of two smaller geometric components. 
Each component can then be approximately sampled using our general framework over smaller supports, further reducing the overall space usage.

\subsubsection*{\bf Instantiating the Framework for Discrete Laplace}

To apply \cref{thm:time-oblivious-distribution-sampler} directly to the discrete Laplace distribution $\DiscreteLapNoise{e^{-\eps}}$, we identify the core support set $\CoreSupportSet{\delta / 2}$ and specify the oracles $\LabelOracle$ and $\ProbOracle$. 
    
\subsubsection*{Core Support Set $\CoreSupportSet{\delta / 2}$.}
Due to the symmetric decay of the distribution's probability mass from the center, the tail bound for the discrete Laplace distribution (\cref{fact:discrete-laplace-tail}) implies that 
$\CoreSupportSet{\delta / 2} = \IntSet{-L}{R}$, where 
$L = R = \ceil{ \frac{1}{\eps} \ln \frac{4}{(1 + e^{-\eps}) \delta} }$.

\subsubsection*{Labeling Oracle $\LabelOracle$.}
For each $i \in \CoreSupportSet{\delta / 2}$, we define its label as $2 \cdot \card{i} - \indicator{i > 0}$. 
Hence, $\LabelOracle(0) = 0$, $\LabelOracle(1) = 1$, $\LabelOracle(2) = -1$, and so forth.

\subsubsection*{Probability Oracle $\ProbOracle$.}
We state an additional lemma, derived from standard results in numerical computation~\citep{Brent_Zimmermann_2010, harris2020array}.
\begin{lemma}[Binary Expansion of Discrete Laplace Probability]
    \label{lem:binary-expansion-discrete-laplace}
    Given $\eps \in \Q_+$, $t \in \Z$, and $\ell \in \N^+$, the binary expansion of 
    $\P{ \DiscreteLapNoise{e^{-\eps}} = t } = \frac{1 - e^{-\eps}}{1 + e^{-\eps}} \cdot e^{ -\eps \cdot \card{t} }$ 
    up to $\ell$ bits after the fractional point, denoted by 
    $\BinaryExp{ \P{ \DiscreteLapNoise{e^{-\eps}} = t } }{ \ell }$, 
    can be computed in $O \PAREN{ \poly{ \ln \card{t}, \ell } }$ time.
\end{lemma}

By \cref{lem:binary-expansion-discrete-laplace}, given 
$\ell = \ceil{ \log_2 \bigparen{ \frac{2}{\delta} } + \log_2 \card{ \CoreSupportSet{\delta / 2} } }$ and $t \in \CoreSupportSet{\delta / 2}$, 
the $\ell$-bit binary expansion 
$\BinaryExp{\P{\DiscreteLapNoise{e^{-\eps}} = t}}{\ell}$ 
can be computed in time 
$O \paren{ \poly{ \ln t, \ell } } 
= O \big( \textsc{poly} ( \ln \frac{1}{\eps} , \ln \frac{1}{\delta} ) \big)$.
Consequently, $\ProbOracle$ can be implemented with total time 
$O \bigparen{ \frac{1}{\eps} \cdot \ln \frac{1}{\delta} \cdot \poly{ \ln \frac{1}{\eps}, \ln \frac{1}{\delta} } } 
= O \big( \frac{1}{\eps} \cdot \textsc{poly} ( \ln \frac{1}{\eps}, \ln \frac{1}{\delta} ) \big)$, 
and requires an array of 
$O \bigparen{ \frac{1}{\eps} \cdot \ln \frac{1}{\delta} \cdot \ell } = 
O \bigparen{ \frac{1}{\eps} \cdot \ln \frac{1}{\delta} \cdot \bigparen{ \ln \frac{1}{\delta} + \ln \frac{1}{\eps} } }$ bits to store the computed values.

\begin{corollary}[Time-Oblivious Sampler for Discrete Laplace Distribution] 
    \label{thm:time-oblivious-discrete-laplace}
    Let $\eps \in \Q_+$ and $\delta \in \Q_+ \cap \paren{0, 1}$.
    Then, there exists a $(0, \delta)$-approximate time-oblivious sampler for $\DiscreteLapNoise{e^{-\eps}}$ with the following properties
    \begin{itemize}[label=$\triangleright$, leftmargin=0.7cm]
        \item It has pre-computation time 
        $O \PAREN{ \frac{1}{\eps} \cdot \poly{ \ln \frac{1}{\eps}, \ln \frac{1}{\delta} } }$, 
        and uses 
        $O \PAREN{  \frac{1}{\eps} \cdot \PAREN{ \ln \frac{1}{\delta} } \cdot \PAREN{ \ln \frac{1}{\delta} + \ln \frac{1}{\eps} } }$ bits of memory.
        \item Each sampler consumes at most $\log_2 \PAREN{ \frac{1}{\eps} \cdot \ln \frac{1}{\delta} } + \log_2 \frac{1}{\delta} + O(1)$ unbiased random bits, which runs in 
        $O \PAREN{ \frac{1}{\omega} \cdot \PAREN{ \ln \frac{1}{\eps} + \ln \frac{1}{\delta} } }$ time in the worst case.
    \end{itemize}
\end{corollary}

\subsubsection*{\bf Memory-Efficient Construction}

In this subsection, for $\eps \in \paren{0, 1}$, we show how to reduce the 
$O \bigparen{ \frac{1}{\eps} \cdot \ln \frac{1}{\delta} \cdot \bigparen{ \ln \frac{1}{\delta} + \ln \frac{1}{\eps} } }$ 
bit memory usage in \cref{thm:time-oblivious-discrete-laplace} 
to 
$O \bigparen{ \bigparen{ \frac{1}{\eps} + \ln \frac{1}{\delta} } \cdot \bigparen{ \ln \frac{1}{\delta} + \ln \frac{1}{\eps} } }$ 
bits, as stated in \cref{thm:approximate-discrete-laplace-sampler-formal}, 
at the cost of using slightly more random bits per sample,  
to complete the proof of \cref{thm:approximate-discrete-laplace-sampler-formal}.

The sampler construction is provided in \cref{algo:finite-discrete-laplace-noise}. 
It relies on two key decompositions. 
The first expresses a discrete Laplace random variable as a function of three simpler random variables: 
a Bernoulli, a uniform over $\{-1, 1\}$, and a geometric. 
The Bernoulli and uniform variables can be (approximately) sampled directly.

The second decomposition breaks the geometric random variable into a composition of two simpler geometric variables, 
which are then sampled by our general framework based on the Alias method (\cref{algo:finite-Alias-method}).

\begin{algorithm}[!t]
    \caption{Discrete Laplace Noise $\algoDLap$}
    \label{algo:finite-discrete-laplace-noise}
    \begin{algorithmic}[1]
        \Statex \hspace{-6.8mm} {\bf Procedure:}        \textsc{Initialization}
        \Statex \hspace{-4.6mm} {\bf Input:} 
            $\eps \in \Q_+$, $\delta \in \paren{0, 1} \cap \Q_+$; 

        \State $r \gets 2^{\ceil{\log_2 \frac{1}{\eps}}}$
        \State $\algoFiniteAlias^\downarrow \gets \algoFiniteAliasInitialization \paren{\operatorname{\mathbb{G}eo} \paren{e^{-\eps}, \range{r}}, \delta / 3}$
        \State $\algoFiniteAlias^\uparrow \gets \algoFiniteAliasInitialization \paren{ \operatorname{\mathbb{G}eo} \paren{ e^{- r \cdot \eps} }, \delta / 3}$
        \State 
        $
            \dlcenter' \gets \BinaryExp{
                \dlcenter = \P{ \DiscreteLapNoise{ e^{-\eps} } = 0 }
            }{
                \ell = \log_2 \frac{3}{\delta}
            }
        $
        
        \Statex \hspace{-6.8mm} {\bf Procedure:}        \textsc{Sample}
        \Statex \hspace{-4.8mm} {\bf Output}: Random variable $Z$ s.t. $\tvdistance{Z}{\DiscreteLapNoise{e^{-\eps}}} \le \delta$
        \State $B \gets \Bernoulli{\dlcenter'}$
            \State $S \gets \UniformNoise{\set{-1, 1}}$ 
            \State $X \gets \algoFiniteAlias^\uparrow.\sample \paren{}$
            \State $Y \gets \algoFiniteAlias^\downarrow.\sample \paren{}$
        \State {\bf return} $Z \gets \indicator{B = 0} \cdot S \cdot \PAREN{1 + r \cdot X + Y}$
    \end{algorithmic}
\end{algorithm}

\subsubsection*{Decomposition of the Discrete Laplace Distribution.}
The discrete Laplace distribution $\DiscreteLapNoise{e^{-\eps}}$ is defined by $\P{\DiscreteLapNoise{e^{-\eps}} = t} \propto e^{ - \card{t} \cdot \eps}$ for each $t \in \Z$.
This distribution is symmetric about zero, and conditioned on $t > 0$ or $t < 0$, it reduces to a geometric distribution.
This observation motivates the following sampling procedure: first, sample a Bernoulli random variable $B \sim \Bernoulli{ \dlcenter }$, where $\dlcenter \doteq \P{ \DiscreteLapNoise{ e^{-\eps} } = 0 }$, to decide whether to return $0$;
if not, sample a sign $S \sim \UniformNoise{ \set{ -1, 1 } }$ to determine the output’s sign, and a geometric variable $X \sim \GeometricNoise{ e^{-\eps} }$ to determine the magnitude.
This decomposition is formalized in \cref{fact:discrete-laplace-decomposition}.

\begin{fact}[Discrete Laplace Decomposition]
\label{fact:discrete-laplace-decomposition}
    Given $\eps \in \R_+$, let $\dlcenter \doteq \P{ \DiscreteLapNoise{ e^{-\eps} } = 0 }$, 
    $B \sim \Bernoulli{ \dlcenter }$, 
    $S \sim \UniformNoise{ \set{ -1, 1 } }$, 
    and $X \sim \GeometricNoise{ e^{-\eps} }$. 
    Then the random variable defined by 
    $Y \doteq \indicator{ B = 0 } \cdot S \cdot \PAREN{ 1 + X }$
    satisfies $Y \sim \DiscreteLapNoise{ e^{-\eps} }$.
\end{fact}

\begin{proof}[Proof of \cref{fact:discrete-laplace-decomposition}]
    Clearly $\P{Y = 0} = \P{B = 1} = \P{\DiscreteLapNoise{p} = 0}$.
    For each $t \in \N_+$, we have 
    \begin{align}
        \P{Y = t} 
            &= \P{B = 0} \cdot \P{S = 1} \cdot \P{X = t - 1} \\
            &= \PAREN{1 - \P{\DiscreteLapNoise{p} = 0}} \cdot \frac{1}{2} \cdot \PAREN{1 - p} \cdot p^{t - 1} \\
            &= \PAREN{1 - \frac{1 - p}{1 + p}} \cdot \frac{1}{2} \cdot \PAREN{1 - p} \cdot p^{t - 1} \\
            &= \frac{1 - p}{1 + p} \cdot p^{t}. 
    \end{align}
\end{proof}

One caveat is that $\dlcenter$ is a real number, and thus cannot be computed exactly or sampled from $\Bernoulli{\dlcenter}$ without approximation. 
However, as in \cref{lem:binary-expansion-discrete-laplace}, standard numerical techniques from \citet{Brent_Zimmermann_2010} allow us to compute an $\ell$-bit binary approximation $\dlcenter' = \BinaryExp{\dlcenter}{\ell}$ in $\Tilde{O}(\ell)$ time.
As long as $\ell \in \Omega \paren{ \log \frac{1}{\delta} }$, we can ensure that $\tvdistance{ \Bernoulli{\dlcenter} }{ \Bernoulli{\dlcenter'} } \in O(\delta)$.
It remains to discuss how to approximately sample from $\GeometricNoise{ e^{-\eps} }$ within total variation distance $O(\delta)$.

\subsubsection*{Decomposition of the Geometric Distribution.}
Applying our general framework (\cref{thm:time-oblivious-distribution-sampler}, \cref{algo:finite-Alias-method}) to approximately sample from $\GeometricNoise{e^{-\eps}}$ directly requires arrays whose lengths are proportional to the size of the core support set $\CoreSupportSet{\delta / 2} = \IntSet{0}{ \ceil{ \frac{1}{\eps} \cdot \ln \frac{2}{\delta} } }$ (see \cref{thm:time-oblivious-distribution-sampler} for the definition).

To reduce space usage, we cover the interval $\IntSet{0}{ \ceil{ \frac{1}{\eps} \cdot \ln \frac{2}{\delta} } }$ into smaller intervals of equal length: for $r \doteq 2^{\ceil{ \log_2 \frac{1}{\eps} }}$, define $\cI_i \doteq \IntSet{ i \cdot r }{ i \cdot r + r - 1 }$ for each $i = 0, \ldots, \ceil{ \ceil{ \frac{1}{\eps} \cdot \ln \frac{2}{\delta} } / r }$.
The key observation is that, conditioned on $Z$ falling into interval $\cI_i$, the relative position $t$ within the interval follows identical distribution $\GeometricNoise{e^{-\eps}}{\range{r}}$ (recall \cref{def:geometric-distribution}): for $Z \sim \GeometricNoise{e^{-\eps}}$,
\[
    \P{ Z = i \cdot r + t \mid Z \in \cI_i } \propto e^{-t \cdot \eps}, \quad \forall t \in \range{r}.
\]

This motivates a two-stage sampling procedure using independent geometric random variables: the first determines the interval index $i$, and the second determines the offset $t$ within $\cI_i$.
This decomposition is formalized in \cref{lemma:decomposition-of-geometric-distribution}.

\begin{lemma}[Geometric Decomposition]
    \label{lemma:decomposition-of-geometric-distribution}
    Let $\eps \in \R_+$, and $r \doteq 2^{\ceil{\log_2 \frac{1}{\eps}}}$.  
    Suppose $X \sim \GeometricNoise{e^{- r \cdot \eps}}$ and $Y \sim \GeometricNoise{e^{-\eps}}{\range{r}}$.  
    Then $r \cdot X + Y$ is distributed as $\GeometricNoise{e^{-\eps}}$.
\end{lemma}

\begin{proof}[Proof of \cref{lemma:decomposition-of-geometric-distribution}]
    For each $t \in \N$,
    \begin{align}
        \P{r \cdot X + Y = t}
            &= \P{X = \floor{t / r}} \cdot \P{Y = t \mod r} \\
            &= \PAREN{1 - e^{- \eps \cdot r }} \cdot e^{ - \eps \cdot r \cdot \floor{t / r} } \cdot \frac{1 - e^{-\eps}}{1 - e^{-\eps \cdot r}} \cdot e^{-\eps \cdot \PAREN{t \mod r}} \\
            &= \PAREN{1 - e^{-\eps}} \cdot e^{ 
                -\eps \cdot \PAREN{
                    r \cdot \floor{t / r}
                    + \PAREN{t \mod r}
                }
            } \\
            &= \PAREN{1 - e^{-\eps}} \cdot e^{ 
                -\eps \cdot t
            }
    \end{align}
\end{proof}

\noindent
Now we are ready to formally prove \cref{thm:approximate-discrete-laplace-sampler-formal}.

\begin{proof}[Proof of \cref{thm:approximate-discrete-laplace-sampler-formal}]
    The precomputation cost is similar to that in~\cref{thm:time-oblivious-discrete-laplace}.  
    It therefore suffices to analyze the memory usage, sampling cost, and sample quality.
    
    In \cref{algo:finite-discrete-laplace-noise}, 
    the sampler $\algoFiniteAlias^\downarrow$ for $\GeometricNoise{e^{-\eps}}{ \range{r} }$, within total variation distance $\delta / 3$, uses 
    $O \PAREN{ \frac{1}{\eps} \cdot \PAREN{ \ln \frac{1}{\delta} + \ln \frac{1}{\eps} } }$ bits of memory, 
    consumes $\log_2 \frac{1}{\eps} + \log_2 \frac{1}{\delta} + O(1)$ unbiased random bits per sample, 
    and runs in $O \PAREN{ \frac{1}{\omega} \cdot \ln \frac{1}{\eps} }$ time in the worst case.
    
    The sampler $\algoFiniteAlias^\uparrow$ for $\GeometricNoise{e^{- r \cdot \eps}}$, within total variation distance $\delta / 3$, uses 
    $O \PAREN{ \ln^2 \frac{1}{\delta} }$ bits of memory, 
    consumes $\log_2 \PAREN{ \ln \frac{1}{\delta} } + \log_2 \frac{1}{\delta} + O(1)$ unbiased random bits, 
    and runs in $O \PAREN{ \frac{1}{\omega} \cdot \ln \frac{1}{\delta} }$ time in the worst case.
    
    Finally, in \cref{algo:finite-discrete-laplace-noise}, we have 
    $\tvdistance{B}{\Bernoulli{\dlcenter}} \le \delta / 3$, 
    $\tvdistance{X}{\GeometricNoise{e^{- r \cdot \eps} }} \le \delta / 3$ 
    and $\tvdistance{Y}{\GeometricNoise{e^{-\eps}}{ \range{r}}} \le \delta / 3$.
    Then, by the data processing inequality (\cref{prop:data-processing-inequality}), sub-additivity of total variation distance (\cref{prop:tvd-subadditivity}), we obtain 
    \begin{align*}
        &\tvdistance{Z}{\DiscreteLapNoise{e^{-\eps}}} \\
            &\qquad \le \tvdistance{
                \set{B, S, X, Y}
            }{
                \set{
                    \Bernoulli{\dlcenter},
                    \UniformNoise{\set{-1, 1}},
                    \GeometricNoise{e^{- r \cdot \eps} },
                    \GeometricNoise{e^{-\eps}}{ \range{r}}
                }
            } \\
            &\qquad \le \tvdistance{B}{\Bernoulli{\dlcenter}} 
            + \tvdistance{X}{\GeometricNoise{e^{- r \cdot \eps} }}
            + \tvdistance{Y}{\GeometricNoise{e^{-\eps}}{ \range{r}}} \\
            &\qquad \le \delta.
    \end{align*}
\end{proof}

\newpage
\section{Proofs for \texorpdfstring{\cref{sec:dp-sparse-histogram}}{DP-Sparse-Histogram}}
\label{sec:proofs-for-dp-sparse-histogram}

\subsection{Privacy Guarantee}
\label{sec:missing-privacy-proofs-for-dp-sparse-histogram}

\begin{proof}[Proof of \cref{lem:dist-support-case-2a}]
    Recall that in \emph{Case 2}, we have  $\supp \paren{\hist} \cup \set{ i^* } = \supp \paren{\hist'}$, and therefore 
    \begin{align*}
            \hist'[i^*] = 1, \, \hist{i^*} = 0, \,
            \hist{j^*} > \hist'[j^*]  > 0.
    \end{align*}

    \noindent
    By definition of $\mu$ and $\mu'$, for each $S$, 
    \begin{align}
        \label{eq:expansion-of-mu}
        \mu(S) = \P{I = S} 
            &= \sum_{J \subseteq \supp \paren{\hist}} \P{ I = S \mid I_1 = J } \cdot \P{ I_1 = J }, \\
        \mu'(S) = \P{I' = S} 
                &= \sum_{J' \subseteq \supp \paren{\hist'}} \P{ I' = S \mid I_1' = J'} \cdot \P{ I_1' = J'}.
    \end{align}

    \noindent
    We can partition the subsets in $\supp \paren{\hist'}$ into the ones which contain $i^*$ and those which do not as follows: 
    \begin{align*}
        \set{ J' : J' \in \supp \paren{\hist'} }
            = \set{ J : J \in \supp \paren{\hist} }
            \cup \set{ J \cup \set{i^*} : J \in \supp \paren{\hist} }.
    \end{align*}
    Therefore
    \begin{align*}
        \P{I' = S} 
            &= \hspace{-4mm} \sum_{J \subseteq \supp \paren{\hist}} \hspace{-4mm}
                 \Big( 
                    \P{ I' = S \mid I_1' = J} \P{ I_1' = J } 
                + \P{ I' = S \mid I_1' = J \cup \set{i^*} } \P{ I_1' = J \cup \set{i^*} } \Big)
            .
    \end{align*}
    Based on the definition of $\distSupport'_{\bar{\eventSelectIstart}}$, we have 
    \begin{align}
        \distSupport'_{\bar{\eventSelectIstart}}(S)
            = \sum_{J \subseteq \supp \paren{\hist}} 
                    \P{ I' = S \mid I_1' = J} \cdot \P{ I_1' = J \mid i^* \notin I_1'} 
    \end{align}
    Comparing it with \cref{eq:expansion-of-mu}, it follows directly from the construction of \cref{algo:tail-item-padding} that 
    $$
        \P{ I' = S \mid I_1' = J } = \P{ I = S \mid I_1 = J }.
    $$
    To compare $\P{ I_1 = J }$ and $\P{ I_1' = J \mid i^* \notin I_1'}$, we need the following lemma.

    \begin{lemma}
        \label{lemma:expansion-of-probablity-for-neighboring-private-samples}
        For each $i \in [d]$, let $Z_i$ (or $Z_i'$) be the indicator for $i \in I_1$ (or $i \in I_1'$).
        For each $J \subseteq \supp \paren{\hist}$, define $r_J \doteq \P{ Z_{j^*} = \indicator{j^* \in J} } / \P{ Z_{j^*}' = \indicator{j^* \in J} }$. 
        Then 
        \begin{align*}
            \P{ I_1' = J } 
                &= (1 - p_\tau) \cdot r_J \cdot \P{ I_1 = J }, 
            \qquad
            \P{ I_1' = J \cup \set{i^*} } 
                = p_\tau \cdot r_J \cdot \P{ I_1 = J }.
        \end{align*}
    \end{lemma}
    \noindent
    Based on the privacy guarantee of the noise sampler in \cref{algo:tail-item-padding} (\cref{thm:purified-approximate-discrete-laplace-sampler-informal}), it holds that $r_J \in (e^{-\eps / 2}, e^{\eps / 2})$.
    Therefore, for each $J \subseteq \supp \paren{\hist}$, 
    \begin{align*}
        \P{ I_1' = J \mid i^* \notin I_1'}
            &= \frac{
                \P{ I_1' = J , i^* \notin I_1'}
            }{
                \P{ i^* \notin I_1' }
            } 
            = \frac{
                \P{ I_1' = J }
            }{
                1 - p_\tau
            } \\
            &= r_J \cdot \P{ I_1 = J } 
            \in (e^{-\eps / 2}, e^{\eps / 2}) \cdot \P{ I_1 = J },
    \end{align*}
    where by definition of $p_\tau$, $\P{ i^* \notin I_1' } = 1 - p_\tau$ and 
    $\P{ I_1' = J , i^* \notin I_1'} = \P{ I_1' = J }$ since for each 
    $J \subseteq \supp \paren{\hist}$, $i^* \notin J$.
    It concludes that 
    \begin{equation}
        \distSupport'_{\bar{\eventSelectIstart}}(S) \in (e^{-\eps / 2}, e^{\eps / 2}) \cdot \mu(S).
    \end{equation}
\end{proof}

\begin{proof}[Proof of \cref{lem:dist-support-case-2b}]
    Following the same argument as in the proof of \cref{lem:dist-support-case-2a},  
    for each $S$ we have
    \begin{align}
        \distSupport'_{\eventSelectIstart}(S)
            &= \sum_{J \subseteq \supp \paren{\hist}}
                \P{ I' = S \mid I_1' = J \cup \set{i^*} }
                \cdot
                \P{ I_1' = J \cup \set{i^*} \mid i^* \in I_1' } .
    \end{align}
    We compare this with
    \[
        \mu(S)
            = \P{I = S}
            = \sum_{J \subseteq \supp \paren{\hist}}
                \P{ I = S \mid I_1 = J } \cdot \P{ I_1 = J } .
    \]

    \paragraph{Step 1: Comparing the marginal probabilities.}
    By \cref{lemma:expansion-of-probablity-for-neighboring-private-samples},
    \begin{align*}
        \P{ I_1' = J \cup \set{i^*} \mid i^* \in I_1' }
            &= \frac{
                \P{ I_1' = J \cup \set{i^*},\, i^* \in I_1' }
            }{
                \P{ i^* \in I_1' }
            }
            = \frac{
                \P{ I_1' = J \cup \set{i^*} }
            }{
                p_\tau
            } \\
            &= r_J \cdot \P{ I_1 = J }
            \in (e^{-\eps/2}, e^{\eps/2}) \cdot \P{ I_1 = J },
    \end{align*}
    where $\P{ i^* \in I_1'} = p_\tau$ by definition.

    \noindent
    \paragraph{Step 2: Comparing the conditional probabilities.}
    We have the following lemma.
    \begin{lemma}
        \label{lemma:comparion-of-conditional-padding-probablities}
        Let $\ConditionalSupportRatio \doteq \frac{d}{k}$.
        For each $J \subseteq \supp \paren{\hist}$,
        \begin{equation}
            \P{ I' = S \mid I_1' = J \cup \set{i^*} }
            \le
            \ConditionalSupportRatio \cdot
            \P{ I = S \mid I_1 = J } .
        \end{equation}
    \end{lemma}

    \noindent
    Applying the lemma and combining both steps,
    \begin{align}
        \mu(S) 
            &\ge \sum_{J \subseteq \supp \paren{\hist}} \frac{1}{\ConditionalSupportRatio} \cdot \P{ I' = S \mid I_1' = J \cup \set{i^*} } \cdot e^{-\eps / 2} \cdot \P{ I_1' = J \cup \set{i^*} \mid i^* \in I_1' } \\
            &\ge \frac{1}{\ConditionalSupportRatio} \cdot e^{-\eps / 2} \cdot \distSupport'_{\eventSelectIstart} (S).
    \end{align}
\end{proof}

\begin{proof}[Proof of \cref{lem:dist-support-case-4}]
    Recall in this case, $\supp \paren{\hist} \cup \set{i^*} = \supp \paren{\hist'} \cup \set{ j^* }$.
    First, observe that the subsets in $\supp \paren{\hist}$ can be partitioned into those containing $j^*$ and those that do not.  
    For each $J \subseteq \supp \paren{\hist}$ with $j^* \notin J$, we have $J \subseteq \supp \paren{\hist} \cap \supp \paren{\hist'}$.  
    Hence, 
    \begin{equation}
        \begin{aligned}
            \mu(S) = \P{I = S} 
                &= \sum_{J \subseteq \supp \paren{\hist} \cap \supp \paren{\hist'} } { 
                    \P{ I = S \mid I_1 = J } \cdot \P{ I_1 = J } 
                } \\
                &+ \sum_{J \subseteq \supp \paren{\hist} \cap \supp \paren{\hist'}} { 
                    \P{ I = S \mid I_1 = J \cup \set{j^*} } \cdot  \P{ I_1 = J \cup \set{j^*} } . 
                }
        \end{aligned}
    \end{equation}
    Similarly,
    \begin{equation}
        \begin{aligned}
            \mu'(S) = \P{I' = S} 
                &= \sum_{J \subseteq \supp \paren{\hist} \cap \supp \paren{\hist'} } { 
                    \P{ I' = S \mid I_1' = J } \cdot \P{ I_1' = J } 
                } \\
                &+ \sum_{J \subseteq \supp \paren{\hist} \cap \supp \paren{\hist'}} { 
                    \P{ I' = S \mid I_1' = J \cup \set{i^*} } \cdot  \P{ I_1' = J \cup \set{i^*} } . 
                }
        \end{aligned}
    \end{equation}
    Further, by the definition of 
    $\distSupport_{\bar{\eventSelectJstart}}$, 
    $\distSupport'_{\bar{\eventSelectIstart}}$,
    $\distSupport'_{\eventSelectIstart}$, 
    $\distSupport_{\eventSelectJstart}$, 
    it holds that for each $S$
    \begin{align}
        \distSupport_{\bar{\eventSelectJstart}}(S)
            &= \sum_{J \subseteq \supp \paren{\hist} \cap \supp \paren{\hist'} } { 
                \P{ I = S \mid I_1 = J } \cdot \P{ I_1 = J \mid j^* \notin I_1 } 
            }. \\
        \distSupport_{\eventSelectJstart}(S)
            &= \sum_{J \subseteq \supp \paren{\hist} \cap \supp \paren{\hist'}} { 
                \P{ I = S \mid I_1 = J \cup \set{j^*} } \cdot  \P{ I_1 = J \cup \set{j^*} \mid j^* \in I_1 } . 
            } \\
        \distSupport'_{\bar{\eventSelectIstart}}(S)
            &= \sum_{J \subseteq \supp \paren{\hist} \cap \supp \paren{\hist'} } { 
                \P{ I' = S \mid I_1' = J } \cdot \P{ I_1' = J \mid i^* \notin I_1' } 
            }. \\
        \distSupport'_{\eventSelectIstart}(S)
            &= \sum_{J \subseteq \supp \paren{\hist} \cap \supp \paren{\hist'}} { 
                \P{ I' = S \mid I_1' = J \cup \set{i^*} } \cdot  \P{ I_1' = J \cup \set{i^*} \mid i^* \in I_1' } . 
            }
    \end{align}

    \subsubsection*{Comparison of $\distSupport_{\bar{\eventSelectJstart}}$ and $\distSupport'_{\bar{\eventSelectIstart}}$}
    Following directly from the construction of Algorithm~\ref{algo:tail-item-padding}, we have $$
        \P{ I' = S \mid I_1' = J } = \P{ I = S \mid I_1 = J }.
    $$
    Let the $Z_i$ ($Z_i'$) be defined as in Lemma~\ref{lemma:expansion-of-probablity-for-neighboring-private-samples}.
    For each $J \subseteq \supp \paren{\hist} \cap \supp \paren{\hist'}$, it holds that 
    \begin{align*}
        \P{ I_1' = J \mid i^* \notin I_1' } 
            &= \prod_{i \in \supp \paren{\hist'} \setminus \set{i^*}} \P{ Z_i' = \indicator{i \in J} } \\
            &= \prod_{i \in \supp \paren{\hist} \setminus \set{j^*}} \P{ Z_i = \indicator{i \in J} } 
            = \P{ I_1 = J \mid j^* \notin I_1 }. 
    \end{align*}
    Next,     
    \begin{align*}
        \P{ I_1' = J \cup \set{i^*} \mid i^* \in I_1' } 
            &= \prod_{i \in \supp \paren{\hist'} \setminus \set{i^*}} \P{ Z_i' = \indicator{i \in J} } \\
            &= \prod_{i \in \supp \paren{\hist} \setminus \set{j^*}} \P{ Z_i = \indicator{i \in J} } 
            = \P{ I_1 = J \cup \set{j^*} \mid j^* \in I_1 }. 
    \end{align*}
    Therefore, 
    \begin{equation}
        \label{eq:sampling-consistency}
        \begin{aligned}
            \P{ I_1 = J \mid j^* \notin I_1 } 
            &= \P{ I_1 = J \cup \set{j^*} \mid j^* \in I_1 } \\
            &= \P{ I_1' = J \mid i^* \notin I_1' } 
            = \P{ I_1' = J \cup \set{i^*} \mid i^* \in I_1' } . 
        \end{aligned}
    \end{equation}
        
    It immediately concludes that $\distSupport_{\bar{\eventSelectJstart}}= \distSupport'_{\bar{\eventSelectIstart}}$.

    \subsubsection*{Comparison of $\distSupport_{\bar{\eventSelectJstart}}$ and $\distSupport'_{\eventSelectIstart}$}
    Based on \cref{eq:sampling-consistency}, it suffices to compare for each $J \subseteq \supp \paren{\hist} \cap \supp \paren{\hist'}$, 
    $
        \P{ I = S \mid I_1 = J }\, \text{ and }\, \P{ I' = S \mid I_1' = J \cup \set{i^*} }.
    $
    There are four cases to be discussed.

    \subsubsection*{Case (i)}: $J \nsubseteq S$.
    Then
    $$
        \P{ I = S \mid I_1 = J } = \P{ I' = S \mid I_1' = J \cup \set{i^*} } = 0.
    $$

    \subsubsection*{Case (ii)}: $J \subseteq S$ and $i^* \notin S$. 
    Then 
    $$
        \P{ I = S \mid I_1 = J } \ge \P{ I' = S \mid I_1' = J \cup \set{i^*} } = 0.
    $$

    \subsubsection*{Case (iii)}: $J \subseteq S$ and $i^* \in S$.
    Recall the definition of $I_2$ in Algorithm~\ref{algo:tail-item-padding}, we have 
    \begin{align}
        \P{I = S \mid I_1 = J} 
            &= \P{ I_2 = S \setminus J }
            = \frac{1}{
                \binom{d - \card{J}}{n + k - \card{J}}
            }, \\
        \P{ I' = S \mid I_1' = J \cup \set{i^*} }
            &= \P{ I_2' = S \setminus \PAREN{J \cup \set{ i^* }} }
            = \frac{1}{
                \binom{d - \card{J} - 1}{n + k - \card{J} - 1}
            }, 
    \end{align}
    Denote $a = d - \card{J}$ and $b = n + k - \card{J}$.
    \begin{align}
        \frac{
             \P{I = S \mid I_1 = J} 
        }{
             \P{ I' = S \mid I_1' = J \cup \set{i^*} }
        }
        &= \frac{
            \binom{a - 1}{b - 1}
        }{
            \binom{a}{b}
        }
        = \frac{
            b
        }{
            a
        } \ge \frac{k}{d}.
    \end{align}    
    It concludes that $\distSupport_{\bar{\eventSelectJstart}} \ge \frac{k}{d} \cdot \distSupport'_{\eventSelectIstart}$.

    \subsubsection*{Comparison of $\distSupport'_{\bar{\eventSelectIstart}}$ and $\distSupport_{\eventSelectJstart}$}
    The comparison follows by symmetry from the analysis of
    $\distSupport_{\bar{\eventSelectJstart}}$ and $\distSupport'_{\eventSelectIstart}$.

\end{proof}

\subsection{Missing Proofs for \cref{sec:missing-privacy-proofs-for-dp-sparse-histogram}}

\begin{proof}[Proof of Lemma~\ref{lemma:expansion-of-probablity-for-neighboring-private-samples}]
        Recall in Algorithm~\ref{algo:tail-item-padding} that 
        $$
            p_\tau \doteq \P{ 1 + \ApproxDiscreteLaplaceNoise{1}{n}{\eps}{\eps \gamma / d} \ge \tau } = \P{\intermediateHist'[i^*] \ge \tau} = \P{i^* \in I_1'}.
        $$
        For each $J \subseteq \supp \paren{\hist}$, it holds that 
        \begin{align*}
            \P{ I_1' = J } 
                &= \P{i^* \notin I_1'} \cdot \P{ Z_{j^*}' = \indicator{j^* \in J} } \cdot \prod_{i \in \supp \paren{\hist} \setminus \set{j^*}} \P{ Z_i' = \indicator{i \in J} } \\
                &= (1 - p_\tau) \cdot r_J \cdot \P{ Z_{j^*} = \indicator{j^* \in J} } \cdot \prod_{i \in \supp \paren{\hist} \setminus \set{j^*}} \P{ Z_i = \indicator{i \in J} } \\
                &= (1 - p_\tau) \cdot r_J \cdot \P{ I_1 = J }, 
        \end{align*}
        Similarly, 
        \begin{align*}
            \P{ I_1' = J \cup \set{i^*} } 
                &= \P{i^* \in I_1'} \cdot \P{ Z_{j^*}' = \indicator{j^* \in J} } \cdot \prod_{i \in \supp \paren{\hist} \setminus \set{j^*}} \P{ Z_i' = \indicator{i \in J} } \\
                &= p_\tau \cdot r_J \cdot \P{ Z_{j^*} = \indicator{j^* \in J} } \cdot \prod_{i \in \supp \paren{\hist} \setminus \set{j^*}} \P{ Z_i = \indicator{i \in J} } \\
                &= p_\tau \cdot r_J \cdot \P{ I_1 = J }.
        \end{align*}

\end{proof}

    \begin{proof}[Proof of Lemma~\ref{lemma:comparion-of-conditional-padding-probablities}]
        There are several cases. 
    
        \subsubsection*{Case One}: $J \nsubseteq S$, then 
        $$
            \P{I = S \mid I_1 = J} = \P{I' = S \mid I_1' = J} = \P{I' = S \mid I_1' = J \cup \set{ i^* } }= 0.
        $$
        
        \subsubsection*{Case Two}: $J \subseteq S$ and $i^* \notin S$, then 
        \begin{align}
            \P{I' = S \mid I_1' = J \cup \set{ i^* } } 
                &= 0
                \le \P{ I = S \mid I_1 = J }.
        \end{align}
        
        \subsubsection*{Case Three}: $J \subseteq S$ and $i^* \in S$.
        First, 
        \begin{align}
            \P{I = S \mid I_1 = J} 
                &= \P{ I_2 = S \setminus J }
                = \frac{1}{
                    \binom{d - \card{J}}{n + k - \card{J}}
                }, \\
            \P{I' = S \mid I_1' = J \cup \set{ i^* }} 
                &= \P{ I_2' = S \setminus \PAREN{J \cup \set{ i^* }} }
                = \frac{1}{
                    \binom{d - \card{J} - 1}{n + k - \card{J} - 1}
                }. 
        \end{align}
        Denote $a = d - \card{J}$ and $b = n + k - \card{J}$.
        \begin{align}
            \frac{
                \P{I' = S \mid I_1' = J \cup \set{ i^* }} 
            }{
                \P{I = S \mid I_1 = J} 
            }
            &= \frac{
                \binom{a}{b}
            }{
                \binom{a - 1}{b - 1}
            }
            = \frac{
                a
            }{
                b
            }.
        \end{align}
        There is an intuitive explanation for this ratio: conditioned on $i^* \notin I_1 = J$, $\P{i^* \in S} = \P{i^* \in I_2 }= \frac{b}{a}$.
        It is easy to see that conditioned on $i^* \in I_2$, the set $I_2$ has exactly the same distribution as $I_2'$ and therefore
        $$
            \P{I_2 = S \setminus J \mid i^* \in I_2} 
            = 
            \P{I_2' = S \setminus \PAREN{J \cup \set{ i^* }} }.
        $$

    \end{proof}

\newpage
\section{Circuit-Based MPC Protocol for Sparse Histograms}
\label{subsec:mpc-circuit-sparse-hist}

In this section, we present a prototype secure multiparty computation (MPC) protocol for releasing an $\eps$-DP sparse histogram with optimal $\ell_\infty$ error guarantees.  
Specifically, we provide a circuit-based implementation of \cref{algo:tail-item-padding}.  
Our design leverages the simple structure of \cref{algo:tail-item-padding}, which facilitates efficient circuit realization. 

It is well known that any polynomial-size circuit family can be securely implemented in the MPC setting using standard techniques \citep{BGW88,GMW87}.  
We leave the optimization of the MPC protocol—such as selecting between binary or arithmetic secret sharing, or incorporating local computation under the semi-honest model—to future work.

\subsubsection*{\bf Circuit Primitives}  
We rely on the following standard circuit primitives, which are widely supported or readily implementable in modern MPC frameworks such as MP-SPDZ~\citep{mp-spdz}.  
Throughout, we assume that all inputs are represented as binary integers of bit-length $\omega$, unless otherwise specified.

\begin{itemize}[leftmargin=6mm, label={\scriptsize $\triangleright$}]
    \item \textsc{Add}: Given $a, b \in \N$, outputs $a + b$. \hfill (Gates: $O(\omega)$)
    
    \item \textsc{Mul}: Given $a, b \in \N$, outputs $a \cdot b$. \hfill (Gates: $O(\omega^2)$)
    
    \item \textsc{LessThan}: Given $a, b \in \N$, outputs $1$ if $a < b$, and $0$ otherwise. \hfill (Gates: $O(\omega)$)
    
    \item \textsc{Equal}: Given $a, b \in \N$, outputs $1$ if $a = b$, and $0$ otherwise. \hfill (Gates: $O(\omega)$)

    \item \textsc{Mux}: Given a bit $b \in \{0, 1\}$ and values $x, y \in \N$, outputs $b \cdot x + (1 - b) \cdot y$. \hfill (Gates: $O(\omega)$)

    \item \textsc{Sort}: Given $a_1, \ldots, a_n \in \N$, returns a sorted sequence $(b_1, \ldots, b_n)$. \hfill 
    (Gates: $O(\omega \cdot n \ln n)$ \citep{AjtaiKS83})

    \item \textsc{RandomTableAccess}: Given a public vector $\vec{a}$ and an index $J \in [m]$, returns $\vec{a}_J$. \hfill (Gates: $O(\omega \cdot m)$ for table of size $m$)
\end{itemize}

\noindent
We also use the following functionalities, which can be efficiently implemented using the above primitives:
\begin{itemize}[leftmargin=6mm, label={\scriptsize $\triangleright$}]
    \item \textsc{Clamp}: Given $x \in \N$ and bounds $a, b \in \N$, outputs $\max(a, \min(x, b))$. This can be implemented using \textsc{LessThan} and \textsc{Mux}. 
    \hfill (Gates: $O(\omega)$)

    \item \textsc{Bernoulli}$(p)$: Outputs $1$ with probability $p$, and $0$ otherwise. This can be implemented using a random input, \textsc{LessThan}, and \textsc{Mux}.  
    \hfill (Gates: $O(\log(1/\delta))$ for precision $\delta$)
\end{itemize}

\noindent
We also assume that uniform random numbers are provided as part of the circuit input.

\subsubsection*{\bf Protocol}
The protocol is presented in \cref{algo:eps-dp-sparse-histogram-circuit}, with the main properties summarized below.

\begin{theorem}[Private Sparse Histogram Circuit]
    \label{thm:purified-approximate-discrete-laplace-sampler-formal}\
    Given integers 
    $n, d, a_\eps, b_\eps, a_\gamma, b_\gamma \in \N_+$ 
    that fit in a constant number of machine words, 
    and $\gamma \doteq a_\gamma / b_\gamma \in \paren{ 1 / n^{O(1)}, 1 }$,  
    there exists a circuit with cost
    \[
        O \PAREN{
            \SortCircuitCost{n}{\omega} 
            + n \cdot \omega 
            + n \cdot \MulCost{\omega}
            + n \cdot \Bigparen{
                \frac{1}{\eps} \cdot \Bigparen{ \ln \frac{1}{\eps} + \ln \frac{d}{\gamma} }
                + \ln^2 \Bigparen{ \frac{d}{\eps \gamma} }
            }
        },
    \]
    where $\SortCircuitCost{n}{\omega}$ is the circuit cost of sorting $n$ elements represented using $\omega$ bits, and $\MulCost{\omega}$ is the circuit cost of multiplying two $\omega$-bit numbers.  
    This circuit implements a mechanism that takes 
    a dataset $\DataSet$ of $n$ user-contributed elements from domain $[d]$ as input, and outputs a histogram 
    $\noisyhist \in \zeroton{n}^d$ approximating the original histogram $\hist$ of $\DataSet$, 
    such that $\norm{\noisyhist}_0 \in O(n)$ and 
    \begin{itemize}[leftmargin=0.6cm, label=$\triangleright$]
    
        \item \textbf{Privacy Guarantee:}  
        $\ourAlgo$ satisfies $2\eps$-differential privacy, where $\eps \doteq a_\eps / b_\eps \in \Q_+$.
    
        \item \textbf{Utility Guarantee:}  
        There exists a universal constant $c_\alpha \in \R_+$ such that for every $\beta \ge 2 \eps \gamma$, letting 
        $\alpha \doteq (c_\alpha / \eps) \cdot \ln \paren{ d / \beta }$,  
        the output $\noisyhist$ is an $(\alpha, \beta)$-simultaneous accurate estimator of $\hist$.
    
    \end{itemize}
\end{theorem}

\begin{algorithm}[H]
    \caption{High-Level Circuit Description of the $\eps$-DP Sparse Histogram Protocol}
    \label{algo:eps-dp-sparse-histogram-circuit}
    \begin{algorithmic}[1]
        \Statex \hspace{-4.8mm} {\bf Input:} 
            parameters $d, n, k = 3n, a_\eps, b_\eps, a_\gamma, b_\gamma \in \N_+$, s.t., $\eps = a_\eps / b_\eps$ and $\gamma = a_\gamma / b_\gamma$
        \Statex \hspace{-4.8mm} \hspace*{\widthof{\textbf{Input:}}} 
            parameter $\tau \in \N_+$
            \Comment{as defined in \cref{eq:def-tau}}
        \Statex \hspace{-4.8mm} \hspace*{\widthof{\textbf{Input:}}} 
            Participants' data $\UserData{1}, \ldots, \UserData{n}$

        \vspace{2mm}
        \Statex \hspace{-4.8mm} \textbf{Histogram Construction}
        \vspace{1mm}
        
        \State \textsc{Sort} the array $\UserData{1}, \ldots, \UserData{n}$ to obtain $\SortedUserData{1}, \ldots, \SortedUserData{n}$. 

        \State $\textit{firstMinusOne} \gets 0$
        \For{$i \in [n]$}
            \If{$\SortedUserData{i} \neq \SortedUserData{i + 1})$} \Comment{$i$ is the last occurrence of $\SortedUserData{i}$}
                \State $c_i \gets i -  \textit{firstMinusOne}$ 
                    \Comment{compute frequency}  
                \State $\textit{firstMinusOne} \gets i$ 
                \State \textsc{Create pair:} $(\SortedUserData{i}, c_i)$ 
            \Else 
                \State \textsc{Create dummy pair:} $(\bot, \bot)$ 
            \EndIf
        \EndFor
        \State Collect all $n$ created pairs into array $I_1$.
        \Comment{Resulting array used for further processing}
        
        \vspace{3mm}
        \Statex \hspace{-4.8mm} \textbf{Noise Addition and Thresholding}
        \vspace{1mm}

        \For{$i \in [n]$} \Comment{each pair in $I_1$}
        \If{$I_1[i]$ is non-dummy}
            \begin{itemize}[leftmargin=1.2cm, label=$\circ$]
                    \item \textsc{Create triple:} 
                        $(y, c, Z)$, where 
                        $Z \gets \PurifiedApproxDiscreteLaplaceMechanism{n}{\eps}{\eps \gamma / d}(c)$.
                    \item If $Z < \tau$, convert the triple into a dummy triple $(\bot, \bot, \bot)$.
                    \item Otherwise, set the third component to $1$ to indicate membership in $I_1$.
                \end{itemize}
            \Else \,
                \textsc{Create dummy triple:} $(\bot, \bot, \bot)$
            \EndIf
        \EndFor
        \State Collect all $n$ created triples into array $I_1$. \Comment{Reuse the notation}
        \State Sort $I_1$ so that non-dummy triples precede dummy triples.
    
        \vspace{3mm}
        \Statex \hspace{-4.8mm} \textbf{Privacy Blanket Sampling}
        \vspace{1mm}
        \State Sample a uniform random subset $I_2 \subseteq [d]$ of size $n + k$
        \State Convert each $i \in I_2$ to a triple $(i, 0, 2)$ 
        \Comment{the label 2 indicates membership in $I_2$}

        \vspace{3mm}
        \Statex \hspace{-4.8mm} \textbf{Privacy Blanket Merging}
        \vspace{1mm}
        \State Merge $I_1$ and $I_2$ to obtain $I$:
        \begin{itemize}[label=$\circ$]
            \item Append $I_2$ to $I_1$
            \item Sort the array by the first component, breaking ties by preferring label 1 over 2
            \item If two triples are equal (based on their first component), keep the first and convert the other to dummy triple.
            \item Re-sort the array by the third component, using the order $1 < 2 < \bot$
            \item Drop the third component from all triples to obtain final pairs
        \end{itemize}

        \vspace{1mm}
        \For{$i \in [n + k]$}
            \State Let $(y, c) \gets I[i]$
            \State Sample $Z \gets \PurifiedApproxDiscreteLaplaceMechanism{n}{\eps}{\eps \gamma / d}(c)$
            \State Replace $I[i]$ with $(y, Z)$
        \EndFor
        \State \Return the first $n + k$ pairs in $I$
    \end{algorithmic}
\end{algorithm}

The key idea is to replace the dictionary-based operations used in the central model with equivalent operations implemented via sorting in circuits.
The protocol proceeds in four phases:
\emph{histogram construction}, 
\emph{noise addition and thresholding}, 
\emph{privacy blanket sampling}, 
and 
\emph{privacy blanket merging}.
We describe each component in detail below.

\subsubsection*{\bf Histogram Construction.}  
Given the participants' data $\UserData{1}, \ldots, \UserData{n}$,  
\cref{algo:eps-dp-sparse-histogram-circuit} constructs an array $I_1$ consisting of all pairs $\paren{y, \hist{y}}$ for $y \in \supp(\hist)$, along with $n - \card{\supp(\hist)}$ dummy pairs to hide the true support size.  
This is achieved by first sorting the inputs $\UserData{1}, \ldots, \UserData{n}$ into $\SortedUserData{1}, \ldots, \SortedUserData{n}$.  
Then, for each $i \in [n]$, if $i$ corresponds to the last occurrence of $\SortedUserData{i}$ in the sorted array,  
the algorithm creates a pair consisting of $\SortedUserData{i}$ and the difference between $i$ and the index of $\SortedUserData{i}$'s first appearance minus one—yielding its count $\hist{y}$.  
Otherwise, the algorithm generates a dummy pair.  
Denote the array of the generated pairs as $I_1$.

This phase uses the following circuit primitives: one \textsc{Sort}, and $O(n)$ instances each of \textsc{Mux}, \textsc{Add}, \textsc{Sub}, and \textsc{Equal}. 
Therefore, it has cost 
$$O(\SortCircuitCost{n}{\omega} + n \cdot \omega).$$

\subsubsection*{\bf Noise Addition and Thresholding.}  
In this phase, for each $i \in [n]$, if $I_1[i] = (y, c)$ is a non-dummy pair,  
\cref{algo:eps-dp-sparse-histogram-circuit} samples a noisy count $Z$ from the mechanism  
$\PurifiedApproxDiscreteLaplaceMechanism{n}{\eps}{\eps \gamma / d}(c)$ and creates a triple $(y, c, Z)$.  
If $Z < \tau$, the triple is replaced by a dummy triple $(\bot, \bot, \bot)$.  
If $I_1[i]$ is a dummy pair, the algorithm directly creates a dummy triple.
Finally, the algorithm collects all triples and sorts them so that non-dummy triples precede dummy ones.  
We reuse the notation and denote the resulting array as $I_1$ again.

After this step, the set of elements appearing in the first component of the triples in $I_1$ in  
\cref{algo:eps-dp-sparse-histogram-circuit} exactly matches those in $I_1$ of \cref{algo:tail-item-padding}.

This phase uses the following circuits: one~\textsc{Sort}, and $O(n)$ instances of \textsc{Mux}, \textsc{LessThan}, and the circuit for $\PurifiedApproxDiscreteLaplaceMechanism{n}{\eps}{\eps \gamma / d}$.  
Hence, the cost is
\[
    O \bigl(
        \SortCircuitCost{n}{\omega} + n \cdot \omega + n \cdot \mathrm{Cost}\bigl(
            \PurifiedApproxDiscreteLaplaceMechanism{n}{\eps}{\eps \gamma / d}
        \bigr)
    \bigr).
\]

\noindent
\textit{Circuit for $\PurifiedApproxDiscreteLaplaceMechanism{n}{\eps}{\eps \gamma / d}$.}
It remains to describe how to implement $\PurifiedApproxDiscreteLaplaceMechanism{n}{\eps}{\eps \gamma / d}$.  
We will show that the circuit for it can be implemented with cost
\[
    O\left(
        \frac{1}{\eps} \cdot \Bigparen{ \ln \frac{1}{\eps} + \ln  \frac{d}{\gamma} } 
        + \ln^2 \Bigparen{ \frac{d}{\eps \gamma} }
        + \MulCost{\omega}
    \right).
\]
Based on the proof of \cref{thm:purified-approximate-discrete-laplace-sampler} and the structure of \cref{alg:purification}, this mechanism can be implemented as follows:
\begin{enumerate}[leftmargin=0.8cm]
    \item Initialize an instance of $\DLapSampler{\eps}{\delta}$ as defined in \cref{thm:approximate-discrete-laplace-sampler-formal}, where $\delta = \frac{e^\eps - 1}{e^\eps + 1} \cdot \frac{\eps \gamma / d}{1 - \eps \gamma / d} \cdot \frac{1}{1 + n}$.
    
    \item Compute
    \begin{equation}
        \PurifiedApproxDiscreteLaplaceMechanism{n}{\eps}{\eps \gamma / d}(c)
            \doteq (1 - B) \cdot \clamp{c + Z'}{0}{n} + B \cdot \UniformNoise{\zeroton{n}}, \label{eq:purified-sampler-mpc}
    \end{equation}
    where $B \sim \Bernoulli{\eps \gamma / d}$ and $Z' \sim \DLapSampler{\eps}{\delta}$.
\end{enumerate}

\noindent
Recall that $\UniformNoise{\zeroton{n}}$ can be provided as part of the circuit input.  
We also assume that $\eps \gamma / d$ is replaced by the largest power of $2$ smaller than $\eps \gamma / d$, denoted as $\gamma''$.  
Accordingly, we adjust the parameter $\delta$ to 
\[
    \delta = \frac{e^\eps - 1}{e^\eps + 1} \cdot \frac{\gamma''}{1 - \gamma''} \cdot \frac{1}{1 + n}.
\]
This approximation does not affect the privacy guarantee, the asymptotic utility, or the asymptotic circuit complexity.

Under this setting, the Bernoulli variable $B \sim \Bernoulli{\gamma''}$ can be sampled exactly using the \textsc{Bernoulli}$(p)$ circuit primitive, which requires $O\left( \log_2 \frac{d}{\eps \gamma} \right)$ gates.  
Overall, the implementation of $\PurifiedApproxDiscreteLaplaceMechanism{n}{\eps}{\eps \gamma / d}(c)$ uses the following circuit components: 
one \textsc{Bernoulli}, one \textsc{Mux}, one \textsc{Clamp}, and one circuit for $\DLapSampler{\eps}{\delta}$.

\noindent
\textit{Circuit for $\DLapSampler{\eps}{\delta}$.}  
We will show that the circuit for $\DLapSampler{\eps}{\delta}$ can be implemented with cost
\[
    O\left(
        \frac{1}{\eps} \cdot \ln \frac{1}{\eps \delta} + \ln^2 \frac{1}{\delta} + \MulCost{\omega}
    \right).
\]
We briefly review the design of $\DLapSampler{\eps}{\delta}$ from \cref{algo:finite-discrete-laplace-noise}. Each sample $Z'$ from this mechanism has the form
\begin{equation*}
    Z' \gets (1 - B') \cdot S \cdot \paren{1 + r \cdot X + Y},
\end{equation*}
where:
\begin{itemize}[leftmargin=0.8cm, label=$\circ$]
    \item $B' \sim \algoBernoulli{\dlcenter}$, for 
    \[
        \begin{array}{c}
            \dlcenter \gets \BinaryExp{ 
                \P{ \DiscreteLapNoise{ e^{-\eps} } = 0 }
            }{
                \ell = \log_2 \frac{3}{\delta}
            }, 
        \end{array}
    \]
    the $\ell$-bit binary expansion of $\P{ \DiscreteLapNoise{ e^{-\eps} } = 0 }$;
    \item $S \sim \UniformNoise{\set{-1, 1}}$;
    \item $X$ and $Y$ are sampled using the finite Alias method (\cref{algo:finite-Alias-method}):
    \begin{align*}
        X \sim \algoFiniteAlias^\uparrow.\sample(), 
        \qquad 
        Y \sim \algoFiniteAlias^\downarrow.\sample(),
    \end{align*}
    where the alias samplers are initialized as: with $r = 2^{\ceil{\log_2 \frac{1}{\eps}}}$, 
    \begin{align*}
        \algoFiniteAlias^\uparrow &\gets \algoFiniteAliasInitialization \paren{ \operatorname{\mathbb{G}eo} \paren{ e^{- r \cdot \eps} }, \delta / 3}, \\
        \algoFiniteAlias^\downarrow &\gets \algoFiniteAliasInitialization \paren{\operatorname{\mathbb{G}eo} \paren{e^{-\eps}, \range{r}}, \delta / 3}.
    \end{align*}
\end{itemize}

\noindent
$S$ can be part of the input to the circuit, and $B'$ can be implemented with the \textsc{Bernoulli} circuit primitive.
It suffices to discuss the sampling of $X$ and $Y$.
We focus on $X$, as the procedure for $Y$ is analogous.

The sampler $\algoFiniteAlias^\uparrow$ consists of two arrays of length 
$m \in O\left( \frac{1}{r \cdot \eps} \cdot \ln \frac{1}{\delta} \right)$:  
a probability array $\vec{b}$, where each entry has bit-length $O\left( \ln \frac{1}{\delta} \right)$, and an alias array $\vec{a}$, where each entry has bit-length $O(\ln m)$.
Sampling from $\algoFiniteAlias^\uparrow$ proceeds in three steps:
\begin{itemize}[leftmargin=0.8cm, label=$\circ$]
    \item Sample a uniform random index $J \in \range{m}$.
    \item Retrieve $\vec{a}[J]$ and $\vec{b}[J]$.
    \item Sample a Bernoulli random variable $\Bernoulli{\vec{b}[J]}$; return $J$ if the outcome is $1$, otherwise return $\vec{a}[J]$.
\end{itemize}

\noindent
In this construction, the first step is provided as part of the circuit input (i.e., random index $J$ is given).
The second step is implemented using the circuit primitive \textsc{RandomTableAccess}, whose cost is 
\[
    O \PAREN{ \frac{1}{r \cdot \eps} \cdot \ln \frac{1}{\delta} \cdot \left( \ln \frac{1}{\delta} + \ln m \right) }
    = O \PAREN{ \ln^2 \frac{1}{\delta} },
\]
corresponding to the total bit-length required to store $\vec{a}$ and $\vec{b}$.
The third step can be efficiently implemented using the circuit primitives \textsc{Bernoulli} and \textsc{Mux}.

Similarly, the circuit for sampling $Y$ has cost 
\[
    O \PAREN{ r \cdot \left( \ln \frac{1}{\delta} + \ln r \right) }
    = O \PAREN{ \frac{1}{\eps} \cdot \ln \frac{1}{\eps \delta}  }.
\]

\subsubsection*{\bf Privacy Blanket Sampling.}  
As discussed in \cref{subsec:proof-private-sparse-histogram}, to generate a uniformly random subset of size $n + k$ from $[d]$, we use the following efficient method:
\begin{itemize}[leftmargin=0.8cm, label=$\circ$]
  \item Sample $m = 4 \cdot (n + k)$ elements independently and uniformly from $[d]$, where $k = n$.
  
  \item Remove duplicates by first sorting the sampled array. Then, compare each element with its predecessor: if two adjacent elements are equal, mark the first one as a dummy. Finally, sort again to move all dummy elements to the end.
  
  \item If the resulting array contains at least $n + k$ distinct elements, return the first $n + k$ elements. 
  
  \item Otherwise, terminate \cref{algo:eps-dp-sparse-histogram-circuit} early and return a fixed sparse histogram, 
  e.g., one in which elements $1$ through $n$ each have count $1$.
  While this may slightly degrade utility, it does not compromise privacy.  
  As shown in \cref{subsec:proof-private-sparse-histogram}, this increases the failure probability of the utility guarantee by at most an additive $\sqrt{n + k} \cdot \exp(- (n + k) / 16)$.
\end{itemize}

\noindent
This phase uses the following circuit primitives: two \textsc{Sort} operations, and $O(n)$ applications of \textsc{Equal} and \textsc{Mux}.
Therefore, it has cost 
$$O(\SortCircuitCost{n}{\omega} + n \cdot \omega).$$

\subsubsection*{\bf Privacy Blanket Merging.}  
We describe how to pad the non-dummy entries in $I_1$ with sampled entries from $I_2$ to produce an output array of size $n + k$.  
At this stage, all entries in both $I_1$ and $I_2$ are represented as triples.  
We begin by appending $I_2$ to the end of $I_1$, resulting in an array of size $2n + k$.  

Next, we sort the combined array based on the first component of each triple.  
We then scan the array to identify duplicates—elements appearing in both $I_1$ and $I_2$—and convert the copy from $I_2$ into a dummy triple, ensuring that the original from $I_1$ is retained.  
After this deduplication step, we re-sort the array so that all dummy triples appear at the end.  
Then, we drop the third component from each remaining triple, converting them into pairs.

Finally, for the first $n + k$ pairs, we sample fresh noise values for their second components and release the resulting array.

\noindent
This phase uses the following circuit primitives: two \textsc{Sort} operations, and $O(n)$ applications of \textsc{Equal}, \textsc{Mux}, and the circuit for $\PurifiedApproxDiscreteLaplaceMechanism{n}{\eps}{\eps \gamma / d}$.  
Therefore, it has cost 
\[
    O\bigl(\SortCircuitCost{n}{\omega} + n \cdot \omega + n \cdot \mathrm{Cost}\bigl(\PurifiedApproxDiscreteLaplaceMechanism{n}{\eps}{\eps \gamma / d}\bigr)\bigr).
\]

\newpage

\section{Extension to Add/Remove Model}
\label{sec:add-remove-model}

In \cref{sec:dp-sparse-histogram}, we present a deterministic linear-time algorithm achieving optimal $\ell_\infty$ error for the pure DP sparse histogram problem under the replacement neighboring model (see \cref{sec:problem-description} for its definition).
We now investigate the following question:

\begin{quote}
    \textbf{Research Question.} Can a comparable algorithm be designed under another commonly studied notion of neighboring datasets, the \emph{add/remove model}, defined formally below?
\end{quote}

\subsubsection*{\bf Add/Remove Neighboring Model.}
Two datasets $\DataSet$ and $\DataSet'$ are said to be \emph{neighboring}, denoted $\DataSet \sim \DataSet'$, if they differ by the addition or removal of a single entry; that is,
$\DataSet' = \DataSet \cup { \UserData{n+1} }$ or
$\DataSet' = \DataSet \setminus { \UserData{i} }$ for some $i \in [n]$.

Accordingly, two histograms $\hist, \hist' \in \N^d$ are called \emph{neighboring} if they are induced by neighboring datasets under this model.

In the add/remove model, there is no fixed upper bound on the dataset size, so it is often referred to as the \emph{unbounded setting}.
A commonly studied relaxation is the \emph{upper-bounded setting}, where a public upper bound on the dataset size is assumed.

This add/remove model is fundamentally more challenging than the replacement model—for example, the dataset size $n$ cannot be revealed under add/remove neighbors, whereas it can be published in the replacement model.

\subsubsection*{\bf Impossibility Result.}

The answer to the research question is \textbf{negative}.
If an algorithm’s running time depends deterministically on $n$, then its running time leaks $n$: neighboring datasets can have different sizes and thus different execution times.
Conversely, if an algorithm’s running time is independent of $n$, then achieving optimal $\ell_\infty$ error is impossible.
Since $n$ may be arbitrarily larger than the fixed running time, there exist inputs for which the algorithm does not even have enough time to read the entire dataset, making accurate histogram estimation unattainable.

On the other hand, we show that it is possible to design an expected linear-time algorithm with optimal $\ell_\infty$ error whose running time remains private, in the sense formalized below.

\subsubsection*{\bf Time Obliviousness.}
The privacy of the running time is captured with by an extended notion differential privacy.
Formally, we require the joint distribution of the algorithm’s output and running time to satisfy differential privacy.

\begin{definition}[$\paren{\eps, \delta}$-DP Time-Oblivious Algorithm~\citep{DovDNT23}] 
    \label{def:joint-output-timing-dp}
    Let $\eps \in \R_+$ and $\delta \in [0, 1]$.  
    A randomized algorithm $\cM: \cY \rightarrow \cZ$ is said to be \emph{$\paren{\eps, \delta}$-differentially private time-oblivious} if, for every $\DataSet, \DataSet' \in \cY$ such that $\DataSet \sim \DataSet'$,  
    the joint distributions of $\paren{ \cM (\DataSet), \runningtime{\cM (\DataSet)} }$ and $\paren{ \cM (\DataSet'), \runningtime{\cM (\DataSet')} }$ are $\paren{\eps, \delta}$-indistinguishable,  
    where $\runningtime{\cM (\DataSet)}$ denotes the running time of $\cM (\DataSet)$.
\end{definition}

A similar definition was proposed by \citet{RatliffV24}, under the name \emph{$\paren{\eps, \delta}$-Joint Output/Timing Privacy}. 
Their formulation is more comprehensive, as it accounts for the influence of the computational environment on execution time by modeling it as part of the algorithm’s input.  
In this work, we adopt the formulation of \citet{DovDNT23} for simplicity of presentation.

\subsection*{From Replacement to Add/Remove Neighboring Model}
\label{subsec:replacement-to-add-remove}

Our solution builds on the recent framework of \citet{ratliff2025securing} for converting differentially private, time-oblivious algorithms (referred to as Joint Output/Timing private algorithms in their work) from the \emph{upper-bounded setting} (see the discussion following \cref{def: Differential Privacy}) to the \emph{unbounded setting}, under the add/remove neighboring model. 
The original framework is developed in the RAM model, where each memory cell can store an arbitrarily large natural number.

\begin{proposition}[\citep{ratliff2025securing}]
    \label{prop:ub-to-unbounded}
    For all $\beta_1 > 0$, $\eps, \eps_1 > 0$ and $\eps$-Joint Output/Timing private mechanism $\cM : \cY \rightarrow \cZ$ in the RAM model for the upper-bounded setting, 
    there exists a $\eps + \eps_1$-Joint Output/Timing private mechanism $\cM' : \cY \rightarrow \cZ$ in the RAM model, for the unbounded setting, such that 
    \begin{equation}
        \tvdistance{\cM \paren{y} }{\cM' \paren{y} } \le \beta_1.
    \end{equation}
\end{proposition}

For clarity of presentation, we tailor the framework to the private sparse histogram problem in Word-RAM model.
The pseudocode is given in \cref{alg:private-upper-bound-search}, and the construction relies on two key components.

\begin{algorithm}[H]
    \caption{Applying \citeauthor{ratliff2025securing}'s Framework}
    \label{alg:private-upper-bound-search}
    \begin{algorithmic}[1]
        \Statex \hspace{-4.8mm} {\bf Input:} 
            parameters $d, n, k, a_\eps, b_\eps, a_\gamma, b_\gamma \in \N_+$, s.t., $\eps = a_\eps / b_\eps$ and $\gamma = a_\gamma / b_\gamma$
        \Statex \hspace{-4.8mm} \hspace*{\widthof{\textbf{Input:}}} 
            parameter $\tau \in \N_+$
            \Comment{as defined in \cref{eq:def-tau}}
        \Statex \hspace{-4.8mm} \hspace*{\widthof{\textbf{Input:}}} 
            histogram $\DataSet = \set{\UserData{1}, \ldots, \UserData{n}}$ 
        \Statex \hspace{-4.8mm} \hspace*{\widthof{\textbf{Input:}}} 
            additional parameters $a_{\eps_1}, b_{\eps_1}, a_{\beta_1}, b_{\beta_1} \in \N_+$, s.t., $\eps_1 = a_{\eps_1} / b_{\eps_1}$ and $\beta_1 = a_{\beta_1} / b_{\beta_1}$
        \For{$k \in \N_+$}
            \State $\eps_k \gets \eps_1 / 2^k$, \quad $\beta_k \gets \beta_1 / 2^k$
            \State $\hat{n}_k \gets \ceil{ \frac{8}{\eps_k} \cdot \ln \frac{1}{\beta_k} }$
            \State $t \gets \min\{n, \hat{n}_k\}$
            \State $\tilde{n} \gets \PurifiedApproxDiscreteLaplaceMechanism{\hat{n}_k}{\eps_k}{\beta_k}{ t }$
            \Comment{$\tilde{n} \approx \clamp{t + \DiscreteLapNoise(e^{-\eps_k})}{0}{\hat{n}_k}$}
            \If{$\tilde{n} < \frac{1}{2} \cdot \hat{n}_k$}
                \State Truncate $\DataSet$ to its first $\hat{n}_k$ entries if $n > \hat{n}_k$
                \State Construct $\hist$ based on $\DataSet$
                \State \Return $\ourAlgo(d, \hat{n}_k, k, a_\eps, b_\eps, a_\gamma, b_\gamma, \hist)$
                \Comment{\cref{algo:tail-item-padding}}
            \EndIf
        \EndFor
    \end{algorithmic}
\end{algorithm}

\subsubsection*{\bf Construction.}
First, suppose we have an estimate $\hat{n}$ such that $n \le \hat{n} \le c \cdot n$ for some constant $c$.
We run \cref{algo:tail-item-padding} with $n$ replaced by $\hat{n}$, keeping all other parameters fixed.
This modifies the algorithm in two ways:
(1) noise is now sampled over the range $\zeroton{\hat{n}}$, which preserves the original privacy guarantees and tail bounds; and
(2) more items are sampled in line~\ref{line:blanket-sampling}, which intuitively enhances the privacy of the selected item set $I$.
This intuition is formally supported by revisiting the proof of \cref{thm:private-sparse-histogram-formal}, presented in \cref{subsec:proof-private-sparse-histogram}.
Although replacing $n$ with $\hat{n}$ increases the running time, the asymptotic complexity remains the same, since $\hat{n} \in O(n)$.

Second, \citet{ratliff2025securing} show that, for given $\beta_1 \in (0,1)$, \cref{alg:private-upper-bound-search} is an $\eps_1$-DP time-oblivious algorithm that computes an estimate $\hat{n}$ satisfying $n \le \hat{n} \le c \cdot n$ with probability at least $1 - \beta_1$.
The algorithm constructs a sequence of candidate values $\hat{n}_k$, starting from a small initial guess and approximately doubling at each step, until it finds some $\hat{n}_k \ge n$. 
At each iteration, it spends privacy budget $\varepsilon_1 / 2^k$ to compare $\hat{n}_k$ to $n$: it proceeds with probability $1 - \beta_1 / 2^k$ if $n \ge \hat{n}_k$, and halts with the same probability if $n \le \hat{n}_k / 2$. 
The privacy budget and failure probability are geometrically decayed to ensure that their total sums remain bounded by the overall privacy and failure parameters.

There is one caveat: the iteration may terminate with some $\hat{n}_k < n$, though this occurs with low probability. 
If we then run \cref{algo:tail-item-padding} on a histogram constructed from the full dataset $\DataSet$, using $\hat{n}_k$ in place of $n = \card{\DataSet}$, the algorithm may fail. 
In particular, the privacy blanket sampling step (\cref{algo:tail-item-padding}, line~\ref{line:blanket-sampling}) requires $\hat{n}_k + k \ge \card{I_1}$, where $\card{I_1}$ can be as large as $n$. 
To address this, \cref{alg:private-upper-bound-search} truncates $\DataSet$ to its first $\hat{n}_k$ elements before constructing $\hist$ and invoking \cref{algo:tail-item-padding}. 
This truncation preserves the sensitivity of neighboring datasets and thus maintains privacy, at the cost of a small utility loss (with low probability).

\subsubsection*{\bf Privacy, Utility Guarantees, and Running Time.}
The joint output/time privacy guarantees and the utility bounds follow directly from \cref{prop:alias-method-exact} and the properties of \cref{algo:tail-item-padding}, as formalized in \cref{thm:private-sparse-histogram-formal}.

Analyzing the running time of \cref{alg:private-upper-bound-search} in the word-RAM model requires more care. 
When $k \in O(\log n)$, the input parameters and internal variables such as $\varepsilon_k$, $\beta_k$, and $\hat{n}_k$ all fit within a constant number of words. 
If $\ourAlgo$ (\cref{algo:tail-item-padding}) is not invoked, each iteration runs in $\tilde{O}(1 / \varepsilon)$ deterministic time, and there are at most $O(\log n)$ such iterations.

When the algorithm terminates, it satisfies $\Pr[n \le \hat{n}_k \le c \cdot n] \ge 1 - \beta_1$ for some constant $c$. 
Thus, with probability at least $1 - \beta_1$, the cost of invoking $\ourAlgo$ is 
$O(n)$.

For larger $k \in \N$ such that $\hat{n}_k > c \cdot n$, the variables $\varepsilon_k$, $\beta_k$, and $\hat{n}_k$ may exceed word size and require $O(k)$ bits to represent. 
This increases the cost of both $\PurifiedApproxDiscreteLaplaceMechanism{\hat{n}_k}{\varepsilon_k}{\beta_k}{t}$ and $\ourAlgo$ by a factor of $O(k)$. 
However, since the probability of reaching iteration $k$ decays faster than geometrically once $\hat{n}_k > c \cdot n$, the overall expected running time remains bounded by 
$O(n)$.

\newpage
\section{Related Work}
\label{sec:related work}

In this section, we review related work. 
We begin with possible side-channel attacks and defenses, then discuss noise sampling algorithms for differential privacy, and finally cover differentially private frequency estimation algorithms.

\subsection{Side-Channel Attacks}
While differentially private mechanisms offer strong theoretical guarantees, practical implementations can introduce side channels—such as floating-point rounding artifacts and timing variability—that can be exploited to compromise privacy.

\subsubsection*{\bf Floating Point Attack}
\citet{Mironov12} was the first to expose the vulnerability of implementing the Laplace mechanism using double-precision floating-point numbers.  
He observed that certain floating-point values cannot be generated due to the finite precision and rounding effects inherent in floating-point arithmetic, which can make it possible to distinguish between neighboring inputs and thereby break the guarantees of differential privacy.  
\citet{JinMRO22} later showed that similar attacks are also possible against the Gaussian mechanism.
One solution to these vulnerabilities is to adopt discrete versions of the Laplace and Gaussian mechanisms \citep{GhoshRS09, Canonne0S20}, which avoid the pitfalls of floating-point arithmetic.

\subsubsection*{\bf Timing Attack}
\citet{JinMRO22} further demonstrated that even discrete mechanisms designed to avoid floating-point issues---such as those by \citet{Canonne0S20} and \citet{google_noise}---can still leak information through timing side channels.
These mechanisms typically employ \emph{geometric sampling}, which when directly simulated, repeatedly performs biased coin tosses until the first ``head'' occurs.  
This introduces a positive correlation between the sampled value and the algorithm's running time, potentially leaking information.

\paragraph{Time-Oblivious Sampling.}
\citet*{DovDNT23} systematically study noise sampling algorithms resilient to timing-attack, and their implications for DP mechanisms. 
An algorithm is defined as time-oblivious if its output distribution and running time distribution are independent. 
Their key findings include:
1) a discrete distribution admits a time-oblivious sampling if and only if it has a finite support and a rational probability mass function, for which they provide sampling algorithms with optimal number of unbiased random bits;
and 2) such a distribution admits a worst-case time complexity sampling algorithm if the least common multiple of the denominators of its rational probabilities is a power of $2$.

As a time-oblivious sampling algorithm can require significantly more random bits than the classical Knuth-Yao sampler \citep{Knuth1976TheCO}, \citet{DovDNT23} introduce a relaxed notion of $\PAREN{\eps, \delta}$ time-oblivious sampler (see \cref{def:approx-time-oblivious-sampler}).
They design both $\eps$-pure and $\PAREN{0, \delta}$-approximate time-oblivious samplers that use logarithmic number of random bits but exponential space.

\paragraph{DP Time-Oblivious Algorithm.}
Extending their notion of $(\eps, \delta)$ time-oblivious sampler, \citet{DovDNT23} defines DP time-oblivious mechanisms (see \cref{def:joint-output-timing-dp}).
They show that any pure DP time-oblivious  mechanism over infinitely many datasets gives an irrelevant output with some constant probability.
Nevertheless, they demonstrate that a pure DP mechanism can be transformed into a time-oblivious with nearly the same privacy and utility guarantees, although the efficiency of this transformation remains an open problem due to its reliance on manipulating countably many probability values.

A similar definition was proposed by \citet{RatliffV24} under the name \emph{$(\eps, \delta)$-Joint Output/Timing Privacy}. 
Their formulation is more comprehensive, as it explicitly models the computational environment as part of the algorithm’s input, thereby capturing its influence on execution time. 
They also introduce a general framework that composes timing-stable programs with randomized delays to enforce timing privacy. 
However, this framework guarantees only approximate timing differential privacy.

In follow-up work, \citet{ratliff2025securing} present a new framework for converting pure, time-oblivious differentially private algorithms from the \emph{upper-bounded setting} (see the discussion following \cref{def: Differential Privacy}) to the \emph{unbounded setting}, under the add/remove neighboring relation and in the RAM model.

\subsection{Noise Generation}

This subsection reviews methods for generating random noise used in differentially private algorithms,  
with a focus on Bernoulli, discrete Laplace, and discrete Gaussian distributions.  
Many of these methods were developed in the MPC setting,  
and when simulated in the central model, they can prevent timing attacks.

\citet{DworkKMMN06} presented a distributed protocol for generating shares of approximate Gaussian or approximate geometric random variables, secure against malicious participants. 
In particular, they proved that each bit in the binary representation of a geometric random variable (over an interval whose length is a power of $2$) can be generated independently according to distinct Bernoulli distributions, and they provided closed-form formulas for these distributions.

Thus, the problem of efficiently generating geometric random variables reduces to that of generating Bernoulli random variables. 
\citet{DworkKMMN06} further studied efficient protocols for generating Bernoulli random variables $\Bernoulli{p}$ for some $p \in \paren{0, 1}$. 
They first observed that such a random variable can be generated using at most $2$ unbiased random bits in expectation.
Second, they presented deterministic-time protocols that, with high probability, generate a batch of $m$ biased bits whose statistical distance to $\Bernoulli{p}$ is at most $2^{-\ell}$. 
Their protocols are described as circuits:
\begin{enumerate}[leftmargin=0.6cm]
    \item The first has depth $\Theta \PAREN{ \ln^2 \paren{m \cdot \ell} }$, gate count $\Theta \PAREN{ m \cdot \ell \cdot \PAREN{ \ell + \ln m } \cdot \ln m }$, and input size $\Theta \PAREN{ m }$.
    \item The second has depth $\Theta \PAREN{ \ln m }$, gate count $\Theta \PAREN{ m^2 \cdot \ell }$, and input size $\Theta \PAREN{ m }$.
    \item The third has depth $\Theta \PAREN{ \ln \paren{ m + \ell } }$, gate count $\Theta \PAREN{ m \ell \cdot \ln \PAREN{ m + \ell } }$, and input size $\Theta \PAREN{ m  \ln \paren{ m + \ell } }$.
\end{enumerate}
When directly simulating these circuits on a central server, they require $\Omega \PAREN{ m \cdot \ell \cdot \ln \PAREN{ m + \ell } }$ running time.

\citet{ChampionSU19} design an MPC protocol that, given $p \in \PAREN{0, 1}$, samples $m$ random variables whose total variation distance to $m$ i.i.d. Bernoulli random variables $\Bernoulli{p}$ is at most $\ell$. The protocol achieves this with an amortized communication and computation cost of $O \PAREN{ \ln \PAREN{ \ell + \ln m } }$. 
A key ingredient in their approach is the circuit construction for stacks proposed by \citet{ZahurE13}.

\citet{Canonne0S20} show how to generate random variables that exactly follow discrete Gaussian or discrete Laplacian distributions under the word-RAM model using rejection sampling. 
Their paper also includes a subroutine for sampling an exact Bernoulli random variable $\Bernoulli{\exp \PAREN{ -\gamma }}$ for some positive rational number $\gamma \in \Q_+$.
Their algorithms require an unlimited amount of memory,  
run in $O(1)$ expected time,  
but could be vulnerable to potential timing attacks \citep{JinMRO22}.

\citet{KnottVHSIM21} present a software framework for secure MPC primitives in machine learning, which includes sampling algorithms for the Bernoulli, continuous Laplace, and Gaussian distributions. 
They represent floating-point values using fixed-point encoding with a length of $L$ bits. 
The framework also provides algorithms for evaluating complex functions such as $e^x$, $\sin x$, and $1 / x$. 
However, the authors do not explicitly discuss the running time or the number of bits required for initialization and intermediate computations to achieve a specified final precision when evaluating these functions via their algorithms.

\citet{WeiYFCW23} provide MPC protocols that realize the Bernoulli sampling and approximate geometric random sampling algorithms from \citet{DworkKMMN06}. 
Using these building blocks, they construct an MPC protocol for generating approximate discrete Laplace variables, which they further employ to develop an MPC protocol that approximates the discrete Gaussian sampling algorithm of \citet{Canonne0S20}.

\citet{KellerM0TZ24} propose MPC protocols that approximate the discrete Laplace and discrete Gaussian sampling algorithms of \citet{Canonne0S20} in a more direct and streamlined manner.

\citet{franzese2025secure} present an MPC protocol for noise generation over a finite domain, based on the table lookup method. 
Given a probability distribution $\mu$ over a finite domain, their method constructs $O(\ell)$ tables, each of size 
$O \paren{ \card{ \CoreSupportSet{2^{-\ell}} } }$,
where 
$\CoreSupportSet{2^{-\ell}}$
denote a subset of minimum size such that $\mu \paren{ \CoreSupportSet{2^{-\ell}} } \ge 1 - 2^{-\ell}$.
These tables are then used to sample a random variable whose statistical distance to $\mu$ is at most $O \PAREN{ 2^{-\ell} }$, with a running time proportional to the number of tables accessed during sampling.

\subsection{Differentially Private Frequency Estimation}

This subsection reviews differentially private algorithms for estimating item frequencies, including frequency oracles and private histograms.  
A private histogram explicitly releases a noisy version of the entire frequency vector, while a frequency oracle is a data structure that allows on-demand, query-based access to noisy frequency estimates.  
Every private histogram implicitly defines a frequency oracle.

\subsubsection*{\bf Additional Error Measure}
To motivate frequency oracle, we first introduce another commonly considered error measure, the per-query accuracy guarantee.

\begin{definition}[$(\alpha, \beta)$-Accurate Estimator]
\label{def:accurate-estimator}
    Let $\alpha \in \R_+$ and $\beta \in [0, 1]$.  
    A random variable $X$ is an $(\alpha, \beta)$-accurate estimator of $t \in \R$ if  
    $\P{ \card{X - t} \ge \alpha } \le \beta$.
\end{definition}
A histogram $\noisyhist$ achieves $(\alpha, \beta)$ per-query accuracy if, for each $i \in [d]$, $\noisyhist{i}$ is an $(\alpha, \beta)$-accurate estimator of $\hist{i}$.  
In the sparse histogram setting, these two notions coincide.  
\citet{BalcerV19} show that, for any $\beta \in (0, 1/2]$, an $\eps$-DP histogram $\noisyhist$ with $\norm{\noisyhist}_0 \le n'$ cannot achieve $(\alpha, \beta)$ per-query accuracy for  
\(
    \alpha \in \Omega \bigparen{ \min \set{ \frac{1}{\eps} \ln \frac{d}{n' \cdot \beta}, \, n } },
\)
which matches the simultaneous accuracy upper bounds obtained by our algorithm in \cref{sec:dp-sparse-histogram}.

There exist other private succinct representations of $\hist$ beyond sparse histograms, known as \emph{frequency oracles}, which can bypass the per-query accuracy lower bounds discussed below.

\subsubsection*{\bf Frequency Oracle.}
Let $\hist \in \IntSet{0}{n}^d$ be a histogram with at most $n$ non-zero entries. 
A \emph{frequency oracle} for $\hist$ is a data structure that, given $i \in [d]$, returns an estimate of $\hist{i}$.

\citet{BalcerV19} propose an $\eps$-DP frequency oracle using $O\PAREN{n \cdot \ln d}$ bits of space.  
It provides expected per-query error $O\PAREN{1/\eps}$, query time $\tilde{O}\PAREN{n/\eps}$, and expected simultaneous error $O\PAREN{(1/\eps) \cdot \ln d}$ over all $i \in [d]$.

\citet{LebedaAP22} give an $\PAREN{\eps, \delta}$-DP frequency oracle using $O\PAREN{n \cdot \ln \PAREN{d + n}}$ bits of space.  
For $\delta = 0$, it achieves expected per-query error $O\PAREN{1/\eps}$, query time $O\PAREN{\ln d}$, and expected simultaneous error $O\PAREN{(1/\eps) \cdot \ln d}$.  
For $\delta > 0$, the per-query error remains $O\PAREN{1/\eps}$, the query time becomes $O\PAREN{\ln \PAREN{1/\delta}}$, and the simultaneous error is $O\PAREN{(1/\eps) \cdot \ln \PAREN{1/\delta}}$.

\citet{LolckP24} further reduce the space to $O\PAREN{n \cdot \ln n}$ bits and the query time to $O\PAREN{\ln \ln d}$ in the $\delta = 0$ setting, using techniques from error-correcting codes.

As discussed, compared to sparse frequency histograms, these representations typically incur higher computational cost when estimating $\hist{i}$ for a given $i \in [d]$. 
Moreover, they are insufficient for solving the following commonly studied variant of the heavy hitters problem~\citep{BNS19,DBLP:conf/aistats/WuW22}: privately identifying all $i \in [d]$ such that $\hist{i} \ge \Delta$, for $\Delta$ as small as possible.
Addressing this task with a frequency oracle directly requires querying all $d$ domain elements, incurring $O(d)$ query overhead.
Thus, this line of work is orthogonal to our focus on sparse histograms.

\subsubsection*{\bf Private Histogram.}
A \emph{private histogram} is a data structure that releases a noisy version of the entire frequency vector $\hist$.
\citet{DworkMNS06} introduced the $\eps$-DP Laplace mechanism: it publishes a privatized histogram $\noisyhist$ of $\hist$ by adding continuous Laplacian noise $\LapNoise{\exp \paren{- \eps / 2}}$ to each entry independently, where $\LapNoise{\exp \paren{- \eps / 2}}$ is supported over $\R$ with density function $p(x)$ satisfying $p(x) \propto \exp \PAREN{ - \eps \cdot \card{x} / 2}, \forall x \in \R$.
The expected error is $O \PAREN{1 / \eps}$ for a single entry in $\noisyhist$, and the expected maximum error is $O \PAREN{1 / \eps \cdot \ln d}$ over all entries. 
It is known that these errors are asymptotically optimal \citep{HardtT10, BeimelBKN14}.
However, this mechanism has running time $O(d)$ and assumes real arithmetic.

To avoid real arithmetic, \citet{GhoshRS09} proposed replacing the continuous Laplace noise $\LapNoise{\exp \paren{- \eps / 2}}$ with a discrete variant, denoted $\DiscreteLapNoise{\exp \paren{- \eps / 2}}$.  
This distribution, referred to as the two-sided geometric distribution in the original paper, has probability mass function $P(z) \propto \exp \PAREN{- \eps \cdot \card{z} / 2}$ for all $z \in \Z$.  
The discrete variant preserves the optimal asymptotic error bounds of the Laplace mechanism.  
However, sampling from $\DiscreteLapNoise{\exp \paren{- \eps / 2}}$ requires unbounded memory access and has running time bounded only in expectation, exposing it to potential timing attacks.

To avoid the $O(d)$ running time, the \textit{stability-based histogram algorithm}, first proposed by \citet{KorolovaKMN09} and later presented by \citet{BunNS19}, adds $\DiscreteLapNoise{\exp \paren{- \eps / 2}}$ noise only to the nonzero entries of $\hist$ and releases only those entries whose noisy counts exceed a threshold $\tau \in \Theta \paren{ \frac{1}{\eps} \cdot \ln { \frac{1}{\delta} } }$.  
This method satisfies $(\eps, \delta)$-differential privacy, achieves a strict running time of $\Tilde{O}(n)$, and incurs an expected error of $O \paren{ \frac{1}{\eps} }$ for each entry whose true count is at least $\Omega \PAREN{ \frac{1}{\eps} \cdot \ln { \frac{1}{\delta} } }$.  
The expected maximum error over all entries is $O \PAREN{ \frac{1}{\eps} \cdot \ln { \frac{1}{\delta} } }$.

\citet{CormodePST12} present an $\eps$-DP variant of the stability-based histogram that adds $\DiscreteLapNoise{e^{- \eps / 2}}$ to all entries of $\hist$ and then releases only those whose noisy counts exceed a threshold $\tau \in \Theta \PAREN{ \frac{1}{\eps} \cdot \ln d }$.  
This variant achieves an expected error of $O \PAREN{ \frac{1}{\eps} }$ for each entry whose true count is at least $\Omega \PAREN{ \frac{1}{\eps} \cdot \ln d }$, and the expected maximum error over all entries is $O \PAREN{ \frac{1}{\eps} \cdot \ln d }$.  
However, a naive implementation of this method still requires $O(d)$ running time.
The key observation is that, for entries whose true count is zero, the probability that their noisy count exceeds the threshold $\tau$ is $p_\tau \in \Tilde{O} \PAREN{ \frac{1}{d} }$.  
Therefore, the number of such false positive entries follows the binomial distribution $\Binomial{d - \norm{\hist}_0}{p_\tau}$, which has an expected value of $\Tilde{O}(1)$, where $\norm{\hist}_0$ is the number of nonzero entries in $\hist$.  
Recently, \citet{DBLP:journals/corr/QiuY25} presented a variant of \citep{CormodePST12} and explicitly discussed how to sample $\Binomial{d - \norm{\hist}_0}{p_\tau}$ in the Real-RAM model in $O(1)$ expected time.
However, it remains open whether one can sample from $\Binomial{d - \norm{\hist}_0}{p_\tau}$ in strict $\Tilde{O}(n)$ time in the Word-RAM model. 

To overcome the aforementioned issues, \citet{BalcerV19} proposed releasing the $n$ largest noisy counts, which ensures that the output size is always bounded by $n$.  
Their approach achieves the same asymptotic error as that of \citet{CormodePST12}.  
Instead of adding noise to all entries in $\hist$ and then reporting the top $n$, they design a more sophisticated algorithm that samples these noisy counts directly, achieving a strict running time of $\Tilde{O}(n^2)$.

\citet{LebedaT24} study the problem of constructing private sparse histograms in the streaming setting.  
They propose an $\PAREN{\eps, \delta}$-DP algorithm that operates in strict time and outputs a sparse histogram using $2 \cdot k$ words of space.  
The histogram includes all elements whose frequencies are at least $n / \PAREN{1 + k} + \Omega \PAREN{1 / \eps \cdot \ln \PAREN{1 / \delta}}$, along with their frequency estimates.  
By treating the frequencies of elements not included in the histogram as zero, the expected maximum frequency estimation error is bounded by $n / \PAREN{1 + k} + O \PAREN{1 / \eps \cdot \ln \PAREN{1 / \delta}}$.  
Their algorithm also has an $\eps$-DP variant achieving expected maximum estimation error $n / \PAREN{1 + k} + O \PAREN{1 / \eps \cdot \ln d}$, which reduces to $O \PAREN{1 / \eps \cdot \ln d}$ when $k = n$.  
However, their approach relies on existing techniques for releasing the top-$n$ noisy counts, whereas the best known algorithm with strict time and bounded memory for this problem has time complexity $\tilde{O}\PAREN{n^2}$ \citep{BalcerV19}.

\end{document}